\documentclass[11pt]{article}
\RequirePackage{fullpage}
\usepackage{graphicx} 
\usepackage{multirow}

\title{Switching Classes: Characterization and Computation}

\author{Dhanyamol Antony\thanks{Department of Computer Science and Automation, Indian Institute of Science, Bengaluru, India. {\tt dhanyamola@iisc.ac.in.}
}
  \and
  Yixin Cao\thanks{Department of Computing, Hong Kong Polytechnic University, Hong Kong, China. {\tt yixin.cao@polyu.edu.hk}. 
  }
  \and Sagartanu Pal\thanks{Department of Computer Science and Engineering, Indian Institute of Technology Dharwad, India. {\tt 183061001@iitdh.ac.in}.
  }
  \and R.B.\ Sandeep\thanks{Department of Computer Science and Engineering, Indian Institute of Technology Dharwad, India. {\tt sandeeprb@iitdh.ac.in}.
  }
}
\date{}
\usepackage{tikz,amsmath, bm, amsthm}
\usepackage{color}

\usetikzlibrary{decorations.pathreplacing}
\usetikzlibrary{positioning, calc,chains,fit,shapes}
\usepackage{enumerate}
\usepackage{booktabs,array}
\usepackage{cleveref}
\usepackage{tikz}
\usetikzlibrary{quotes}
\usepackage{graphics}
\usepackage[T1]{fontenc}
\usepackage{microtype}
\usepackage{comment}
\usepackage{thm-restate}

\usepackage{subcaption}
\usepackage{url}
\usepackage[utf8]{inputenc}
\usepackage[T1]{fontenc}
\usepackage[normalem]{ulem}

\newtheorem{theorem}{Theorem}[section]
\newtheorem{lemma}[theorem]{Lemma}
\newtheorem{observation}[theorem]{Observation}
\newtheorem{corollary}[theorem]{Corollary}
\newtheorem{proposition}[theorem]{Proposition}

\newtheorem{construction}{Construction}

\newcommand{\pname}[1]{\textnormal{\textsc{#1}}}

\newcounter{rowcntra}[table]

\AtBeginEnvironment{tabular}{\setcounter{rowcntra}{0}}

\newcommand{\NPC}{NP-Complete}
\newcommand{\TRUE}{TRUE}

\newcommand{\YES}{yes}
\newcommand{\NO}{no}

\newcommand{\SWTF}[1]{\pname{Switching-to-\ensuremath{\mathcal{F}(#1)}}}
\newcommand{\LW}[1]{\ensuremath{\mathcal{L}(#1)}}
\newcommand{\UP}[1]{\ensuremath{\mathcal{U}(#1)}}
\newcommand{\SW}[1]{\ensuremath{\mathcal{S}(#1)}}

\newcommand{\TSAT}{\pname{3-SAT}}

\newcommand{\FiSATM}{\pname{Monotone NAE 5-SAT}}
\newcommand{\KSATM}{\pname{Monotone NAE $k$-SAT}}
\newcommand{\TSATM}{\pname{Monotone NAE 3-SAT}}
\newcommand{\KaSATM}{\pname{Monotone NAE $(k-1)$-SAT}}

\begin{document}

\maketitle



\begin{abstract}
In a graph, the switching operation reverses adjacencies between a subset of vertices and the others. 
For a hereditary graph class $\mathcal{G}$, we are concerned with the maximum subclass and the minimum superclass of $\mathcal{G}$ that are closed under switching.
We characterize the maximum subclass for many important classes $\mathcal{G}$, and prove that it is finite when $\mathcal{G}$ is minor-closed and omits at least one graph.
For several graph classes, we develop polynomial-time algorithms to recognize the minimum superclass. We also show that the recognition of the superclass is NP-complete for $H$-free graphs when $H$ is a sufficiently long path or cycle, and it cannot be solved in subexponential time assuming the Exponential Time Hypothesis.

\end{abstract}

\section{Introduction}
\label{sec:introduction}

In a graph $G$, the operation of \emph{switching} a subset $A$ of vertices is to reverse the adjacencies between $A$ and $V(G)\setminus A$.  Two vertices $x\in A$ and $y\in V(G)\setminus A$ are adjacent in the resulting graph if and only if they are not adjacent in $G$.
The switching operation, introduced by van Lint and Seidel \cite{lint-66-switching} (see more at \cite{seidel1974graphs,seidel1976asurvey, 2seidel1991two}),
is related to many other graph operations, most notably variations of graph complementation.
The \emph{complement} of a graph $G$ is a graph defined on the same vertex set of $G$, where a pair of distinct vertices are adjacent if and only if they are not adjacent in $G$.  The \emph{subgraph complementation} on a vertex set $A$ is to replace the subgraph induced by $A$ with its complement, while keeping the other part, including connections between $A$ and the outside, unchanged~\cite{DBLP:journals/algorithmica/AntonyGPSSS22}.
Switching $A$ is equivalent to taking the complement of the graph itself and the subgraphs induced by $A$ and $V(G)\setminus A$.
Indeed, the widely used \emph{bipartite complementation} operation of a bipartite graph is nothing but switching one part of the bipartition.
A special switching operation where $A$ consists of a single vertex is also well studied.
It is a nice exercise to show that switching $A$ is equivalent to switching the vertices in $A$ one by one.
This is somewhat related to the local complementation operation \cite{DBLP:journals/dm/EhrenfeuchtHR04}.

Two graphs are \emph{switching equivalent} if one can be obtained from the other by switching.
Colbourn and Corneil~\cite{DBLP:journals/dam/ColbournC80} proved that deciding whether two graphs are switching equivalent is polynomial-time equivalent to the graph isomorphism problem.
Another interesting topic is to focus on graphs from a hereditary graph class $\mathcal{G}$---a class is \emph{hereditary} if it is closed under taking induced subgraphs.
There are two natural questions in this direction.  Given a graph $G$, \begin{itemize}
\item whether $G$ can be switched to a graph in $\mathcal{G}$? and
\item whether all switching equivalent graphs of $G$ are in $\mathcal{G}$?
\end{itemize}
We use the \emph{upper $\mathcal{G}$ switching class} and the \emph{lower $\mathcal{G}$ switching class}, respectively, to denote the set of positive instances of these two problems.
\footnote{We are using switching classes in a different way from previous work.  In this work, classes mean graph classes, while in the literature, they mean equivalent classes.}
Since switching the empty set does not change the graph, the answer of the first question is yes for every graph in $\mathcal{G}$, while the answer of the second question can only be yes for a graph in $\mathcal{G}$.
Thus, the class $\mathcal{G}$ is sandwiched in between these two switching classes.
Note that the three classes collapse into one when $\mathcal{G}$ is closed under switching, e.g., complete bipartite graphs.

Both switching classes are also hereditary.
For the upper switching class, if a graph $G$ can be switched to a graph $H$ in $\mathcal{G}$, then any induced subgraph of $G$ can be switched to an induced subgraph of $H$, which is in $\mathcal{G}$ because $\mathcal{G}$ is hereditary.
For the lower switching class, recall that a hereditary graph class $\mathcal{G}$ can be characterized by a (not necessarily finite) set $\mathcal{F}$ of forbidden induced subgraphs.  A graph is in $\mathcal{G}$ if and only if it does not contain any forbidden induced subgraph.
If $G$ contains any induced subgraph that is switching equivalent to a graph in $\mathcal{F}$, then $G$ cannot be in the lower $\mathcal{G}$ switching class.
Thus, the forbidden induced subgraphs of the lower $\mathcal{G}$ switching class are precisely all the graphs that are switching equivalent to some graphs in $\mathcal{F}$.

Even when $\mathcal{G}$ has an infinite set of forbidden induced subgraphs, the lower $\mathcal{G}$ switching class may have very simple structures.
The list of forbidden induced subgraphs obtained as above is usually not minimal.
For example, Hertz~\cite{DBLP:journals/dam/Hertz99} showed that the lower perfect switching class has only four forbidden induced subgraphs, all switching equivalent to the five-cycle.
In the same spirit as Hertz~\cite{DBLP:journals/dam/Hertz99}, we characterize the lower $\mathcal{G}$ switching classes of a number of important graph classes. 

\begin{theorem}\label{thm:lower-class-1}
  The lower $\mathcal{G}$ switching class is characterized by a finite number of forbidden induced subgraphs when $\mathcal{G}$ is one of the following graph classes: weakly chordal, comparability, co-comparability, permutation, distance-hereditary, Meyniel, bipartite, chordal bipartite, complete multipartite, complete bipartite, chordal, strongly chordal, interval, proper interval, Ptolemaic, and block.
\end{theorem}
Indeed, since the forbidden induced subgraphs of 
threshold graphs are $2K_2, C_4$, and $P_4$~\cite{chvatal1973set}, by the arguments given above, the forbidden subgraphs of 
the lower threshold switching class are all graphs on four vertices (every graph on four vertices is switching equivalent to a graph in $\{2K_2, C_4, P_4\}$). This class, consisting of only graphs of order at most three, is finite.  
Also finite are lower switching classes of minor-closed graph classes that are nontrivial\footnote{We thank an anonymous reviewer for the bound in Theorem~\ref{thm:lower-class-2}, which improves the bound in a previous version of this manuscript.} (there exists at least one graph not in this class).

\begin{restatable}{theorem}{lowerclasstwo}
\label{thm:lower-class-2}
  Let $\mathcal{G}$ be a nontrivial minor-closed graph class, and let $p$ be the smallest order of a forbidden minor of $\mathcal{G}$. Then 
  $|V(G)| = O(p\sqrt{p})$, for graphs $G$ in lower 
  $\mathcal{G}$ switching class.
\end{restatable}

Theorems~\ref{thm:lower-class-1}
and~\ref{thm:lower-class-2} immediately imply polynomial-time and constant-time algorithms, respectively, for recognizing these lower switching classes, i.e., deciding whether a graph is in the class.
We remark that there are classes $\mathcal{G}$ such that the lower $\mathcal{G}$ switching class has an infinite number of forbidden induced subgraphs.

The upper $\mathcal{G}$ switching classes turn out to be more complicated.  
These classes are nontrivial
even for the class of $H$-free graphs for a fixed graph $H$.
Although $\mathcal{G}$ has only one forbidden induced subgraph, the number of forbidden induced subgraphs of the upper $\mathcal{G}$ switching class is usually infinite. Based on our current knowledge, exceptions do exist but are rare \cite{DBLP:journals/jgt/JelinkovaK14}.
Even so, for many graph classes $\mathcal{G}$, polynomial-time algorithms for recognizing the upper $\mathcal{G}$ switching class exist, e.g., bipartite graphs~\cite{DBLP:journals/fuin/HageHW03}.
Our understanding of this problem is very limited, even for classes defined by forbidding a single graph $H$.
For all graphs $H$ on at most three vertices, polynomial-time algorithms are known for recognizing the upper $H$-free switching class~\cite{ DBLP:journals/fuin/HageHW03,hayward1996recognizingp3,kratochvil1992computational}.
Of a graph $H$ on four vertices, the four-path \cite{DBLP:journals/dam/Hertz99} and the claw \cite{DBLP:journals/jgt/JelinkovaK14} have been settled.  We present a polynomial-time algorithm for paw-free graphs.
If two graphs $H_{1}$ and $H_{2}$ are complements to each other, then the recognition of the upper $H_{1}$-free switching class is polynomially equivalent to that of the upper $H_{2}$-free switching class.
Thus, the remaining cases on four vertices are the diamond, the cycle, and the complete graph.
We made attempt to them by solving the class of forbidding the four-cycle and its complement, which is known as pseudo-split graphs.

\begin{theorem}\label{thm:upper-class-algorithm}
  The upper $\mathcal{G}$ switching class can be recognized in polynomial time when $\mathcal{G}$ is one of the following graph classes: paw-free graphs, pseudo-split graphs, split graphs,
$\{K_{1,p}, \overline{K_{1,q}}\}$-free graphs, and bipartite chain graphs.
\end{theorem}

In Theorem~\ref{thm:upper-class-algorithm}, we want to highlight the algorithms for pseudo-split graphs and for split graphs.
We actually show a stronger result.  Any input graph $G$ has only a polynomial number of ways to be switched to a graph in these two classes, and we can enumerate them in polynomial time.  Thus, the algorithms can apply to hereditary subclasses of pseudo-split graphs, provided that these subclasses themselves can be recognized in polynomial time.  This is only possible when the lower switching classes of them are finite.  It is unknown whether the other direction also holds true.

Jel{\'{\i}}nkov{\'{a}} and Kratochv{\'{\i}}l~\cite{DBLP:journals/jgt/JelinkovaK14} 
found
graphs $H$ such that the upper $H$-free switching class is hard to recognize.  The smallest graph they found is on nine vertices.
More specifically, they showed that, for all $k \ge 3$, there is a graph of order $3 k$ with this property.
The graph is obtained from a three-vertex path by substituting one degree-one vertex with an independent set of $k$ vertices, and each of the other two vertices with a clique of $k$ vertices.
We show that the recognition of the upper $H$-free switching class is already hard when $H$ is a cycle on seven vertices or a path on ten vertices.
Our proofs can be adapted to longer ones.

\begin{theorem}\label{thm:hardness}
  Deciding whether a graph is switching equivalent to a $P_{10}$-free graph or a $C_{7}$-free graph is \NPC{}, and it cannot be solved in subexponential time (on $|V(G)|$) assuming the Exponential Time Hypothesis.
\end{theorem}
Since the problem admits a trivial $2^{|V(G)|}\cdot |V(G)|^{O(1)}$-time algorithm, by enumerating all subsets of $V(G)$, our bound in Theorem~\ref{thm:hardness} is asymptotically tight.
We conjecture that it is \NPC{} to decide whether a graph can be switched to an $H$-free graph when $H$ is a cycle or path of length six.

Theorem~\ref{thm:lower-class-1} and \ref{thm:lower-class-2} are proved in Section~\ref{sec:lower}, Theorem~\ref{thm:upper-class-algorithm} is proved in Section~\ref{sec:upper}, and Theorem~\ref{thm:hardness} is proved in Section~\ref{sec:hardness}.

\subsection*{Other related work} Jel{\'{\i}}nkov\'{a} et al.~\cite{DBLP:journals/dmtcs/JelinkovaSHK11} studied the parameterized complexity of the recognition problem of the upper switching classes.
Let us remark that there is also study on the upper switching classes for non-hereditary graph classes. For example, 
we can decide in polynomial time whether a graph can be switching equivalent to a Hamiltonian graph~\cite{DBLP:conf/tagt/EhrenfeuchtHHR98} or to an Eulerian graph~\cite{DBLP:journals/fuin/HageHW03}, but it is \NPC{} to decide whether a graph can be switching equivalent to a regular graph \cite{kratochvil2003complexity}.
Cameron~\cite{cameron1977cohomological} and Cheng and Wells Jr.~\cite{DBLP:journals/jct/ChengW86} generalized the switching operation to directed graphs.
Foucaud et al.~\cite{DBLP:journals/algorithmica/FoucaudHLMP22} studied switching operations in a different setting.

Seidel~\cite{seidel1976asurvey} showed that the size of a maximum set of switching inequivalent graphs on $n$ vertices is equivalent to the number of two-graphs of size $n$. 
This is further shown to be the same as the number Eulerian graphs on $n$ vertices~\cite{mallows1975two} and graphs on $2 n$ vertices admitting certain coloring~\cite{DBLP:journals/siamdm/OhYY13}.
Bodlaender and Hage~\cite{DBLP:journals/tcs/BodlaenderH12} showed that the switching operation does not change the cliquewidth of a graph too much, though it may change the treewidth significantly.  
The switching equivalence between graphs in certain classes can be decided in polynomial time.  For example, acyclic graphs because two forests are switching equivalent if and only if they are isomorphic~\cite{DBLP:journals/ejc/HageH98}.
In a complementary study, Hage and Harju~\cite{DBLP:journals/siamdm/HageH04} characterized graphs that cannot be switched to any forest.  They are either a small graph on at most nine vertices, or switching equivalent to a cycle.

From a graph $G$ on $n$ vertices, we can obtain $n$ graphs by switching each vertex, called the \emph{switching deck} of~$G$.
The \emph{switching reconstruction conjecture} of
Stanley~\cite{DBLP:journals/jct/Stanley85} asserts that for any $n > 4$, if two graphs on $n$ vertices have the same switching deck, they must be isomorphic.
The conjecture remains widely open, and we know that it holds on triangle-free graphs~\cite{DBLP:journals/jct/EllinghamR92}.
A similar question in digraph is also studied~\cite{DBLP:journals/jgt/BondyM11}.  

\section{Preliminaries}

All the graphs discussed in this paper are finite and simple.  The vertex set and edge set of a graph $G$ are denoted by, respectively, $V(G)$ and $E(G)$.  Let $n = |V(G)|$ and $m = |E(G)|$.
For a subset $U\subseteq V(G)$, we denote by $G[U]$ the subgraph of $G$ induced by $U$, and by $G - U$ the subgraph $G[V(G)\setminus U]$, which is shortened to $G - v$ when $U = \{v\}$.
The \emph{neighborhood} of a vertex $v$, denoted by $N_{G}(v)$, comprises vertices adjacent to $v$, i.e., $N_{G}(v) = \{ u \mid uv \in E(G) \}$, and the \emph{closed neighborhood} of $v$ is $N_{G}[v] = N_{G}(v) \cup \{ v \}$.
The \emph{closed neighborhood} and the \emph{neighborhood} of a set $X\subseteq V(G)$ of vertices are defined as $N_{G}[X] = \bigcup_{v \in X} N_{G}[v]$ and $N_{G}(X) =  N_{G}[X] \setminus X$, respectively.
We may drop the subscript if the graph is clear from the context.
We write $N(u, v)$ and $N[u, v]$ instead of $N(\{u, v\})$ and $N[\{u, v\}]$; i.e., we drop the braces when writing the neighborhood of a vertex set. Two vertex sets $X$ and $Y$ are \emph{complete} (resp., \emph{nonadjacent}) \emph{to} each other if all (resp., no) edges between $X$ and $Y$ are present.

A \emph{clique} is a set of pairwise adjacent vertices, and an \emph{independent set} is a set of vertices that are pairwise nonadjacent. 
For $\ell \ge 3$, we use $C_\ell$ to denote a cycle on $\ell$ vertices. An induced $C_\ell$, for any $\ell\geq 4$, in a graph is called an \emph{$\ell$-hole}. A hole of odd number of vertices is called an \emph{odd hole} and a hole of even number of vertices is called an \emph{even hole}.
A path on $\ell$ vertices is denoted by $P_\ell$, and
a complete graph on $\ell$ vertices is denoted by $K_\ell$.
A star graph on $\ell+1$ vertices is denoted by $K_{1,\ell}$. 

The \emph{disjoint union} of two graphs $G_1$ and $G_2$ is denoted by $G_1+G_2$. 
The \emph{complement graph} $\overline{G}$ of a graph $G$ is defined on the same vertex set $V(G)$, where a pair of distinct vertices $u$ and $v$ is adjacent in $\overline{G}$ if and only if $u v \not\in E(G)$.
By $\overline{\mathcal{G}}$, we denote the set of complements of graphs in $\mathcal{G}$.
By $\mathcal{G}^c$, we denote the set of graphs not in $\mathcal{G}$.
The switching of a vertex subset $A$ of a graph $G$ is denoted by $S(G, A)$.  It has the same vertex set as $G$ and its edge set is
\[
E(G[A])\cup  E(G - A)\cup \{u v\mid u\in A, v\in V(G)\setminus A, uv\not\in E(G)\}.
\]
The following observations are immediately from the definition. 
The \emph{symmetric difference} of two sets is defined as $A \Delta B = (A\setminus B) \cup (B\setminus A)$.
\begin{proposition}[folklore]\label{prop:basic properties}
  Let $G$ be a graph, and $A, B \subseteq V(G)$. 
  \begin{itemize}
  \item $S(G, A) = S(G, (V(G)\setminus A))$.
  \item $S(S(G, A), A) = G$.
  \item $S(S(G, A), B) = S(S(G, B), A) = S(G, A\Delta B)$.
  \item  $\overline{S(G,A)}= S(\overline{G},A)$.
  \end{itemize}
\end{proposition}

Two graphs $G$ and $G'$ are called \emph{switching equivalent} if $S(G,A)=G'$ for some $A\subseteq V(G)$.
By Proposition~\ref{prop:basic properties}, switching is an equivalence relation.
For example, the eleven graphs of order 4 can be partitioned into the following three sets
\[
 \{C_4, \overline{K_{3} + K_{1}}, 4K_1\},
 \{2K_2, {K_{3} + K_{1}}, K_4\},
 \{P_4, {K_{2} + 2 K_{1}}, \overline{K_{2} + 2 K_{1}}, {P_{3} + K_{1}}, \overline{P_{3} + K_{1}}\}.
\]
Note that $\overline{K_{3} + K_{1}}$ is the claw, $\overline{P_{3} + K_{1}}$ is the paw, and $\overline{K_{2} + 2 K_{1}}$ is the diamond; see Figure~\ref{fig:small-graphs} and \ref{fig:s-c4-c5}a.
For a graph $G$, we use $\SW{G}$ to denote the set of non-isomorphic graphs that can be obtained from $G$ by switching. 
Figure~\ref{fig:s-c4-c5} illustrates $\SW{C_4}$
and $\SW{C_5}$.
For a set $\mathcal{G}$ of graphs, by $\SW{\mathcal{G}}$ we denote the union of
$\SW{G}$ for $G\in \mathcal{G}$.

A graph $G$ is called a \emph{split graph}
if the vertex set of $G$ can be partitioned in such a way that one is a 
clique and the other is an independent set. 
\emph{Split partitions} of a split graph refer to such (clique,
independent set) partitions.
An \emph{edgeless graph} is a graph without any edges.
A graph is \emph{complete bipartite} it its vertices can be partitioned into two sets $X, Y$ such that there is 
an edge between $x$ and $y$, for every $x\in X$ and for every $y\in Y$. It is denoted by $K_{|X|,|Y|}$.

Two graphs $G$ and $G'$ are \emph{isomorphic}, if there is a bijective function $f:V(G)\longrightarrow V(G')$
such that $uv$ is an edge in $G$ if and only if 
$f(u)f(v)$ is an edge in $G'$. 
For two graphs $G$ and $H$, we say that $G$ is 
$H$-free if there is no induced subgraph of $G$ isomorphic to $H$. In general, for two sets $\mathcal{G}$ and $\mathcal{H}$ of graphs, we say that
$\mathcal{G}$ is $\mathcal{H}$-free if $G$ is $H$-free
for every $G\in \mathcal{G}$ and for every $H\in \mathcal{H}$. By $\mathcal{F}(\mathcal{H})$, we denote the class of $\mathcal{H}$-free graphs. Note that $\mathcal{F}(\mathcal{H}\cup\mathcal{H}') = \mathcal{F}(\mathcal{H})\cap \mathcal{F}(\mathcal{H'})$.

For a graph property $\mathcal{G}$, the \emph{lower $\mathcal{G}$ switching class}, denoted by \LW{\mathcal{G}}, consists of all graphs $G$ with $\SW{G}\subseteq \mathcal{G}$.  
Note that every graph in \LW{\mathcal{G}}\ is also in $\mathcal{G}$.  Thus, \LW{\mathcal{G}}\ is the maximal subset $\mathcal{G}'$ of $\mathcal{G}$ such that $\SW{\mathcal{G}'} =\mathcal{G}'$.
The \emph{upper $\mathcal{G}$ switching class}, denoted by \UP{\mathcal{G}}, consists of all graphs $G$ with $\SW{G}\cap \mathcal{G}\neq \emptyset$. Clearly, every graph in $\mathcal{G}$ is in \UP{\mathcal{G}}. Therefore, the \UP{\mathcal{G}}\ 
is the minimal superset $\mathcal{G}'$ of $\mathcal{G}$
such that $\SW{\mathcal{G}'}=\mathcal{G}'$.
We note that $\UP{\mathcal{G}} = \SW{\mathcal{G}}$.
The following propositions are immediate from the definitions and Proposition~\ref{prop:basic properties}. 
\begin{proposition}
  For a hereditary graph class $\mathcal{G}$,
  both \LW{\mathcal{G}} and \UP{\mathcal{G}} are hereditary.
\end{proposition}

\begin{proposition}\label{pro:switching classes}
  Let $\mathcal{G}$ and $\mathcal{G}'$
  be graph classes. Then the following hold true.
  \begin{enumerate}
      \item\label{pro:switching classes:item:complement} $\overline{\LW{\mathcal{G}}}$ = \LW{\overline{\mathcal{G}}}, $\overline{\UP{\mathcal{G}}}$ = \UP{\overline{\mathcal{G}}}.
      \item\label{pro:switching classes:item:c} $(\LW{\mathcal{G}})^c = \UP{\mathcal{G}^c}$.
      \item\label{pro:switching classes:item:subset}   If $\mathcal{G}'\subseteq \mathcal{G}$,
      then \LW{\mathcal{G}'} $\subseteq$ \LW{\mathcal{G}} and \UP{\mathcal{G}'} $\subseteq$ \UP{\mathcal{G}}.
      \item\label{pro:switching classes:item:intersection} \LW{\mathcal{G}} $\cap$ \LW{\mathcal{G}'} = \LW{\mathcal{G}\cap\mathcal{G}'}.
      \item\label{pro:switching classes:item:union} \UP{\mathcal{G}} $\cup$ \UP{\mathcal{G}'} = \UP{\mathcal{G}\cup\mathcal{G}'}.
  \end{enumerate}
\end{proposition}

\begin{proposition}
    \label{pro:lower-h-free}
    For a set $\mathcal{H}$ of graphs, 
    \LW{\mathcal{F}(\mathcal{H})} = $\mathcal{F}(\UP{\mathcal{H})}$.
\end{proposition}

\begin{figure}[h]
\tikzstyle{empty vertex}  = [{circle, draw, fill = white, inner sep=1.pt}]
  \centering \small
\begin{subfigure}[b]{0.15\linewidth}
    \centering
    \begin{tikzpicture}[every node/.style={empty vertex},scale=.25]
      \node (s) at (0,4) {};
      \node (a1) at (-2,0) {};
      \node (b1) at (2,0) {};
      \node (c) at (0,2) {};
      \draw[] (c) -- (s);
      \draw[] (a1) -- (c) -- (b1) -- (a1);
    \end{tikzpicture}
    \caption{paw}\label{fig:paw}
  \end{subfigure}
  \,
  \begin{subfigure}[b]{0.15\linewidth}
    \centering
    \begin{tikzpicture}[every node/.style={empty vertex},scale=.25]
      \node (a1) at (-2, 0) {};
      \node (v) at (0, -2.5) {};
      \node (b1) at (2, 0) {};
      \node (c) at (0, 2.5) {};
      \draw (b1) -- (c) -- (a1) -- (v) -- (b1) -- (a1);
    \end{tikzpicture}
    \caption{diamond}\label{fig:diamond}
  \end{subfigure}
  \,
\begin{subfigure}[b]{0.15\linewidth}
    \centering
    \begin{tikzpicture}[every node/.style={empty vertex},scale=.4]
      \def\n{4}
      \def\radius{1.5}      
      \node (v0) at ({90 + 180 / \n}:\radius) {};
      \foreach \i in {1,..., \n} {
        \pgfmathsetmacro{\angle}{90 - (\i - .5) * (360 / \n)}
        \node (v\i) at (\angle:\radius) {};
        \draw let \n1 = {int(\i - 1)} in (v\n1) -- (v\i);
      }
      \node[empty vertex] (t) at (0, \radius * 1.5) {};
      \draw (v1) -- (t) -- (v4);
    \end{tikzpicture}
    \caption{house}\label{fig:house}
  \end{subfigure}
  \,
  \begin{subfigure}[b]{0.15\linewidth}
    \centering
    \begin{tikzpicture}[scale=.5]
      \foreach[count =\j] \i in {1, 2, 3} 
        \draw ({120*\i-30}:2) -- ({120*\i-30}:1) -- ({120*\i+90}:1);
        \foreach[count =\j] \i in {1, 2, 3} {
          \node[empty vertex] (u\i) at ({120*\i-30}:2) {};
          \node[empty vertex] (v\i) at ({120*\i-30}:1) {};
      }
    \end{tikzpicture}
    \caption{net}\label{fig:net}    
  \end{subfigure}
  \,
  \begin{subfigure}[b]{0.15\linewidth}
    \centering
    \begin{tikzpicture}[scale=.5]
      \foreach[count =\j] \i in {1, 2, 3} 
        \draw ({120*\i-90}:1) -- ({120*\i+30}:1) -- ({120*\i-30}:2) -- ({120*\i-90}:1);
        \foreach[count =\j] \i in {1, 2, 3} {
          \node[empty vertex] (u\i) at ({120*\i-30}:2) {};
        \node[empty vertex] (v\i) at ({120*\i-90}:1) {};
      }
    \end{tikzpicture}
    \caption{sun}\label{fig:sun}    
  \end{subfigure}
  \,
  \begin{subfigure}[b]{0.16\linewidth}
    \centering
    \begin{tikzpicture}[yscale=.67, every node/.style = empty vertex]
      \foreach \x in {0, 1} {
        \draw (\x, 0) -- (\x, 2);
      }
      \foreach \y in {0, 1, 2}
      \draw (0, \y) node {} --  (1, \y) node {};
    \end{tikzpicture}
    \caption{domino}\label{fig:domino}    
  \end{subfigure}
  \caption{small graphs.}
  \label{fig:small-graphs}
\end{figure}

\begin{figure}[h]
\tikzstyle{filled vertex}  = [{circle,draw=blue,fill=black!50,inner sep=1pt}]  \tikzstyle{empty vertex}  = [{circle, draw, fill = white, inner sep=1.pt}]
  \centering \small
  \begin{subfigure}[b]{0.38\linewidth}
    \centering
    \begin{tikzpicture}[scale=.6]
      \def\n{4}
      \def\radius{1.5}      
      \coordinate (v0) at ({90 + 180 / \n}:\radius) {};
      \foreach \i in {1,..., \n} {
        \pgfmathsetmacro{\angle}{90 - (\i - .5) * (360 / \n)}
        \coordinate (v\i) at (\angle:\radius) {};
        \draw let \n1 = {int(\i - 1)} in (v\n1) -- (v\i);
      }
      \foreach \i in {1,..., \n}
      \node[empty vertex] at (v\i) {};
    \end{tikzpicture}
    \;
    \begin{tikzpicture}[scale=.6]
      \def\n{4}
      \def\radius{1.5}      
      \node[filled vertex] (v4) at ({90 + 180 / \n}:\radius) {};
      \foreach \i in {1,..., 3} {
        \pgfmathsetmacro{\angle}{90 - (\i - .5) * (360 / \n)}
        \node[empty vertex] (v\i) at (\angle:\radius) {};
      }
      \foreach \i in {1, 3, 4} 
      \draw (v\i) -- (v2);
    \end{tikzpicture}
    \;
    \begin{tikzpicture}[scale=.6]
      \def\n{4}
      \def\radius{1.5}      
      \foreach \i in {1, 3}
      \node[filled vertex] at (v\i) {};
      \foreach \i in {2, 4}
      \node[empty vertex] at (v\i) {};
    \end{tikzpicture}
    \caption{$\mathcal{S}(C_{4}) = \{C_4, \mathrm{claw}, 4K_1\}$}
  \end{subfigure}
  \quad
  \begin{subfigure}[b]{0.56\linewidth}
    \centering
    \begin{tikzpicture}[scale=.5]
      \def\n{5}
      \def\radius{1.5}      
      \foreach \i in {0, 1,..., \n} {
        \pgfmathsetmacro{\angle}{90 - (\i) * (360 / \n)}
        \node[empty vertex] (v\i) at (\angle:\radius) {};
      }
      \foreach \i in {1,..., \n} {
        \pgfmathsetmacro{\p}{int(\i - 1)}
        \draw (v\i) -- (v\p);
      }
    \end{tikzpicture}
    \;
    \begin{tikzpicture}[scale=.5]
      \def\n{5}
      \def\radius{1.5}      
      \foreach \i in {0} {
        \pgfmathsetmacro{\angle}{90 - (\i) * (360 / \n)}
        \node[filled  vertex] (v\i) at (\angle:\radius) {};
      }
      \foreach \i in {1,..., 4} {
        \pgfmathsetmacro{\angle}{90 - (\i) * (360 / \n)}
        \node[empty vertex] (v\i) at (\angle:\radius) {};
      }
      \foreach \i in {2,..., 4} {
        \pgfmathsetmacro{\p}{int(\i - 1)}
        \draw (v\i) -- (v\p);
      }
      \draw (v2) -- (v0) -- (v3);
    \end{tikzpicture}
    \;
    \begin{tikzpicture}[scale=.5]
      \def\n{5}
      \def\radius{1.5}      
      \foreach \i in {1, 2, 3} {
        \pgfmathsetmacro{\angle}{90 - (\i) * (360 / \n)}
        \node[filled vertex] (v\i) at (\angle:\radius) {};
      }
      \foreach \i in {4,..., 6} {
        \pgfmathsetmacro{\angle}{90 - (\i) * (360 / \n)}
        \node[empty vertex] (v\i) at (\angle:\radius) {};
      }
      \foreach \i in {2,..., 6} {
        \pgfmathsetmacro{\p}{int(\i - 1)}
        \draw (v\i) -- (v\p);
      }
      \draw (v2) -- (v5) -- (v3);
    \end{tikzpicture}
    \;
    \begin{tikzpicture}[scale=.5]
      \def\n{5}
      \def\radius{1.5}      
      \foreach \i in {1, 4} {
        \pgfmathsetmacro{\angle}{90 - (\i) * (360 / \n)}
        \node[filled vertex] (v\i) at (\angle:\radius) {};
      }
      \foreach \i in {0, 2, 3, 5} {
        \pgfmathsetmacro{\angle}{90 - (\i) * (360 / \n)}
        \node[empty vertex] (v\i) at (\angle:\radius) {};
      }
      \foreach \i in {2,..., 4} {
        \pgfmathsetmacro{\p}{int(\i - 1)}
        \draw (v\i) -- (v\p);
      }
    \end{tikzpicture}
    \caption{$\mathcal{S}(C_{5}) = \{C_5, \mathrm{bull}, \mathrm{gem}, P_4+K_1\}$}
  \end{subfigure}
  \caption{Switching equivalent graphs of $C_{4}$ and $C_{5}$.  Switching the solid nodes (or the rest) results in the first graph in the list.}
\label{fig:s-c4-c5}
\end{figure}

 \section{Lower switching classes}
\label{sec:lower}

Every (odd) hole of length at least seven contains an induced $P_{4}+K_{1}$, and its complement contains an induced gem.  Both $P_{4}+K_{1}$ and the gem are in $\mathcal{S}(C_{5})$; see Figure~\ref{fig:s-c4-c5}b.  Thus, all the forbidden induced subgraphs of perfect graphs, namely, odd holes and their complements, boil down to $\mathcal{S}(C_{5})$, and the lower perfect switching class is equivalent to the lower $C_5$-free switching class~\cite{DBLP:journals/dam/Hertz99}.  In the same spirit, we characterized the lower $\mathcal{G}$ switching classes of a number of important graph classes listed in Figure~\ref{fig:main-subclasses}. The results are listed in Table~\ref{tab:lower}.

Since all these lower switching classes have finite characterizations, they can be recognized in polynomial time.
For the class of chordal graphs and several of its subclasses, we show a stronger structural characterization of their lower switching classes.  They have to be proper interval graphs with a very special structure.

We crucially use Lemma~\ref{lem:lower-subclass} which states that
if lower $\mathcal{G}'$ switching class is equivalent to lower $\mathcal{G}$ switching class, for any $\mathcal{G}\subseteq \mathcal{G}'$, then it is also equivalent to lower $\mathcal{G}''$ switching class, where $\mathcal{G}''$ is sandwiched between $\mathcal{G}'$ and $\mathcal{G}$. The proof is direct from Proposition~\ref{pro:switching classes}(\ref{pro:switching classes:item:subset}).

\begin{lemma}
    \label{lem:lower-subclass}
    Let $\mathcal{G}, \mathcal{G}'$ be classes of graphs
    such that $\mathcal{G}\subseteq \mathcal{G}'$.
    If \LW{\mathcal{G}'} $=$ \LW{\mathcal{G}}, then 
    \LW{\mathcal{G}''} $=$ \LW{\mathcal{G}}, for every graph class
    $\mathcal{G}''$ such that $\mathcal{G}\subseteq\mathcal{G}''\subseteq\mathcal{G}'$.
    In particular, the following is true.
    Let $\mathcal{H}, \mathcal{H}'$ be sets of graphs such that $\mathcal{H}'\subseteq\mathcal{H}$. 
    If \LW{\mathcal{F}(\mathcal{H})} $=$ \LW{\mathcal{F}(\mathcal{H}')}, then \LW{\mathcal{F}(\mathcal{H}'')} $=$ \LW{\mathcal{F}(\mathcal{H}')}, for every set $\mathcal{H}''$
    such that $\mathcal{H}'\subseteq\mathcal{H}''\subseteq\mathcal{H}$.
\end{lemma}

\begin{figure}[h]
  \tikzstyle{class} = [shape=rectangle, rounded corners, draw, align=center, top color=white, bottom color=blue!20]
  \tikzstyle{others} = [shape=rectangle, minimum height = 6mm, rounded corners, draw, fill=white, thin, align=center]
  \centering\small
  \begin{tikzpicture}[xscale=.8, yscale=.5]
    \node[others] (wchordal) at (8.0, 15) {weakly chordal};
    \node[others] (meyniel) at (12, 15) {Meyniel};
    \node[others] (chordal) at (8.5, 12.0) {chordal};
    \node[others] (dh) at (5, 12) {distance-hereditary};
    \node[others] (schordal) at (10.0, 9.0) {strongly chordal};
    \node[others] (ptolemaic) at (5.0, 9.0) {Ptolemaic};
    \node[others] (block) at (6, 6) {block};
    \node[others] (pig) at (8.0, 4) {proper interval};
    \node[others] (interval) at (8.5, 6.0) {interval};
    \node[others] (coco) at (18, 9.0) {co-comparability};
    \node[others] (bc) at (12.5, 12.0) {chordal bipartite};
    \node[others] (perm) at (14.5, 6.0) {permutation};
    \node[others] (comp) at (18.0, 15) {comparability};
    \node[others] (bp) at (15.5, 13.5) {bipartite};
    \node[others] (cbp) at (12.0, 4.0) {complete bipartite};
    \draw[thick, bend right] (interval) -- (schordal) -- (chordal) (block) -- (schordal);
    \draw[thick] (perm) -- (coco) -- (interval) -- (pig);
    \draw[thick] (chordal) -- (wchordal) -- (dh);
    \draw[thick] (wchordal) -- (perm) (wchordal) -- (bc);
    \draw[thick] (dh) -- (ptolemaic) -- (block) (chordal) -- (ptolemaic);
    \draw[thick] (dh) -- (meyniel) -- (chordal) (meyniel) -- (bc);
    \draw[thick] (bc) -- (bp) -- (comp); 
    \draw[thick] (perm) -- (comp); 
    \draw[thick] (bp) -- (meyniel); 
    \draw[thick] (dh) -- (cbp); 
    \draw[thick] (cbp) -- (bc); 
    \draw[thick] (cbp) -- (perm); 
  \end{tikzpicture}
  \caption{\small The Hasse diagram of graph classes studied in Section~\ref{sec:lower}.}
  \label{fig:main-subclasses}
\end{figure}
 
\subsection{Some simple characterizations}

\setlength\extrarowheight{2pt}
\begin{table}[]
    \centering
\begin{tabular}{ p{4.5cm}p{5.0cm}p{2.0cm} }
\noalign{\hrule height 1.5pt}
$\mathcal{G}$ & $\LW{\mathcal{G}}$ & By \\ 
\noalign{\hrule height 1.5pt}
weakly chordal, permutation & $\LW{\mathcal{F}(\{C_5, C_6, \overline{C_6}\})}$ & Lemma~\ref{weakly chordal classes: lm1},\ref{permutation switching class: lm1} \\ 
distance-hereditary & $\LW{\mathcal{F}(\{\mathrm{domino},\mathrm{house}, C_5, C_6\})}$ & Lemma~\ref{distance-hereditary classes: lm1}\\ 
comparability & $\LW{\mathcal{F}(\{C_5, \overline{C_6}\})}$ & Lemma~\ref{comparability switching classes: lm1} \\ 
co-comparability & $\LW{\mathcal{F}(\{C_5, {C_6}\})}$ & Lemma~\ref{permutation switching class: lm1} \\ 
Meyniel graphs & $\LW{\mathcal{F}(\{C_5, \mathrm{house}\})}$ & Lemma~\ref{lem:lower-meyniel} \\ 
complete bipartite, chordal bipartite, bipartite & complete bipartite & Lemma~\ref{lem:lower-bipartite-complete} \\ 
chordal, strongly chordal, interval, proper interval, Ptolemaic & $\mathcal{C}_0$ & Corollary~\ref{chordal switching classes: lm1} \\ 
block & $(+)$, $(+,0,+)$, $(1,1,1)$, and $(1,0,1,0,1)$ & Lemma~\ref{block switching class: lm1} \\ 
\noalign{\hrule height 1.5pt}
\end{tabular}
    \caption{Lower switching classes of various graph classes}
    \label{tab:lower}
\end{table}

To see a simple application of Lemma~\ref{lem:lower-subclass}, let $\mathcal{G}$ be the class of complete bipartite graphs and $\mathcal{G}'$ be the class of bipartite graphs. Since $K_3$ and $K_2+K_1$ are switching equivalents, and bipartite graphs are $K_3$-free, we obtain that lower bipartite switching class is $\{K_3,K_2+K_1\}$-free. Recall that $\{K_3,K_2+K_1\}$-free graphs are exactly the class of complete bipartite graphs. Further, switching a complete bipartite graph results in a complete bipartite graph. 
Therefore, lower $\mathcal{G}''$ switching class is equivalent to the class of complete bipartite graphs, 
where $\mathcal{G}''$ is a subclass of bipartite graphs and a superclass of complete bipartite graphs, such as bipartite graphs, complete bipartite graphs, and chordal bipartite graphs\footnote{Chordal bipartite graphs are bipartite graphs in which every cycle of length at least six has a chord.}.

\begin{lemma}
    \label{lem:lower-bipartite-complete}
    Let $\mathcal{G}$ be any subclass of bipartite graphs and any superclass of complete bipartite graphs. Then \LW{\mathcal{G}} is the class of complete bipartite graphs.  
\end{lemma}

For the rest of this subsection, let $\mathcal{H}$ be the set of all graphs 
having an induced subgraph isomorphic to at least one graph in \SW{C_5}. Let $\mathcal{H}'$ be $\{C_5\}$. A \emph{building} is obtained from a hole by adding an edge connecting two vertices of distance two; e.g., the house, see Figure~\ref{fig:small-graphs}. An \emph{odd building} is a building with odd number of vertices.

Lemma~\ref{lem:lower-subclass} implies
Corollary~\ref{cor:lower-sandwich}.

\begin{corollary}
    \label{cor:lower-sandwich}
    \LW{\mathcal{F}(\mathcal{H})} = \LW{\mathcal{F}(\mathcal{H}'')} = \LW{\mathcal{F}(\mathcal{H}')} = \LW{\mathcal{F}(C_5)}, for every
    $\mathcal{H}''$ such that $\mathcal{H}'\subseteq\mathcal{H}''\subseteq\mathcal{H}$.
\end{corollary}

\begin{observation}
    \label{obs:lower-h}
    $\mathcal{H}$ contains $C_5$, holes of length at least seven, complements of holes of length at least seven, and buildings of at least six vertices.
\end{observation}
\begin{proof}
    The set $\mathcal{H}$ contains $C_5$ by definition. It 
    contains $C_\ell$, for every $\ell\geq 7$, as $C_\ell$ has 
    an induced $P_4+K_1\in \SW{C_5}$. It also contains $\overline{C_\ell}$, for $\ell\geq 7$, as $\overline{C_\ell}$
    contains a gem which is in \SW{C_5}. It contains a building of at least six vertices, as such a graph contains a bull ($\in \SW{C_5}$, see Figure~\ref{fig:s-c4-c5}b) as an induced subgraph. 
\end{proof}

A graph is \emph{weakly chordal} if it does not contain any induced cycle of length at least five or its complement.
Let $\mathcal{H}''$ contain $C_5$, all holes of at least seven vertices, and all complements of holes of at least seven vertices.
Let $\mathcal{J}=\{C_6, \overline{C_6}\}$. Note that the set of forbidden induced subgraphs of weakly chordal
graphs is $\mathcal{H}''\cup \mathcal{J}$. 
By Observation~\ref{obs:lower-h}, $\mathcal{H}'\subseteq\mathcal{H}''\subseteq\mathcal{H}$.
By Proposition~\ref{pro:switching classes}(\ref{pro:switching classes:item:intersection}) and Corollary~\ref{cor:lower-sandwich}, we obtain that \LW{\mathcal{F}(\mathcal{H}''\cup \mathcal{J})} = \LW{\mathcal{F}(\mathcal{H}'')\cap\mathcal{F}(\mathcal{J})} = \LW{\mathcal{F}(\mathcal{H}'')} $\cap$ \LW{\mathcal{F}(\mathcal{J})} = \LW{\mathcal{F}(\mathcal{H'})} $\cap$ \LW{\mathcal{F}(\mathcal{J})} = \LW{\mathcal{F}(\mathcal{H'}\cup \mathcal{J})}. Thus we obtain Lemma~\ref{weakly chordal classes: lm1}. 

\begin{lemma}
  \label{weakly chordal classes: lm1}
  The lower weakly chordal switching class is equivalent to \LW{\mathcal{F}(\{C_5, C_6, \overline{C_6}\})}.
\end{lemma}

In any induced subgraph of a \emph{distance-hereditary graph} $G$, two vertices in the same component have the same distance as in $G$.
The forbidden induced subgraphs of distance-hereditary graphs are domino (Figure \ref{fig:domino}), gem, house ($\overline{P_5}$), and holes of length at least five. 
We set $\mathcal{H}''$ to be the set of all holes of length five or at least seven, and $\mathcal{J}=\{C_6, \mathrm{domino}, \mathrm{gem}, \mathrm{house}\}$. 
By Observation~\ref{obs:lower-h}, $\mathcal{H}'\subseteq\mathcal{H}''\subseteq\mathcal{H}$.
Then by Proposition~\ref{pro:switching classes}(\ref{pro:switching classes:item:intersection}) and Corollary~\ref{cor:lower-sandwich}, we obtain that \LW{\mathcal{F}(\mathcal{H}''\cup \mathcal{J})} $=$ \LW{\mathcal{F}(\mathcal{H}'')\cap \mathcal{F}(\mathcal{J})} $=$ \LW{\mathcal{F}(\mathcal{H}'')} $\cap$ \LW{\mathcal{F}(\mathcal{J})} $=$ \LW{\mathcal{F}(\mathcal{H}')} $\cap$ \LW{\mathcal{F}(\mathcal{J})} $=$ \LW{\mathcal{F}(\mathcal{H}'\cup\mathcal{J})} $=$ \LW{\mathcal{F}(\{C_5, C_6, \mathrm{domino}, \mathrm{gem}, \mathrm{house}\})}. Then by the fact that gem is in $\SW{C_5}$, we obtain Lemma~\ref{distance-hereditary classes: lm1}.

\begin{lemma}
  \label{distance-hereditary classes: lm1}
  The lower distance-hereditary switching class is equivalent to the \LW{\mathcal{F}(\{domino,house, C_5, C_6\})}.
\end{lemma}

A graph is a \emph{comparability graph} if its edges can be oriented in a transitive way, i.e., the existence of arcs $x y$ and $y z$ forces the existence of the arc $x z$ in the orientation.
Gallai~\cite{zbMATH03246066} listed all forbidden induced subgraphs of this class.  The list is long and hence not reproduced here.  The following summary is sufficient for our purpose.

\begin{proposition} (\cite{zbMATH03246066})
    \label{comparability switching classes: prop1}
    \begin{enumerate}[i)]
    \item     A comparability graph is $\{C_{5}, \overline{C_{6}}\}$-free.
    \item Every $\{\mathrm{bull}, \mathrm{gem}, \overline{C_{6}}, C_{2 k + 1}, k \ge 2\}$-free graph is a comparability graph.
    \end{enumerate}
\end{proposition}

Let $\mathcal{H}''$ be the set of odd holes (i.e., $\mathcal{H}''=\{C_{2k+1},k\geq\ 2\}$) and $\mathcal{J} = \{\mathrm{bull}, \mathrm{gem}, \overline{C_{6}}\}$. 
By Observation~\ref{obs:lower-h}, $\mathcal{H}'\subseteq\mathcal{H}''\subseteq\mathcal{H}$. 
Then by Corollary~\ref{cor:lower-sandwich}, we obtain that \LW{\mathcal{F}(\mathcal{H}'')} $=$ \LW{\mathcal{H}'}.
Then by Proposition~\ref{pro:switching classes}, we obtain that
\LW{\mathcal{F}(\mathcal{H''}\cup\mathcal{J})} $=$ \LW{\mathcal{F}(\mathcal{H''})\cap\mathcal{F}(\mathcal{J})} $=$ \LW{\mathcal{F}(\mathcal{H''})} $\cap$ \LW{\mathcal{F}(\mathcal{J})} $=$ \LW{\mathcal{F}(\mathcal{H'})} $\cap$ \LW{\mathcal{F}(\mathcal{J})} $=$ \LW{\mathcal{F}(\mathcal{H'}\cup\mathcal{J})} = \LW{\mathcal{F}(\{C_5, \mathrm{bull}, \mathrm{gem}, \overline{C_{6}}\})}. Then, since $\mathrm{bull}, \mathrm{gem}\in \SW{C_5}$, we obtain
that \LW{\mathcal{F}(\mathcal{H''}\cup\mathcal{J})} $=$ \LW{\mathcal{F}(C_5, \overline{C_{6}})}, and since the class of 
comparability graphs is a superclass of $\mathcal{F}(\mathcal{H}''\cup \mathcal{J})$ and a subclass of $\mathcal{F}(\mathcal{H}'\cup\mathcal{J})$, by Lemma~\ref{lem:lower-subclass}, we obtain 
Lemma~\ref{comparability switching classes: lm1}.  

\begin{lemma} \label{comparability switching classes: lm1}
  The lower comparability switching class is equivalent to \LW{\mathcal{F}(\{C_5, \overline{C_6}\})}.
\end{lemma}

The complement of a comparability graph is a \emph{co-comparability graph}.  
A graph $G$ is a \emph{permutation graph} if it is both a
comparability graph and a co-comparability graph \cite{dushnik-41-poset}. Then by Proposition~\ref{pro:switching classes}(\ref{pro:switching classes:item:complement}) and \ref{pro:switching classes}(\ref{pro:switching classes:item:intersection}), we obtain Lemma~\ref{permutation switching class: lm1}.

\begin{lemma}
  \label{permutation switching class: lm1}
  The lower co-comparability switching class is equivalent to \LW{\mathcal{F}(\{C_5, {C_6}\})}.
  The lower permutation switching class is equivalent to \LW{\mathcal{F}(\{C_5, {C_6}, \overline{C_6}\})}.
\end{lemma}

A graph is \emph{Meyniel} if every odd cycle that is not a triangle has at least two chords.
The forbidden induced subgraphs of Meyniel graphs are odd holes and odd buildings~\cite{meyniel1984115}. Then with a similar analysis, we obtain Lemma~\ref{lem:lower-meyniel}.
\begin{lemma}
  \label{lem:lower-meyniel}
  The lower Meyniel switching class is equivalent to the \LW{\mathcal{F}(\{C_5, \mathrm{house}\})}.
\end{lemma}

\subsection{Chordal graphs and subclasses}

\begin{figure}[h]
\tikzstyle{filled vertex}  = [{circle,draw=blue,fill=black!50,inner sep=1pt}]  \tikzstyle{empty vertex}  = [{circle, draw, fill = white, inner sep=1.pt}]
  \centering \small
    \begin{tikzpicture}[scale=.5]
      \def\n{6}
      \def\radius{1.5}      
      \coordinate (v0) at ({90 + 180 / \n}:\radius) {};
      \foreach \i in {1,..., \n} {
        \pgfmathsetmacro{\angle}{90 - (\i - .5) * (360 / \n)}
        \coordinate (v\i) at (\angle:\radius) {};
        \draw let \n1 = {int(\i - 1)} in (v\n1) -- (v\i);
      }
      \foreach \i in {1,..., \n}
      \node[empty vertex] at (v\i) {};
    \end{tikzpicture}
    \;
    \begin{tikzpicture}[scale=.5]
      \def\n{6}
      \def\radius{1.5}      
      \node[filled vertex] (v2) at ({90 - (2 - .5) * (360 / \n)}:\radius) {};
      \coordinate (v3) at ({90 - (3 - .5) * (360 / \n)}:\radius) {};
      \foreach \i in {4,..., 7} {
        \pgfmathsetmacro{\angle}{90 - (\i - .5) * (360 / \n)}
        \coordinate (v\i) at (\angle:\radius) {};
        \draw let \n1 = {int(\i - 1)} in (v\n1) -- (v\i);
      }
      \foreach \i in {1,3, 4, ..., \n} {
        \pgfmathsetmacro{\angle}{90 - (\i - .5) * (360 / \n)}
        \node[empty vertex] (v\i) at (\angle:\radius) {};
      }
      \foreach \i in {4, ..., \n} 
        \draw (v2) -- (v\i);
    \end{tikzpicture}
    \;
    \begin{tikzpicture}[scale=.5]
      \def\n{6}
      \def\radius{1.5}      
      \coordinate (v3) at ({90 - (3 - .5) * (360 / \n)}:\radius) {};
      \foreach \i in {4,..., 7} {
        \pgfmathsetmacro{\angle}{90 - (\i - .5) * (360 / \n)}
        \coordinate (v\i) at (\angle:\radius) {};
        \draw let \n1 = {int(\i - 1)} in (v\n1) -- (v\i);
      }
      \foreach \i in {2, 4, 5, \n} {
        \pgfmathsetmacro{\angle}{90 - (\i - .5) * (360 / \n)}
        \node[empty vertex] (v\i) at (\angle:\radius) {};
      }
      \foreach \i in {1, 3} {
        \pgfmathsetmacro{\angle}{90 - (\i - .5) * (360 / \n)}
        \node[filled vertex] (v\i) at (\angle:\radius) {};
        \draw (v5) -- (v\i);
      }
\end{tikzpicture}
    \;
    \begin{tikzpicture}[scale=.5]
      \def\n{6}
      \def\radius{1.5}
      \coordinate (v1) at ({90 - (1 - .5) * (360 / \n)}:\radius) {};
      \foreach \i in {2,..., 5} {
        \pgfmathsetmacro{\angle}{90 - (\i - .5) * (360 / \n)}
        \coordinate (v\i) at (\angle:\radius) {};
        \draw let \n1 = {int(\i - 1)} in (v\n1) -- (v\i);
      }
      \foreach \x in {-1, 1} 
        \foreach \i in {2, 3, 4} {
        \pgfmathsetmacro{\angle}{90 - (\i - .5) * (360 / \n)*\x}
        \draw ({90 - 30*\x}:\radius) -- (\angle:\radius);
      }
      \foreach \i in {2, ..., 5} {
        \pgfmathsetmacro{\angle}{90 - (\i - .5) * (360 / \n)}
        \node[empty vertex] (v\i) at (\angle:\radius) {};
      }
      \foreach \i in {1, 6} {
        \pgfmathsetmacro{\angle}{90 - (\i - .5) * (360 / \n)}
        \node[filled vertex] (v\i) at (\angle:\radius) {};
      }
      \draw (v1) -- (v6);
\end{tikzpicture}
    \;
    \begin{tikzpicture}[scale=.5]
      \def\n{6}
      \def\radius{1.5}      
      \foreach \i in {1, 3, 5} {
        \pgfmathsetmacro{\angle}{90 - (\i - .5) * (360 / \n)}
        \node[empty vertex] at (\angle:\radius) {};
      }
      \foreach \i in {2, 4, 6} {
        \pgfmathsetmacro{\angle}{90 - (\i - .5) * (360 / \n)}
        \node[filled vertex] at (\angle:\radius) {};
        \draw let \n1 = {int(\i - 1)} in (v\n1) -- (v\i);
      }
    \end{tikzpicture}
    \;
    \begin{tikzpicture}[scale=.5]
      \def\n{6}
      \def\radius{1.5}      
      \node[filled vertex] (v0) at ({90 - ( - .5) * (360 / \n)}:\radius) {};
      \foreach \i in {1,..., \n} {
        \pgfmathsetmacro{\angle}{90 - (\i - .5) * (360 / \n)}
        \node[empty vertex] (v\i) at (\angle:\radius) {};
        \draw let \n1 = {int(\i - 1)} in (v\n1) -- (v\i);
      }
      \foreach \i in {4, ..., \n} {
        \draw (v2) -- (v\i);
        \node[empty vertex] at (v\i) {};
      }
      \foreach \i in {1, ..., 3} {
        \draw (v5) -- (v\i);
        \node[filled vertex] at (v\i) {};
      }
    \end{tikzpicture}
  \caption{Switching equivalent graphs of $C_{6}$.   The set $A$ consists of all the empty nodes or all the solid nodes.}
  \label{fig:c6}
\end{figure} 

We start with showing that the lower $\{C_4 , C_5, C_6\}$-free switching class is a subclass of proper interval graphs and has very simple structures.
For the statement of the result and easy reference to graphs in $\mathcal{S}(C_{6})$, we need a handy notation.
Let $a_{1}$, $\ldots$, $a_{p}$ be $p$ nonnegative integers.
For $1 \le i\le p$, we substitute the $i$th vertex of a path on $p$ vertices with a clique of $a_{i}$ vertices.  We denote the resulting graph as $(a_1,a_2,\ldots,a_p)$.  For example, the paw and the diamond are $(1,1,2)$ and $(1,2,1)$, respectively, while the complement of the diamond can be represented as $(2,0,1,0,1)$.  We use ``$+$'' to denote an unspecified positive integer, and hence $(+)$ stands for all complete graphs.  Thus,
\begin{align}
\mathcal{S}(C_6) =& \{C_{6}, (1,1,2,1,1), (2,1,2,0,1), (1,2,2,1), (2,0,2,0,2), (2,2,2) \}.
            \label{eq:c6 switching class}            
\end{align}

Note that a sun and a net (see Figure~\ref{fig:small-graphs}) contains an induced bull (e.g., by removing a degree-one vertex from a net or removing a degree-four vertex from a sun), while any cycle on at least seven vertices contains an induced $P_{4} + K_{1}$.
An interval graph has at most $n$ maximal cliques, and they can be arranged in a linear manner such that each vertex appears in a consecutive sequence of them \cite{fulkerson-65-interval-graphs}.

\begin{lemma}\label{lem:c4 c5 c6 free}
  The lower $\{C_4 , C_5, C_6\}$-free switching class consists of graphs
  $(+)$, $(+,+,1)$, $(+,1,+)$, $(+,0,+)$, $(+,+,1,0,+)$, $(+,0,+,0,1)$, $(+,+,1,+)$, and $(+,+,1,+,+)$.
\end{lemma}
\begin{proof}
  Let $G$ be a graph in the lower $\{C_4 , C_5, C_6\}$-free switching class.
Since the claw is in $\mathcal{S}(C_4)$, while the bull and $P_{4}+K_{1}$ are both in $\mathcal{S}(C_5)$, the graph $G$ must be \{claw, bull, $P_{4}+K_{1}\}$-free.
  Since a net or sun contains an induced bull, and a hole longer than six contains an induced $P_{4}+K_{1}$, the graph $G$ does not contain an induced net, sun, or hole. 
  Thus, $G$ is a proper interval graph \cite{wegner1967eigenschaften}.
Let $\langle Q_{1}, Q_{2}, \ldots, Q_{\ell}\rangle$ be a clique path 
  of a component $G'$ of $G$.
  Note that for $i =1, \ldots, \ell - 1$, none of the following can be empty: $Q_{i}\cap Q_{i+1}$ (because $Q_{i}$ and $Q_{i+1}$ belong to the same component), $Q_{i}\setminus Q_{i+1}$, and $Q_{i+1}\setminus Q_{i}$ (because $Q_{i}$ and $Q_{i+1}$ are maximal cliques).

  We argue that for all $i =2, \ldots, \ell - 1$, both $Q_{i}\setminus (Q_{i-1}\cup Q_{i+1})$ and $Q_{i-1}\cap Q_{i+1}$ are empty.
We take an arbitrary vertex $x_{1}\in Q_{i-1}\setminus Q_{i}$ and an arbitrary vertex $x_{2}\in Q_{i+1}\setminus Q_{i}$.
  Note that $Q_{i}\setminus (Q_{i-1}\cup Q_{i+1})$ and $Q_{i-1}\cap Q_{i+1}$ are disjoint.
  If neither is empty, then there exists a claw induced by $x_{1}$, $x_{2}$, and any pair of vertices $x\in Q_{i}\setminus (Q_{i-1}\cup Q_{i+1})$ and $y\in Q_{i-1}\cap Q_{i+1}$.
  This is impossible since $G$ is claw-free.
  In the rest, at most one of $Q_{i}\setminus (Q_{i-1}\cup Q_{i+1})$ and $Q_{i-1}\cap Q_{i+1}$ can be nonempty.  As a result, $Q_{i-1}\cap Q_{i}\setminus Q_{i+1}\ne\emptyset$ because
  \begin{align*}
    Q_{i-1}\cap Q_{i}\setminus Q_{i+1} =& (Q_{i-1}\cap Q_{i})\setminus (Q_{i-1}\cap Q_{i+1})
    \\
    =& (Q_{i}\setminus Q_{i+1})\setminus (Q_{i}\setminus (Q_{i-1}\cup Q_{i+1})).
  \end{align*}
  If $Q_{i}\setminus (Q_{i-1}\cup Q_{i+1}) = \emptyset$, then $Q_{i-1}\cap Q_{i}\setminus Q_{i+1}  = Q_{i}\setminus Q_{i+1}\ne \emptyset$; if $Q_{i-1}\cap Q_{i+1} = \emptyset$, then $Q_{i-1}\cap Q_{i}\setminus Q_{i+1}  = Q_{i-1}\cap Q_{i}\ne \emptyset$.
  By symmetry, $Q_{i}\cap Q_{i+1}\setminus Q_{i-1}\ne\emptyset$.
We take an arbitrary vertex $x_{3}\in Q_{i-1}\cap Q_{i}\setminus Q_{i+1}$ and an arbitrary vertex $x_{4}\in Q_{i}\cap Q_{i+1}\setminus Q_{i-1}$.
  If there exists a vertex $x$ in $Q_{i}\setminus (Q_{i-1}\cup Q_{i+1})\ne\emptyset$, then $\{x, x_{1}, x_{2}, x_{3}, x_{4}\}$ induces a bull.
  If there exists a vertex $x$ in $Q_{i-1}\cap Q_{i+1}$, then $\{x, x_{1}, x_{2}, x_{3}, x_{4}\}$ induces a gem.
  Since both bull and gem are in $\mathcal{S}(C_{5})$, we have a contradiction.

  Thus, for each $i =2, \ldots, \ell - 1$, the set $Q_{i}$ can be partitioned into $Q_{i-1}\cap Q_{i}$ and $Q_{i}\cap Q_{i+1}$.
This component $G'$ is
  \[
    \left( |Q_{1}\setminus Q_{2}|, |Q_{1}\cap Q_{2}|, |Q_{2}\cap Q_{3}|, \ldots, |Q_{\ell-1}\cap Q_{\ell}|, |Q_{\ell}\setminus Q_{\ell-1}| \right).
  \]
  We note that $\ell \le 4$, and $\ell \le 2$ when $G$ is disconnected.
  If $\ell > 4$, then $G$ contains an induced $P_{4} + K_{1}$. We end with the same contradiction if $\ell > 2$ and there is another component.
  Since $\overline{K_{4}}\in \mathcal{S}(C_{4})$, there cannot be four independent vertices in $G$.
  \begin{itemize}
  \item If $\ell = 4$, then $G$ is $(+,+,1,+,+)$.  If $|Q_{2}\cap Q_{3}| > 1$, then $G$ contains an induced $(1,1,2,1,1)$, which is in $\mathcal{S}(C_{6})$.
  \item If $\ell = 3$, then $G$ is $(+,+,1,+)$.  If both $|Q_{1}\cap Q_{2}| > 1$ and $|Q_{2}\cap Q_{3}| > 1$, then $G$ contains an induced $(1,2,2,1)$, which is in $\mathcal{S}(C_{6})$.
  \item If $\ell = 2$, then $G$ is $(+,+,1,0,+)$, $(+,1,+)$, or $(+,+,1)$.
    If all of $|Q_{1}\setminus Q_{2}|$, $|Q_{1}\cap Q_{2}|$, and $|Q_{2}\setminus Q_{1}|$ are greater than one, then $G$ contains an induced $(2,2,2)$, which is in $\mathcal{S}(C_{6})$.
    If $G$ is connected, it is either $(+,1,+)$ or $(+,+,1)$.
    If $G$ is disconnected, then there is precisely one other component different from $G'$, and it has to be a clique (otherwise there is an induced $\overline{K_{4}}$).
    Since $(2,1,2,0,1)\in \mathcal{S}(C_{6})$, the only possibility is $(+,+,1,0,+)$.
  \item We are in one of the previous cases if $G$ has a non-clique component.  Otherwise, $G$ comprises at most three clique components.
    Note that $(2,0,2,0,2)$ is in $\mathcal{S}(C_{6})$.
Then $G$ is $(+)$, $(+,0,+)$, or $(+,0,+,0,1)$.
  \end{itemize}

  No graph of the form $(+)$, $(+,+,1)$, $(+,1,+)$, $(+,0,+)$, $(+,+,1,0,+)$, $(+,0,+,0,1)$, $(+,+,1,+)$, and $(+,+,1,+,+)$ contains a hole.
  On the other hand, the switching operation on such a graph always leads to a graph of one of these forms.
This completes the proof.
\end{proof}

Let $\mathcal{C}_{0}$ denote the lower $\{C_4 , C_5, C_6\}$-free switching class.
Since chordal graphs are  $\{C_4 , C_5, C_6\}$-free,
lower chordal switching class is a subclass of $\mathcal{C}_0$.
By Lemma~\ref{lem:c4 c5 c6 free}, $\mathcal{C}_0$ is a subclass of lower chordal switching class. Therefore, lower chordal switching class is equivalent to $\mathcal{C}_0$.
This same observation applies to subclasses of chordal graphs that contain all the graphs in $\mathcal{C}_{0}$ and by Lemma~\ref{lem:lower-subclass} to superclasses of chordal graphs which are $\{C_4, C_5, C_6\}$-free. 

\begin{corollary}\label{chordal switching classes: lm1}
  The following switching classes are all equivalent to $\mathcal{C}_{0}$: lower chordal switching class, lower strongly chordal switching class, lower interval switching class, lower proper interval switching class, and lower Ptolemaic switching class.
\end{corollary}
\begin{proof}
  Since chordal graphs, strongly chordal graphs, interval graphs, and proper interval graphs are all hole-free, all the lower switching classes are subclasses of $\mathcal{C}_{0}$ by Proposition~\ref{pro:switching classes}.
  On the other hand, by Lemma~\ref{lem:c4 c5 c6 free}, all the graphs in $\mathcal{C}_{0}$ are proper interval graphs.
  Thus, $\mathcal{C}_{0}$ is a subclass of proper interval switching graphs, hence also a subclass of the first three switching classes.
  Ptolemaic graphs are gem-free chordal graphs.  Since gem is in $\mathcal{S}(C_{5}$), the lower Ptolemaic switching class is also $\mathcal{C}_{0}$.
  Thus, they are all equal.
\end{proof}

A graph $G$ is a \emph{block graph} if every maximal biconnected subgraph is a clique.
Block graphs are precisely diamond-free chordal graphs \cite{DBLP:journals/jct/BandeltM86}.

\begin{lemma}
  \label{block switching class: lm1}
  The lower block switching class is equivalent to \LW{\mathcal{F}(\{C_4,\mathrm{diamond}\})}.
\end{lemma}
\begin{proof}
  Since a block graph is $\{C_{4}, \text{diamond}\}$-free, the first is a subclass of the second.
  Note that a diamond can be switched to a $P_{4}$.
  Since both $C_5$ and $C_6$ contains an induced $P_4$, \LW{\mathcal{F}(\{C_{4}, \text{diamond}\})} is a subclass of $\mathcal{C}_{0}$.
  Thus, every graph in the lower $\{C_{4}, \text{diamond}\}$-free switching class is a block graph by Corollary~\ref{chordal switching classes: lm1}. 
\end{proof}
   
The following can be obtained by checking the list in Lemma~\ref{lem:c4 c5 c6 free}.  We give a simple argument.
\begin{lemma}\label{lem:c4 diamond free}
  The lower $\{C_4 , \text{diamond}\}$-free switching class consists of graphs
  $(+)$, $(+,0,+)$, $(1,1,1)$, and $(1,0,1,0,1)$.
\end{lemma}
\begin{proof}
  Let $G$ be a graph in the lower $\{C_4 , \text{diamond}\}$-free switching class.
  It is obvious when the order of $G$ is at most three.
  In the sequel we assume that $V(G)\ge 4$, and we show that $G$ must be a graph $(+)$ or $(+,0,+)$.
  As we can see in Table~\ref{tbl:switching-equivalents}, every four-vertex graph containing an induced $P_3$ can be switched to a $C_{4}$ or diamond.  Thus, $G$ is $P_3$-free, i.e., a cluster.
  If $G$ has four or more connected components, then there is an induced $4K_1$, a switching equivalent of $C_4$.
  Suppose $G$ has exactly three connected components, then there is an induced copy of the complement of diamond, a switching equivalent of a diamond.
 This concludes the proof.
\end{proof}

A graph $G$ is a \emph{line graph} if there is a one-to-one mapping $\phi$ from $V(G)$ to the edge set of another graph $H$ such that $u v\in E(G)$ if and only if $\phi(u)$ and $\phi(v)$ share an end.
The class of line graphs has nine forbidden induced subgraphs~\cite{beineke1970characterizations}, two of which are switching equivalent to $C_{6}$, and one $C_{4}$.  Although $C_{5}$ is not forbidden, we show that a graph in the lower line switching class contains an induced $C_{5}$ if and only if it is a $C_{5}$.  Thus, this switching class consists of $\mathcal{S}(C_{5})$ and a subclass of $\mathcal{C}_{0}$.

\begin{lemma}
  \label{line switching class: lm1}
  The lower line switching class comprises of
(+), (1,1,1), (2,1,1), (1,2,1), (2,1,2), (+,0,+), (1,1,1,0,1), (2,1,1,0,1), (1,0,1,0,1), (2,0,1,0,1), (2,0,2,0,1), (1,1,1,1), (1,2,1,1), (1,1,1,1,1), and $\mathcal{S}(C_5)$. 
\end{lemma}
\begin{proof}
  Let $G$ be a graph in the lower line switching class.
  We first show that $G$ contains an induced $C_{5}$ if and only if it is a $C_{5}$.
  Suppose for contradiction that $G$ contains an induced $C_{5}$ of length 5 and another vertex $x$ not on this cycle.
  Let $H$ be the vertex set of this $C_{5}$, and $G' = G[H\cup \{x\}]$.
We may assume that $|N(x)\cap H| \ge 3$; otherwise, we consider $S(G, \{x\})$.
  Since a line graph is $W_5$-free, $|N(x)\cap H| < 5$.
The graph is also forbidden in line graphs when $x$ has three consecutive neighbors on $H$.
  Thus, either $|N(x)\cap H| = 4$, or $|N(x)\cap H| = 3$ and the three neighbors are not consecutive.
  In either case, $G'$ contains an induced $C_{4}$, which is switching equivalent to the claw, which is a forbidden induced subgraph for line graphs.

  In the rest, assume $G$ is not in $\mathcal{S}(C_{5})$.  Since a line graph is $(2, 2, 2)$-free and claw-free, $G$ must be in $\mathcal{C}_{0}$.
  By Lemma~\ref{lem:c4 c5 c6 free}, $G$ is
  $(+)$, $(+,+,1)$, $(+,1,+)$, $(+,0,+)$, $(+,+,1,0,+)$, $(+,0,+,0,1)$, $(+,+,1,+)$, or $(+,+,1,+,+)$.
  Of the nine forbidden induced subgraphs of line graphs, seven contains an induced subgraph in $\mathcal{S}(C_{4}, C_{5}, C_{6})$.  The other two contains a switching equivalent of $(1, 3, 1)$ as an induced subgraph. 
Thus, the statement follows.
\end{proof}

\subsection{Minor-closed graphs}

A graph $F$ is a \emph{minor} of a graph $G$ if $F$ can be obtained from a subgraph of $G$ by contracting edges (identifying the two ends of the edge and keeping one edge between the resulting vertex and each of
the neighbors of the end points of the edge).  For example, any cycle contains all shorter cycles as minors.
A graph class $\mathcal{G}$ is \emph{minor-closed} if every minor of a graph in $\mathcal{G}$ also belongs to $\mathcal{G}$.
In other words, there is a set $\mathcal{M}$ of \emph{forbidden minors} such that a graph belongs to $\mathcal{G}$ if and only if it does not contain as a minor any graph in $\mathcal{M}$. 
Since an induced subgraph of a graph $G$ is a minor of $G$, a minor-closed graph class is hereditary.
We say that a graph class is \emph{nontrivial} if there is at least one graph not in the class.

Kostochka~\cite{kostochka1982minimum,kostochka1984lower} and Thomason~\cite{thomason1984extremal} proved that, there exists an absolute constant $c>0$ such that every graph $G$ with at least $c\cdot |V(G)|\cdot p\sqrt{p}$ edges has $K_p$ as a minor. See~\cite{DBLP:journals/jct/Thomason01} for an overview. This helps us to prove  \cref{thm:lower-class-2}:

\lowerclasstwo*

\begin{proof}
    Let $G\in \LW{\mathcal{G}}$ be a graph with $n$ vertices. It is straight-forward to verify that there exists a constant $c'>0$ such that either
    $G$ or $S(G, A)$ has $c'\cdot n^2$ edges, where 
    $A$ is any subset of $V(G)$ with cardinality $\lfloor n/2\rfloor$. If $c'\cdot n^2 \geq c\cdot n\cdot p\sqrt{p}$, then $G$ has a $K_p$-minor.
    Therefore, $n = O(p\sqrt{p})$.
\end{proof}
Theorem~\ref{thm:lower-class-2} implies that, for any nontrivial minor-closed graph $\mathcal{G}$ class, there are only a finite number of graphs in the lower $\mathcal{G}$ switching class. This also means a trivial constant-time algorithm for checking whether a graph belongs to such a lower switching class.
Some interesting minor-closed graph classes are planar, outerplanar~\cite{syslo1979characterizations}, series parallel~\cite{zbMATH03209598}, bounded genus, bounded treewidth~\cite{DBLP:journals/tcs/Bodlaender98}, and bounded pathwidth graphs.
For outerplanar graphs, whose forbidden minors are $K_{4}$ and $K_{2, 3}$, 
we obtain a tight bound of 5 vertices.

\begin{proposition}\label{outerplar switching class: lm1}
    A graph in the lower outerplanar switching class has at most five vertices.
\end{proposition}
\begin{proof}
  Let $G$ be a graph in the lower outerplanar switching class with five or more vertices.
  We start with showing that $G$ is $C_4$-free.
  Suppose for contradiction that $G$ contains an induced cycle $H$ of length four.
  Since $|V(G)| \ge 5$, there is another vertex $x$ not on this cycle.  Let $G'$ be the subgraph of $G$ induced by $V(H)\cup \{x\}$.
  If $x$ is adjacent to an even number
  of vertices on $H$, then $G'$ either is $K_{2, 3}$, or can be switched to a $W_{4}$ (Figure~\ref{fig:switching equivalent house}a), which contains $K_{4}$ as a minor.
  If $x$ has three neighbors on $H$, then both $K_{4}$ and $K_{2, 3}$ are minors of $G'$.
  If $x$ has only one neighbor on $H$, then $G'$ can be switched to the previous case (Figure~\ref{fig:switching equivalent house}b).  
  In either case, $G$ cannot be in the lower outerplanar switching class.
  
  Since $G$ is $K_{4}$-free, it is $2K_2$-free because $2K_2$ and $K_4$ are switching equivalent. 
Thus, $G$ is a pseudo-split graph.
  If $G$ contains an induced $C_{5}$, then it must be a $C_{5}$.
  (Note that $W_{5}$ is not a outerplanar graph, while $\overline{W_{5}}$, i.e., the graph consisting a $C_{5}$ and an isolated vertex, is switching equivalent to $W_{5}$.)
  Otherwise, $G$ is a split graph.
  Let $K\cup I$ be a split partition of $G$.
  Since $\overline{K_4}$ is switching equivalent to $C_4$, the graph $G$ cannot contain an independent set of four vertices.  
  Thus, $|K| \le 3$ and $|I| \le 3$.
  It remains to show that at least one inequality is strict.
  Suppose that $|K| = |I| =3$.
  If there is an isolated vertex (which is in $I$), then
  there is an induced $K_{3} + K_{1}$, which is switching equivalent to $K_{4}$.
  There is a $K_{4}$ if a vertex in $I$ is adjacent to all the vertices in $K$.
  Since $G$ cannot contain an induced $K_{4}$ or claw (both in $\mathcal{S}(C_{4})$), a vertex in $K$ has either one or two neighbors in $I$.
Thus, $G$ is either the net or the sun, both switching equivalent to $W_{5}$.
  Therefore, the order of $G$ is at most five.
\end{proof}
One can derive from the proof all the graphs in this switching class.  They include four graphs of order five, $\mathcal{S}(C_{5})$, eight graphs of order four, and all graphs on three or fewer vertices.

With a similar argument, one can show that the order of a graph in the lower series-parallel switching class is lower than 13. With a computer program, we obtained that the maximum order of graphs in a lower planar switching class is 7, and the seven vertex graphs in the class are graphs in \SW{C_7}; see Figure~\ref{fig:switching equivalent C7}.

\begin{figure}[h]
\tikzstyle{filled vertex}  = [{circle,draw=blue,fill=black!50,inner sep=1pt}]  \tikzstyle{empty vertex}  = [{circle, draw, fill = white, inner sep=1.pt}]
  \centering \small
  \begin{subfigure}[b]{0.35\linewidth}
    \centering
    \begin{tikzpicture}[scale=.6]
      \def\n{4}
      \def\radius{1.5}      
      \node[empty vertex] (c) at (0, 0) {};
      \coordinate (v0) at ({90 + 180 / \n}:\radius) {};
      \foreach \i in {1,..., \n} {
        \pgfmathsetmacro{\angle}{90 - (\i - .5) * (360 / \n)}
        \node[empty vertex] (v\i) at (\angle:\radius) {};
        \draw let \n1 = {int(\i - 1)} in (v\n1) -- (v\i);
        \draw (c) -- (v\i);
      }
    \end{tikzpicture}
    \;
    \begin{tikzpicture}[scale=.6]
      \def\n{4}
      \def\radius{1.5}      
      \coordinate (v0) at ({90 + 180 / \n}:\radius) {};
      \foreach \i in {1,..., \n} {
        \pgfmathsetmacro{\angle}{90 - (\i - .5) * (360 / \n)}
        \node[empty vertex] (v\i) at (\angle:\radius) {};
        \draw let \n1 = {int(\i - 1)} in (v\n1) -- (v\i);
      }
      \node[empty vertex] (c) at (0, 0) {};
      \foreach \i in {1, 4} 
      \draw (c) -- (v\i);
      \foreach \i in {2, 3} 
        \node[filled vertex] at (v\i) {};
    \end{tikzpicture}
    \;
    \begin{tikzpicture}[scale=.6]
      \def\n{4}
      \def\radius{1.5}      
      \node[filled vertex] (c) at (0, 0) {};
      \coordinate (v0) at ({90 + 180 / \n}:\radius) {};
      \foreach \i in {1,..., \n} {
        \pgfmathsetmacro{\angle}{90 - (\i - .5) * (360 / \n)}
        \node[empty vertex] (v\i) at (\angle:\radius) {};
        \draw let \n1 = {int(\i - 1)} in (v\n1) -- (v\i);
      }
    \end{tikzpicture}
\caption{}
  \end{subfigure}
  \begin{subfigure}[b]{0.25\linewidth}
    \centering
    \begin{tikzpicture}[scale=.6]
      \def\n{4}
      \def\radius{1.5}      
      \coordinate (v0) at ({90 + 180 / \n}:\radius) {};
      \foreach \i in {1,..., \n} {
        \pgfmathsetmacro{\angle}{90 - (\i - .5) * (360 / \n)}
        \node[empty vertex] (v\i) at (\angle:\radius) {};
        \draw let \n1 = {int(\i - 1)} in (v\n1) -- (v\i);
      }
      \node[empty vertex] (c) at (0, 0) {};
      \foreach \i in {1, 2, 3}      
      \draw (c) -- (v\i);
    \end{tikzpicture}
    \;
    \begin{tikzpicture}[scale=.6]
      \def\n{4}
      \def\radius{1.5}      
      \coordinate (v0) at ({90 + 180 / \n}:\radius) {};
      \foreach \i in {1,..., \n} {
        \pgfmathsetmacro{\angle}{90 - (\i - .5) * (360 / \n)}
        \node[empty vertex] (v\i) at (\angle:\radius) {};
        \draw let \n1 = {int(\i - 1)} in (v\n1) -- (v\i);
      }
      \node[filled vertex] (c) at (0, 0) {};
      \foreach \i in {4}      
      \draw (c) -- (v\i);
    \end{tikzpicture}
\caption{}
  \end{subfigure}
  \begin{subfigure}[b]{0.14\linewidth}
    \centering
    \begin{tikzpicture}[scale=.6]
      \def\n{4}
      \def\radius{1.5}      
      \coordinate (v0) at ({90 + 180 / \n}:\radius) {};
      \foreach \i in {1,..., \n} {
        \pgfmathsetmacro{\angle}{90 - (\i - .5) * (360 / \n)}
        \node[empty vertex] (v\i) at (\angle:\radius) {};
        \draw let \n1 = {int(\i - 1)} in (v\n1) -- (v\i);
      }
      \node[empty vertex] (c) at (0, 0) {};
      \foreach \i in {1, 3}      
      \draw (c) -- (v\i);
    \end{tikzpicture}
    \caption{}
  \end{subfigure}
  \caption{All five-vertex graph containing a $C_{4}$ form three groups.  The set $A$ consists of the solid nodes.}
  \label{fig:switching equivalent house}
\end{figure}

\subsection{Lower switching classes with infinite forbidden induced subgraphs}
\begin{proposition}
 \label{switching class general: ob1}   
 Let $G$ and $H$ be two graphs.  The graph $G$ is $\mathcal{S}(H)$-free if and only if every switching equivalent graph of $G$ is $\mathcal{S}(H)$-free.
\end{proposition}
\begin{proof}
  Since $G$ is switching equivalent to itself, the sufficiency is trivial.
  We prove the necessity by contradiction.  Suppose that there are subsets $A, U\subseteq V(G)$ such that the subgraph induced by $U$ in $S(G, A)$
  is isomorphic to some graph $H'\in \mathcal{S}(H)$.
If we switch $A$ in $S(G, A)$, the subgraph induced by $U$ in the resulting graph is $S(S(G, A)[U], A\cap U)$, which is switching equivalent to $H'$, hence in $\mathcal{S}(H)$ by Proposition~\ref{prop:basic properties}. 
  Since $S(S(G, A), A) = G$, we have a contradiction because $G$ is supposed to be $\mathcal{S}(H)$-free.
\end{proof}

We have seen a lot of graph classes $\mathcal{G}$ with infinite forbidden induced subgraphs, but the lower $\mathcal{G}$ switching class has only a finite number of them.
This is not always the case.  Indeed, a hole contains an induced copy of a switching equivalent graph of a shorter hole can only happen for very short ones.
For all $\ell \ge 9$, every switching equivalent graph of $C_{\ell}$ either is $C_\ell$ itself or contains a vertex of degree at least three.
\begin{proposition}
  \label{counter example class: ob2}
  Let $i$ and $j$ be two integers with $9\le i < j$.
  Every switching equivalent graph of $C_{j}$ is $\mathcal{S}(C_{i})$-free.
\end{proposition}
\begin{proof}
  By Proposition~\ref{switching class general: ob1}, it suffices to show that $C_{j}$ is $\mathcal{S}(C_{i})$-free.
  We consider $S(C_{i}, A)$ for each subset $A\subseteq V(C_{i})$; we may assume without loss of generality that $|A| \le i/2$.
  Obvious, $C_{j}$ does not contain an induced copy of $C_{i}$.
  Hence, $A \ne \emptyset$, and let $v$ be an arbitrary vertex in $A$.
  Since $v$ has precisely two neighbors in $C_{i}$, and $|V(C_{i})\setminus A| \ge 5$, it has at least three neighbors in $S(C_{i}, A)$.
  Thus, $C_{j}$ does not contain an induced copy of $S(C_{i}, A)$.
  This concludes the proof.
\end{proof}

Therefore, the following lower switching classes have an infinite number of forbidden induced subgraphs.

\begin{corollary}\label{counter example class: ob3}
  For any infinite set $I\subseteq \{9, 10, \ldots\}$, the forbidden induced subgraphs of the lower $\{C_{\ell}, \ell \in I\}$-free switching class are
 $\bigcup_{\ell\in I}\mathcal{S}(C_{\ell})$.
\end{corollary}

 \section{Upper switching classes:  algorithms}\label{sec:upper}

For the recognition of the upper $\mathcal{G}$ switching class, the input is a graph $G$, and the solution is a vertex subset $A\subseteq V(G)$ such that $S(G, A)\in \mathcal{G}$.

\subsection{(Pseudo-)split graphs}

We start with split graphs.
If the input graph $G$ is a split graph, then we have nothing to do.
Suppose that $G$ is in the upper split switching class.  Let $A$ be a solution, and $K\uplus I$ a split partition of $S(G, A)$.
Note that if $A\in \{K, I\}$, then $G$ is a split graph.
We may assume that $A$ intersects both $K$ and $I$: if $A$ is a proper subset of $K$ or $I$, we replace $A$ with $V(G)\setminus A$.
We can guess a pair of vertices $u\in A\cap K$ and $v\in A\cap I$.
The vertex set $V(G)\setminus \{u, v\}$ can be partitioned into four parts, namely, $N(u)\setminus N[v]$, $N(v)\setminus N[u]$, $N(u)\cap N(v)$, and $V(G) \setminus N[u, v]$.
It is easy to see that the first is a subset of $A$ while the second is disjoint from $A$.
The subgraphs
$G[N(u)\cap N(v)]$ and $G - N[u, v]$ must be split graphs, and each admits a special split partition with respect to $A$.  Although a split graph may admit more than one split partition, the following observation allows us to find the desired one by enumeration.

\begin{proposition}[\cite{DBLP:journals/algorithmica/FominGST20}]
  \label{prop:split-partitions}
  A split graph has at most $n$ split partitions
   and they can be enumerated in $O(m+n)$ time.
\end{proposition} 

\begin{figure}[h!]
  \centering 
  \begin{tikzpicture}
    \path (0,0) node[text width=.85\textwidth, inner xsep=20pt, inner ysep=10pt] (a) {
      \begin{minipage}[t!]{\textwidth}
        \begin{tabbing}
          AAA\=Aaa\=aaa\=Aaa\=MMMMMAAAAAAAAAAAA\=A \kill
          1. \> \textbf{if} $G$ is a split graph \textbf{then return} ``\YES'';
          \\
          2. \> \textbf{for each} pair of vertices $u, v\in V(G)$ \textbf{do}
          \\
2.1. \>\> \textbf{if} $G[N(u)\cap N(v)]$ is not a split graph \textbf{then continue};
          \\
          2.2. \>\> \textbf{if} $G - N[u, v]$ is not a split graph \textbf{then continue};
          \\
          2.3. \>\> \textbf{for each} split partition $K_{1}\uplus I_{1}$ of $G[N(u)\cap N(v)]$ \textbf{do}
          \\
          2.3.1. \>\>\> \textbf{for each} split partition $K_{2}\uplus I_{2}$ of $G - N[u, v]$ \textbf{do}
          \\
          2.3.1.1. \>\>\>\> \textbf{if} $S(G, \{u,v\} \cup (N(u)\setminus N[v]) \cup K_1 \cup I_2)$ is a split graph \textbf{then return} ``\YES'';
          \\
          3. \> \textbf{return} ``\NO.''
        \end{tabbing}
      \end{minipage}
    };
    \draw[draw=gray!60] (a.north west) -- (a.north east) (a.south west) -- (a.south east);
  \end{tikzpicture}
  \caption{The algorithm for split graphs.}
  \label{alg:split}
\end{figure}

\begin{theorem} \label{thm:split}
  We can decide in polynomial time whether a graph can be switched to a split graph.
\end{theorem}
\begin{proof}
  We use the algorithm described in Figure~\ref{alg:split}.
  Since the algorithm returns ``\YES'' only when a solution is identified, it suffices to show a solution must be returned for a yes-instance.
  Suppose that $S(G, A)$ is a split graph, and $K\uplus I$ is a split partition of $S(G, A)$.
  If $A$ is the empty set, $K$, or their complements, then $G$ is a split graph, and step~1 returns ``\YES.''
Hence, we may assume without loss of generality that neither $A\cap K$ nor $A\cap I$ is empty; otherwise, we replace $A$ with $V(G)\setminus A$.
In one of the iterations of step~2, the vertex $u$ is in $A\cap K$ the vertex $v$ is in $A\cap I$.
  We first argue that $N(u)\setminus N[v]$ must be a subset of $A$, and $N(v)\setminus N[u]$ must be disjoint from $A$.
  If a vertex $x\in N(v)\setminus N[u]$ is in $A$, or if a vertex $y\in N(u)\setminus N[v]$ is not in $A$, then they are adjacent to $v$ but not $u$ in $S(G, A)$, which is impossible.
Since $K\uplus I$ is a split partition,
  \begin{align*}
    N(v)\cap A &\subseteq K
    \\
    N(u)\setminus A&\subseteq I.
  \end{align*}
  Thus, $G[N(u)\cap N(v)]$ is a split graph, with a split partition $(N(u)\cap N(v)\cap A)\uplus (N(u)\cap N(v)\setminus A)$.
By symmetry, $(V(G)\setminus (N[u, v]\cup A))\uplus (A\setminus (N[u, v]))$ is a split partition of $G - N[u, v]$.
  (Note that $S(\overline G, A) = \overline{S(G, A)}$ is also a split graph.)
  Step~2.3 checks all split partitions of $G[N(u)\cap N(v)]$ and $G - N[u, v]$.
  By Proposition~\ref{prop:split-partitions}, in one of the iterations,
  \begin{align*}
    K_{1} &= N(u)\cap N(v)\cap A,
    \\
    I_{2} &= A\setminus (N[u, v]).
  \end{align*}
Since $A = \{u,v\} \cup (N(u)\setminus N[v]) \cup K_1 \cup I_2$, the algorithm must return ``\YES'' in step~2.3.1.1 of this iteration.

  By Proposition~\ref{prop:split-partitions}, there is a linear number of iterations in step~2.3 and step~2.3.1.
  The algorithm takes $O(n^4(m+n))$ time.
\end{proof}

Step~2 of the algorithm in Figure~\ref{alg:split} can be easily modified to enumerate all solutions.  The same holds when $G$ is a split graph.

\begin{theorem}\label{thm:split2}
  Let $G$ be a graph.  There are a polynomial number of subsets $A$ of $V(G)$ such that $S(G, A)$ is a split graph, and they can be enumerated in polynomial time.
\end{theorem}
\begin{proof}
  Suppose first that $G$ is a split graph, and let $K\uplus I$ be a split partition of $G$.
  The four trivial solutions are $V(G)$, $\emptyset$, $K$, and $I$.
  We first argue that for any solution $A$, we have $|A\cap K| \le 1$ or $|A\cap K| \ge |K| - 1$.
  Suppose for contradiction that
  \[
    2\le |A\cap K| \le |K| - 2.
  \]
  Then two vertices in $A\cap K$ and two vertices in $K\setminus A$ induce an $\overline{C_{4}}$ in $S(G, A)$, a contradiction.
  A symmetric argument show that $|A\cap I| \le 1$ or $|A\cap I| \ge |I| - 1$.
  Thus, the number of solutions is $O(n^2)$, and they can be enumerated in polynomial time.

  In the rest, $G$ is not a split graph.
  As we see in the proof of Theorem~\ref{thm:split}, any solution $A$ of a graph in the upper split switching class corresponds to an iteration of step 2 of Figure~\ref{alg:split}.
After finding a solution, if we we record the solution and continue, instead of terminating the algorithm after returning the solution, we can enumerate all possible solutions.
Since the number of iterations is polynomial on $n$, the statement follows.
\end{proof}

A \emph{pseudo-split graph} is either a split graph, or a graph whose vertex set can be partitioned into a clique $K$, an independent set $I$, and a set $H$ that (1) induces a $C_{5}$; (2) is complete to $K$; and (3) is nonadjacent to $I$.
We say that $K\uplus I \uplus H$ is a \emph{pseudo-split partition} of the graph, where $H$ may or may not be empty.
If $H$ is empty, then $K\uplus I$ is a split partition of the graph.  When $H$ is nonempty, the graph has a unique pseudo-split partition.
Similar to split graphs, the complement of a pseudo-split graph remains a pseudo-split graph.

For pseudo-split graphs, we may start with checking whether the input graph can be switched to a split graph.  We are done if the answer is ``yes.''
Henceforth, we are looking for a resulting graph that contains a hole $C_{5}$.
Suppose that $G$ is in the upper pseudo-split switching class.  Let $A$ be a solution, and $K\uplus I \uplus H$ is a \emph{pseudo-split partition} of $S(G, A)$.
We may assume that $|A\cap H| \ge 3$: otherwise, we replace $A$ with $V(G)\setminus A$.
The subgraph $G[H]$ must be one of Figure~\ref{fig:s-c4-c5}b, and $A\cap H$ are precisely the vertices represented as empty nodes.
We can guess the vertex set $H$ as well as its partition with respect to $A$, and then all the other vertices are fixed by the following observation:
\begin{itemize}
\item $K$ is complete to $H\cap A$ and nonadjacent to $H\setminus A$, and
\item $I$ is complete to $H\setminus A$ and nonadjacent to $H\cap A$.
\end{itemize}

\begin{figure}[h!]
  \centering 
  \begin{tikzpicture}
    \path (0,0) node[text width=.85\textwidth, inner xsep=20pt, inner ysep=10pt] (a) {
      \begin{minipage}[t!]{\textwidth}
        \begin{tabbing}
          AAA\=Aaa\=aaa\=Aaa\=MMMMMAAAAAAAAAAAA\=A \kill
          1. \> \textbf{if} $G$ can be switched to a split graph \textbf{then return} ``\YES'';
          \\
          2. \> \textbf{for each} vertex set $H$ such that $G[H]\in \mathcal{S}(C_{5})$ \textbf{do}
          \\
          2.0. \>\> $H_{1}\leftarrow$ the empty nodes of $G[H]$ as in Figure~\ref{fig:s-c4-c5}b; $H_{2}\leftarrow H\setminus H_{1}$;
          \\
          2.1. \>\> \textbf{for each} vertex $x$ in $V(G)\setminus H$ \textbf{do}
          \\
          2.1.1. \>\>\> \textbf{if} $N(x)\cap H$ is neither $H_{1}$ nor $H_{2}$ \textbf{then continue};
          \\
          2.2. \>\> \textbf{if} $N(H_{1})\setminus H$ does not induce a split graph \textbf{then continue};
          \\
          2.3. \>\> \textbf{if} $N(H_{2})\setminus H$ does not induce a split graph \textbf{then continue};
          \\
          2.4. \>\> \textbf{for each} split partition $K_{1}\uplus I_{1}$ of the subgraph induced by $N(H_{1})\setminus H$ \textbf{do}
          \\
          2.4.1. \>\>\> \textbf{for each} split partition $K_{2}\uplus I_{2}$ of the subgraph induced by $N(H_{2})\setminus H$ \textbf{do}
          \\
          2.4.1.1. \>\>\>\> \textbf{if} $S(G, H_{1} \cup K_1 \cup I_2)$ is a pseudo-split graph \textbf{then return} ``\YES'';
          \\
          3. \> \textbf{return} ``no.''
        \end{tabbing}
      \end{minipage}
    };
    \draw[draw=gray!60] (a.north west) -- (a.north east) (a.south west) -- (a.south east);
  \end{tikzpicture}
  \caption{The algorithm for pseudo-split graphs.}
  \label{alg:pseudo-split}
\end{figure}

\begin{theorem}
\label{thm:$2K_2 ,C_4$}
We can decide in polynomial time whether a graph can be switched to a pseudo-split graph.
\end{theorem}
\begin{proof}
  We use the algorithm described in Figure~\ref{alg:pseudo-split}.
  Since the algorithm returns ``\YES'' only when a solution is identified, it suffices to show a solution must be returned for a yes-instance.
  Suppose that $S(G, A)$ is a pseudo-split graph, and $K\uplus I\uplus H$ is a pseudo-split partition of $S(G, A)$.
  If $H$ is empty, then $S(G, A)$ is a split graph, and step~1 returns ``\YES.''
  Henceforth, $H\ne\emptyset$.
  We may assume that $|A\cap H| \ge 3$: otherwise, we replace $A$ with $V(G)\setminus A$.
  One of the iterations of step~2 uses the vertex set $H$.
  Note that $G[H] = S(C_{5}, A\cap H)$, and hence must be a graph in $\mathcal{S}(C_{5})$.  Moreover, $A\cap H$ are precisely the vertices represented as empty nodes in Figure~\ref{fig:s-c4-c5}b.
  Thus, $H_{1} = A\cap H$, and $H_{2} = H\setminus A$.
  By definition, $H$ is complete to $K$ and nonadjacent to $I$ in $S(G, A)$.
Thus, $N(H_{1})\setminus H$ and $N(H_{2})\setminus H$ is a partition of $V(G)\setminus H$.
  More specifically,
  \begin{align*}
    K &= (A\cap N(H_{1})\setminus H) \cup (N(H_{2})\setminus (A\cup H)),
    \\
    I &= (N(H_{1})\setminus (A\cup H)) \cup (A\cap N(H_{2})\setminus H).
  \end{align*}
  Thus, both $N(H_{1})\setminus H$ and $N(H_{2})\setminus H$ induce split graphs.
  The algorithm will pass the tests in steps~2.1--2.3.
  By Proposition~\ref{prop:split-partitions}, in one of the iterations of step~2.4, 
  \begin{align*}
    K_{1} &= A\cap N(H_{1})\setminus H,
    \\
    I_{2} &= A\cap N(H_{2})\setminus H.
  \end{align*}
  Since $A = H_{1}\cup K_1 \cup I_2$, the algorithm must return ``\YES'' in step~2.4.1.1 of this iteration.

  By Proposition~\ref{prop:split-partitions}, there is a linear number of iterations in step~2.4 and step~2.4.1.
  The algorithm takes $O(m^{3} n^4)$ time.
\end{proof}

Similar to Theorem~\ref{thm:split2}, we have the following result.

\begin{theorem}
  Let $G$ be a graph.  There are a polynomial number of subsets $A$ of $V(G)$ such that $S(G, A)$ is a pseudo-split graph, and they can be enumerated in polynomial time.
\end{theorem}

As a result, we have an algorithm for any hereditary subclass $\mathcal{G}$ of pseudo-split graphs that can be recognized in polynomial time.

\begin{corollary}
  Let $\mathcal{G}$ be any subclass of pseudo-split graphs.  If $\mathcal{G}$ can be recognized in polynomial time, then we can decide in polynomial time whether a graph can be switched to a pseudo-split graph.  
\end{corollary}

Since a graph has $2^{n}$ subsets, and the switching of only a polynomial number of them leads to a pseudo-split graph, every graph of sufficiently large order can be switched to a graph that is not a pseudo-split graph.

\begin{corollary}
  The lower pseudo-split switching class is finite.
\end{corollary}

 \subsection{Paw-free graphs}

Since a paw contains an induced $C_{3}$ and an induced $\overline{P_3}$, both ${C_3}$-free graphs and $\overline{P_3}$-free graphs are paw-free.
Olariu~\cite{DBLP:journals/ipl/Olariu88} showed that
a connected paw-free graph is ${C_3}$-free or $\overline{P_3}$-free.
Note that $\overline{P_3}$-free graphs are precisely complete multipartite graphs, and a connected complete multipartite graphs either is trivial or has at least two parts.

Hayward~\cite{hayward1996recognizingp3} presented a polynomial-time algorithm for deciding whether a graph is switching equivalent to a $C_{3}$-free graph.  Kratochv{\'\i}l et al.~\cite{kratochvil1992computational} dealt with $P_3$-free graphs; see also Jel{\'{\i}}nkov{\'{a}} and Kratochv{\'{\i}}l~\cite{DBLP:journals/jgt/JelinkovaK14}.
We may start with calling these algorithms to check whether $G$ can be switched to a ${C_3}$-free graph or a $\overline{P_3}$-free graph.  We proceed only when the answers are both ``\NO.''  Hence, we are looking for a set $A\subseteq V(G)$ such that $S(G, A)$ is not connected and contains a triangle.
It is quite simple when $S(G, A)$ has three or more components.  We can always assume that $A$ intersects two of them.  We guess one vertex from each of these intersections, and an arbitrary vertex from another component (which can be in $A$ or not).  The three vertices are sufficient to determine $A$.
It is more challenging when $S(G, A)$ comprises precisely two components.
The crucial observation here is that one of the components is ${C_3}$-free and the other $\overline{P_3}$-free.
We have assumed the graph contains a triangle.
If both components contain triangles, hence $\overline{P_3}$-free, then $S(G,A)$ can be switched to a complete multipartite graph, contradicting the assumption above.
We guess a triple of vertices that forms a triangle in $S(G,A)$, and they can determine $A$.

A \emph{co-component} of a graph $G$ is a component of the complement of $G$.  Indeed, a graph is a complete multipartite if and only if every co-component is an independent set.

\begin{figure}[h!]
  \centering 
  \begin{tikzpicture}
    \path (0,0) node[text width=.85\textwidth, inner xsep=20pt, inner ysep=10pt] (a) {
      \begin{minipage}[t!]{\textwidth}
        \begin{tabbing}
          AAA\=Aaa\=aaa\=Aaa\=Aaa\=MMMMMAAAAAAAAAAAA\=A \kill
          1. \> \textbf{if} $G$ can be switched to a $\overline{P_{3}}$- or $C_{3}$-free graph \textbf{then return} ``\YES'';
          \\
          2. \> \textbf{for each} pair of nonadjacent vertices $u_1, u_2$ \textbf{do} // three or more components.
          \\
          2.1. \>\> \textbf{for each} $u_{3}\in V(G)\setminus N[u_1, u_2]$ \textbf{do} 
          \\
          2.1.1. \>\>\> $A\leftarrow \{x\in V(G)\mid |N[x]\cap \{u_1, u_2, u_3\} \le 1\}$;
\\
          2.1.2. \>\>\> \textbf{if} $S(G, A)$ is paw-free \textbf{then return} ``\YES'';
          \\
          2.2. \>\> \textbf{for each} $u_{3}\in N(u_1)\cap N(u_2)$ \textbf{do}  
          \\
          2.2.1. \>\>\> $A\leftarrow (V(G)\setminus N[u_{1}, u_{2}]) \cup ( (N[u_{1}]\Delta N[u_{2}])\setminus N(u_{3}) )$;
\\
          2.2.2. \>\>\> \textbf{if} $S(G, A)$ is paw-free \textbf{then return} ``\YES'';
          \\
          3. \> \textbf{for each} pair of adjacent vertices $u_1, u_2$ \textbf{do} // two components, one containing $C_3$.
          \\
          3.1. \>\> $p\leftarrow$ number of components of $G[N(u_{1})\cap N(u_{2})]$;
          \\
          3.2. \>\> $q\leftarrow$ number of components of $G - N[u_1, u_2]$;
          \\
          3.3. \>\> \textbf{for each} $I\subseteq \{1, \ldots, p\}$ and $J\subseteq \{1, \ldots, q\}$ with $|I|, |J| \le 2$ \textbf{do}
          \\
3.3.1. \>\>\> $X\leftarrow\bigcup_{i\not\in I} i\text{th co-component of } G[N(u_1)\cap N(u_2)]$;
          \\
          3.3.2. \>\>\> $Y\leftarrow\bigcup_{j\in J} j\text{th co-component of } G - N[u_1, u_2]$;
\\
          3.3.3. \>\>\> \textbf{if} $X\ne \emptyset$ \textbf{then}
          \\
          3.3.3.1. \>\>\>\> $u_{3}\leftarrow$ an arbitrary vertex from $X$;
          \\
          3.3.3.2. \>\>\>\> $A\leftarrow X\cup Y\cup ( ( N(u_{1}) \Delta N(u_{2}) )\cap N(u_3))$;
          \\
          3.3.4. \>\>\> \textbf{else}
          \\
          3.3.4.1. \>\>\>\> $u_{3}\leftarrow$ an arbitrary vertex from $V(G) \setminus ( N[u_1, u_2]\cup Y)$;
          \\
          3.3.4.2. \>\>\>\> $A\leftarrow X\cup Y\cup ( ( N(u_{1}) \Delta N(u_{2}) ) \setminus N(u_3))$;
          \\
          3.3.5. \>\>\> \textbf{if} $S(G, A)$ is paw-free \textbf{then return} ``\YES'';
          \\
          4. \> \textbf{return} ``\NO.''
        \end{tabbing}
      \end{minipage}
    };
    \draw[draw=gray!60] (a.north west) -- (a.north east) (a.south west) -- (a.south east);
  \end{tikzpicture}
  \caption{The algorithm for paw-free graphs.
}
  \label{alg:paw-free}
\end{figure}

\begin{theorem} \label{thm:paw-free}
  We can decide in polynomial time whether a graph can be switched to a paw-free graph.
\end{theorem}
\begin{proof}
  We use the algorithm described in Figure~\ref{alg:paw-free}.
Since the algorithm returns ``\YES'' only when a solution is identified (in steps~1--3), it suffices to show that given a yes-instance, the algorithm always finds a solution.
Suppose that there is a set $A\subseteq V(G)$ such that $S(G, A)$ is paw-free.
  If $S(G, A)$ is ${C_3}$-free or $\overline{P_3}$-free, then step~1 returns ``\YES.''
Henceforth, $S(G, A)$ is not connected and contains a triangle.
  
  Suppose first that $S(G, A)$ consists of three or more components.  We may assume without loss of generality that $A$ intersects two or more components of $S(G, A)$; otherwise, we replace $A$ with $V(G)\setminus A$.
  We number the components of $S(G, A)$ as $V_{1}$, $\ldots$, $V_{p}$ such that those with nonempty intersection with $A$ are ordered first.
  In one of the iterations, $u_{1} \in A\cap V_{1}$ and $u_{2} \in A\cap V_{2}$ (note that there is no edge between $A\cap V_{1}$ and $A\cap V_{2}$).
  Depending on whether $A\cap V_{3}$ is empty (i.e., whether $A\subseteq (V_1 \cup V_2)$), we separate into two cases.
  \begin{itemize}
  \item Case 1, $A\cap V_{3}\ne \emptyset$.
    In one of the iterations of step~2.1, $u_{3}$ is a vertex in $A\cap V_{3}$.
    Note that in $S(G, A)$, no vertex is adjacent to two or more vertices in $\{u_1, u_2, u_3\}$.
    Thus, a vertex $x\in V(G)\setminus \{u_1, u_2, u_3\}$ is in $A$ if and only if it is adjacent to at most one vertex in $\{u_1, u_2, u_3\}$.   In this case, step~2.1.2 returns ``\YES.''
  \item Case 2, $A\cap V_{3} = \emptyset$, i.e., $A\subseteq (V_1 \cup V_2)$.
    In one of the iterations of step~2.2, $u_{3}$ is a vertex in $V_{3}$.
    Since no vertex is adjacent to both $u_1$ and $u_2$ in $S(G, A)$, 
  \[
    V(G)\setminus (N[u_{1}, u_{2}]) \subseteq A \subseteq V(G)\setminus (N(u_{1})\cap N(u_{2})).
  \]
  It remains to deal with vertices in $N[u_{1}]\Delta N[u_{2}]$.  Note that each vertex in $N[u_{1}]\Delta N[u_{2}]\setminus \{u_1, u_2\}$ is adjacent to precisely one of $u_1$ and $u_2$ in $S(G, A)$.  Thus, $N[u_{1}]\Delta N[u_{2}]\subseteq V_{1}\cup V_{2}$, and
  \[
    A \cap (N[u_{1}]\Delta N[u_{2}]) = (N[u_{1}]\Delta N[u_{2}])\setminus N(u_{3}).
  \]
  Thus, step~2.2.2 returns ``\YES'' in this case.
  \end{itemize}

  In the sequel, $S(G, A)$ consists of precisely two components, denoted by $G_{1}$ and $G_{2}$.
  Recall that there exists a triangle in $S(G, A)$.
  May assume without loss of generality that $G_{1}$ contains a triangle.
  We argue that $G_{2}$ is triangle-free.
  Suppose that both $G_{1}$ and $G_{2}$ contain triangles, and hence they are complete multipartite graphs with at least three parts \cite{DBLP:journals/ipl/Olariu88}.
  Then
  \[
    S(G, A\Delta V(G_1)) = S(S(G, A), V(G_1))
  \]
  is a connected complete multipartite graph, and step~1 should have returned ``\YES.''

  We fix a triangle $\{u_{1}, u_{2}, u_{3}\}$ of $S(G, A)$.
  We may assume without loss of generality that $\{u_{1}, u_{2}\}\subseteq A$; otherwise, we replace $A$ with $V(G)\setminus A$ and renumber the vertices.
  Since $u_{1}$, $u_{2}$, and $u_{3}$ are pairwise adjacent in $S(G, A)$, they belong to different parts of $G_{1}$.
  We number the parts of $G_{1}$ as $V_{1}$, $\ldots$, $V_{r}$ such that $u_{i}\in V_{i}, i = 1, 2, 3$.
  Then
  \[
    N(u_{1})\cap N(u_{2})\cap V(G_1) = A\cap \bigcup_{i = 3}^{r} V_{i} \subseteq A.
  \]
Note that no vertex in $V(G_2)\cap A$ is adjacent to $u_1$ or $u_2$, while every vertex in $V(G_2)\setminus A$ is adjacent to both $u_1$ and $u_2$.
  One of the iterations of step~3 identifies $u_{1}$ and $u_{2}$ correctly. 
If $N(u_{1})\cap N(u_{2})$ intersects both $V(G_1)$ and $V(G_2)$, then all the edges between these two subsets are present in $G$.  Hence, the complement of $G[N(u_{1})\cap N(u_{2})]$ is not connected, and a co-component of $G[N(u_{1})\cap N(u_{2})]$ is completely contained in either $G_{1}$ or $G_{2}$.
  Since $G_2$ is $C_3$-free, it contains at most two co-components of $G[N(u_{1})\cap N(u_{2})]$.
  By symmetry, $G_{2}$ contains at most two co-components of $G - N[u_1, u_2]$.
(If a co-component of $G[N(u_{1})\cap N(u_{2})]$ or $G - N[u_1, u_2]$ is not an independent set, then it must be in $G_2$.)
  In one of the iterations of step~3.3, $I$ and $J$ comprise the indices of the co-components of, respectively, $G[N(u_{1})\cap N(u_{2})]$ and $G - N[u_1, u_2]$ in $G_{2}$.
  Then the two sets defined in steps~3.3.1 and~3.3.2 are
  \begin{align*}
    X &= N(u_{1})\cap N(u_{2})\cap V(G_{1}) = N(u_{1})\cap N(u_{2})\cap A,
    \\
    Y &= V(G_{2})\setminus N[u_{1}, u_{2}] = A\cap V(G_{2}).    
  \end{align*}
  It remains to deal with vertices in $N(u_{1})\Delta N(u_{2})$.
  Note that each vertex in $N(u_{1})\Delta N(u_{2})\setminus \{u_1, u_2\}$ is adjacent to precisely one of $u_1$ and $u_2$ in $S(G, A)$.  Thus, $N(u_{1})\Delta N(u_{2})\subseteq V_{1}\cup V_{2}$, and they are adjacent to $u_{3}$ in $S(G, A)$.
If $X = \emptyset$, then for any vertex $u_{3}\in X\subseteq A$, 
    \[
      A \cap (N(u_{1})\Delta N(u_{2})) = ( N(u_{1}) \Delta N(u_{2}) )\cap N(u_3).
    \]
    Otherwise, for any vertex $u_{3}\in V(G) \setminus ( N[u_1, u_2]\cup Y) = \left(  \bigcup_{i = 3}^{r} V_{i} \right) \setminus A$, 
    \[
      A \cap (N(u_{1})\Delta N(u_{2})) = ( N(u_{1}) \Delta N(u_{2}) )\setminus N(u_3).
    \]
    In either case, Step~3.3.5 always returns ``\YES.''

  The algorithm takes $O(n^4(m+n))$ time.
\end{proof}

 \subsection{$\{K_{1,p}, \overline{K_{1,q}}\}$-free graphs}

This section deals with $\{K_{1,p}, \overline{K_{1,q}}\}$-free graphs.  Since it is trivial when one of $p$ and $q$ is one, we assume throughout that $p,q\geq 2$. 
For a pair of positive integers $p$ and $q$, a graph is a \emph{$(p, q)$-split graph} if its vertex set can be partitioned into $S$ and $T$ such that
$G[S]$ is $K_{p+1}$-free and $G[T]$ is $\overline{K_{q+1}}$-free.
The partition $(S, T)$ is called a \emph{$(p,q)$-split partition} of $G$.
Note that $(1, 1)$-split graphs are precisely split graphs.
The key observation is that if $G$ is a yes-instance, then for any vertex $u$, both subgraphs $G[N[u]]$ and $G - N[u]$ are $(p - 1, q - 1)$-split graphs.

\begin{proposition} [\cite{DBLP:conf/fsttcs/KolayP15, DBLP:journals/algorithmica/AntonyGPSSS22}] \label{star, co-star: prop1.1}
For any pair of fixed positive integers $p$ and $q$, we can decide in polynomial time whether a graph $G$ is a $(p,q)$-split graph, and if the answer is yes, produce in polynomial time all $(p,q)$-split partitions of $G$.
\end{proposition}

\begin{figure}[h!]
  \centering 
  \begin{tikzpicture}
    \path (0,0) node[text width=.85\textwidth, inner xsep=20pt, inner ysep=10pt] (a) {
      \begin{minipage}[t!]{\textwidth}
        \begin{tabbing}
          AAA\=Aaa\=aaa\=Aaa\=MMMMMAAAAAAAAAAAA\=A \kill
          1. \> \textbf{if} $G$ is $\{K_{1,p}, \overline{K_{1,q}}\}$-free \textbf{then return} ``\YES'';
          \\
          2. \> fix an arbitrary vertex $u$;
          \\
          3. \> \textbf{if} $G[N[u]]$ or $G - N[u]$ is not a $(q-1, p-1 )$-split graph \textbf{then return} ``\NO'';
          \\
          4. \> \textbf{for each} $(q-1, p-1)$-split partition $(S_1, T_1)$ of $G[N[u]]$ \textbf{do}
          \\
          4.1. \>\> \textbf{for each} $(q-1, p-1)$-split partition $(S_2, T_2)$ of $G - N[u]$ \textbf{do}
          \\
          4.1.1. \>\>\> \textbf{if} $S(G, T_1\cup S_2)$ is $\{K_{1,p}, \overline{K_{1,q}}\}$-free \textbf{then return} ``\YES'';
          \\
          5. \> \textbf{return} ``\NO.''
        \end{tabbing}
      \end{minipage}
    };
    \draw[draw=gray!60] (a.north west) -- (a.north east) (a.south west) -- (a.south east);
  \end{tikzpicture}
  \caption{The algorithm for $\{K_{1,p}, \overline{K_{1,q}}\}$-free graphs.}
  \label{alg:star,co-star}
\end{figure}

\begin{theorem}
  \label{star, co-star: thm1}
  For any pair of integers $p,q\geq 2$, we can decide in polynomial time whether a graph can be switched to a $\{K_{1,p}, \overline{K_{1,q}} \}$-free graph.
\end{theorem}
\begin{proof}
  We use the algorithm described in Figure~\ref{alg:star,co-star}.
  The algorithm returns ``yes'' only when a solution is verified.  Thus, it suffices to show that it always returns ``yes'' for a yes-instance.
  Step~1 takes care of the trivial case, when the input graph $G$ is $\{K_{1,p}, \overline{K_{1,q}}\}$-free.
  Henceforth, the input graph $G$ is not $\{K_{1,p}, \overline{K_{1,q}}\}$-free.
  Let $A$ be a solution containing the vertex $u$ chosen in step~2.

  Since $S(G, A)$ is $K_{1,p}$-free, there cannot be an independent set of size $p$ in $ N[u]\cap {A}$ or $V(G)\setminus (A\cup N[u])$.
  Likewise, since $S(G, A)$ is $\overline{K_{1,q}}$-free, there cannot be a clique of size $q$ in $ N[u]\setminus A$ or $A\setminus N[u]$.
  Thus,
  \begin{itemize}
  \item $(N[u]\setminus A, N[u]\cap A)$ is a $(q-1, p-1 )$-split partition of $G[N[u]]$, and
  \item   $(A\setminus N[u], V(G)\setminus (A\cup N[u]))$ is a $(q-1, p-1 )$-split partition of $G - {N[u]}$.
  \end{itemize}  
  Step~4 tries all $(q-1, p-1 )$-split partitions of $G[N[u]]$ and $G - {N[u]}$.
In one of the iterations,
  \begin{align*}
    T_{1} &= N[u]\cap A,
    \\
    S_{2} &= A\setminus N[u],
  \end{align*}
and hence $A = T_1 \cup S_2$.  Thus, the algorithm must return ``\YES'' in this iteration.
  
  The main work is done in step~4.
  By Proposition~\ref{star, co-star: prop1.1}, there are a polynomial number of iterations, and each iteration can be conducted in polynomial time.  Thus, the total time is polynomial.
\end{proof}

The algorithm in Figure~\ref{alg:star,co-star} is also applicable to any subclass of $\{K_{1,p}, \overline{K_{1,q}} \}$-free graphs that can be recognized in polynomial-time.   In particular, it can handle $\{K_{p}, K_{1,q} \}$-free graphs as well as $\{\overline{K_{p}}, \overline{K_{1,q}}\}$-free graphs.

\begin{theorem}
  \label{star,co-star: thm1}
  We can decide in polynomial time whether a graph can be switched to a $\{K_{1,p}, \overline{K_{1,q}} \}$-free graph.
\end{theorem}

\subsection{Bipartite chain graphs}\label{sec:bipartite chain}

A bipartite graph on bipartition $L\cup R$ is a \emph{bipartite chain} graph if the neighborhoods of the vertices in $L$ can be ordered linearly with respect to inclusion.
Yannakakis~\cite{yannakakis-82-partial-order-dimension} show that the forbidden induced subgraphs of bipartite chain graphs are $C_{3}, \overline{C_{4}}$, and $C_{5}$.
Hage et al.~\cite{DBLP:journals/fuin/HageHW03} developed an algorithm for the upper bipartite switching class, which, however, does not imply algorithms for subclasses of bipartite graphs.
Indeed, a graph may have an exponential number of solutions if we want to switch it to a bipartite graph; e.g., a complete bipartite graph.

\begin{theorem}\label{lem:bipartite chain}
  A graph is in the upper bipartite chain switching class if and only if it is $\{\overline{C_{4}}, K_{3} + K_{1}, K_4\}$-free and is in the upper bipartite switching class.
\end{theorem}
\begin{proof}
  Since both $K_{3} + K_{1}$ and $K_4$ contains $K_{3}$, the necessity is trivial.
  Now suppose that a graph $G$ is $\{\overline{C_{4}}, K_{3} + K_{1}, K_4\}$-free and is in the upper bipartite switching class.
  Since $K_{3} + K_{1}$ and $K_4$ are the only graphs switching equivalent to $\overline{C_{4}}$, any switching equivalent graph of $G$ is $\overline{C_{4}}$-free.
  Since $G$ is in the upper bipartite switching class, it can be switched to a bipartite graph.  This graph must be a bipartite chain graph.
\end{proof}

We can check the first condition of Lemma~\ref{lem:bipartite chain} by enumerating all induced subgraphs on four vertices, and the second by calling the algorithm of Hage et al.~\cite{DBLP:journals/fuin/HageHW03}. 

\begin{corollary}
  We can decide in polynomial time whether a graph can be switched to a bipartite chain graph.
\end{corollary}

We end this section with the following remark. By Proposition~\ref{pro:switching classes}(\ref{pro:switching classes:item:c}), we know that recognizing \LW{\mathcal{G}} is 
polynomially equivalent to recognizing \UP{\mathcal{G}^c}. This implies polynomial-time algorithms for \UP{\mathcal{G}^c} for all the classes $\mathcal{G}$ for which we proved (in Section~\ref{sec:lower}) the finiteness of \LW{\mathcal{G}} or finiteness of the set of forbidden induced subgraphs of \LW{\mathcal{G}}. For example, this implies that we have polynomial-time algorithms for recognizing \UP{\mathcal{G}}, when $\mathcal{G}$ is any of the following (non-hereditary) graph classes: non-chordal, contains a minor from a fixed set $\mathcal{H}$ (this includes the class non-planar), non-bipartite, non-complete bipartite, non-chordal bipartite, non-weakly chordal, non-distance hereditary, non-comparability, non-co-comparability, non-Meyniel, non-permutation, non-strongly chordal, non-interval, non-proper inerval, non-Ptolemaic, non-block, 
or when non-$\mathcal{G}$ is characterized by a finite set of forbidden induced subgraphs.

 \section{Upper switching classes: hardness}\label{sec:hardness}
In this section we consider recognition problems for \UP{\mathcal{G}}: Given a graph $G$ find whether $G$ can be 
switched to a graph in $\mathcal{G}$. We denote the problem by \textsc{Switching-to-$\mathcal{G}$}. We prove that \SWTF{P_{10}} and \SWTF{C_7} are \NPC\ and cannot be solved in time subexponential in the number of vertices, assuming the Exponential Time Hypothesis (ETH). 

ETH is essentially the conjecture that \TSAT\ cannot be solved in time $2^{o(n)}$-time, where $n$ is the number of variables in the input formula. Under this hypothesis, the Sparsification lemma~\cite{DBLP:journals/jcss/ImpagliazzoPZ01} proves that \TSAT\ cannot be solved even in time $2^{o(n+m)}$-time, where $m$ is the number of clauses in the input formula. To transfer this complexity lower bound to other problems, it is sufficient to provide a \textit{linear reduction} - a polynomial-time reduction in which the size of the resultant instance, under the measure that we are interested in, is a linear function of the size of the input instance. For example, if we obtain a polynomial-time reduction from \TSAT\ to a graph problem $Q$ such that the number of vertices, say $n'$, in the resultant instance of $Q$ is a linear function of $n+m$, then it proves that $Q$ cannot be 
solved in $2^{o(n')}$-time, assuming ETH. We refer to the book~\cite{DBLP:books/sp/CyganFKLMPPS15} for an exposition to these topics. 

Our reductions are from \KSATM. A \KSATM\ instance is a boolean formula $\Phi$ with $n$ variables and $m$ clauses where each clause contains exactly $k$ positive literals (and no negative literals). The objective is to check whether there is a truth assignment to the variables so that there is at least one TRUE literal and at least one FALSE literal in each clause in $\Phi$. It is folklore that the problem is \NPC\ and cannot be solved in subexponential-time assuming ETH. We give here a proof for clarity.

\begin{proposition}[folklore]
    \label{pro:ksatm}
    For every $k\geq 3$,
    \KSATM\ is \NPC. 
    Further, the problem cannot be solved in time $2^{o(n+m)}$, assuming ETH.
\end{proposition}
\begin{proof}
  There is a sequence of linear reductions ( \cite{DBLP:journals/tcs/Rusu19}, Proposition 1 and Table 1) from \TSAT\ to \TSATM. We give a linear reduction from \KaSATM\ to \KSATM, which completes the proof.

  Let $\Phi$ be an instance of \KaSATM. Introduce $k$ new variables $a_1,a_2,\ldots a_k$. 
  Replace every clause $\{x_1, x_2, \ldots, x_{k-1}\}$ with $k$ clauses: $\{x_1, x_2, \ldots$, $x_{k-1}, a_1\}$, $\{x_1, x_2, \ldots, x_{k-1}, a_2\}$, $\{x_1, x_2$, $\ldots, x_{k-1}, a_3\}$, $\ldots$, $\{x_1, x_2, \ldots, x_{k-1}, a_k\}$. 
  In addition to the $km$ such clauses, introduce a new clause $\{a_1, a_2, \ldots, a_k\}$. Let the resultant formula be $\Phi'$. It is straight-forward to verify that there is a truth assignment for the variables in $\Phi$ such that it assigns TRUE to at least one variable and FALSE to at least one variable in $\{x_1, x_2,\ldots, x_{k-1}\}$ if and only if there exists a truth assignment for the variables in $\Phi'$ where at least one variable is assigned TRUE and at least one variable is assigned FALSE in all the $km+1$ clauses in $\Phi'$.
\end{proof}

We make use of the concept of a module in our proofs. A module $M$ in a graph $G$ is a subset of vertices of $G$ such that every pair $u,v$ of vertices in $M$ has the same neighborhood outside $M$, i.e., $N(u)\setminus M = N(v)\setminus M$. The set $V(G)$ and every singleton subset of $V(G)$ are the trivial modules. Every other module is a non-trivial module. 
A graph which has no non-trivial module is known as a prime graph. It is straight-forward to note that every prime graph has at least four vertices and every path on at least four vertices and every cycle on at least five vertices is a prime graph.
The following is a simple observation about induced prime graphs in a graph.

\begin{observation}
    \label{obs:module:prime}
    Let $H$ be a prime graph. If a graph $G$ has a subset $V'$ of its vertices such that $V'$ induces $H$ in $G$, then either $V'$ is a subset of a non-trivial module or $|V'\cap M|=1$ for every non-trivial module $M$ of $G$. 
\end{observation}

\subsection{Path}
\label{sub:p10}

We use the following construction for a reduction from \FiSATM\ to \SWTF{P_{10}}.
\begin{construction}
\label{cons p10-free}
     Let $\Phi$ be a \FiSATM\ formula with $n$ variables $X_1, X_2, \cdots, X_n$, and $m$ clauses $C_1, C_2, \cdots,$ $C_m$. 
    We construct a graph $G_{\Phi}$ as follows:
    \begin{itemize}
        \item  For each variable $X_i$ in $\Phi$, introduce a variable vertex $x_i$.  Let $L$ be the set of all variable vertices, which forms an independent set of size $n$.
        
        \item For each clause $C_i$ in $\Phi$ of the form  $\{\ell_{i1}, \ell_{i2}, \ell_{i3}, \ell_{i4}, \ell_{i5}\}$, 
        introduce a set of clause vertices, also named $C_i$, consisting of an independent set of size 5, denoted by $I_i$, and  5 disjoint $P_9$s each of which is denoted by  $B_{ij}$, for $1\leq j\leq 5$. Let $B_i=\bigcup_{j=1}^{5} B_{ij}$. 
        The adjacency among the set $B_{ij}$ and $I_i$, for  $1\leq j\leq 5$, is in such a way that the set of vertices in the $P_9$ induced by the $B_{ij}$, except one of the end vertex $v_{ij}$, is complete to $I_i$. 
        Note that $C_i$ =  $B_i\cup I_i$.
        The set of union of all clause vertices is denoted by $C$. 
        Let the 5 vertices introduced (in the previous step)      for the variables $\ell_{i1}, \ell_{i2}, \ell_{i3}, \ell_{i4},
        \ell_{i5}$ be denoted by     $L_i=\{x_{i1},x_{i2},x_{i3},x_{i4}, x_{i5}\}$.
        Make the adjacency between the vertices in $L_i$ and the sets of $P_9$s in $B_{i}$s in such a way that, taking one vertex from each set $B_{ij}$ along with the variable vertices in $L_i$ induces a $P_{10}$, where the vertices in $L_i$ correspond to an independent set of size 5 in $P_{10}$. More precisely, $x_{i1}$ is complete
        to $B_{i1}$ and $x_{ij}$ is complete to $B_{i(j-1)}\cup B_{ij}$, for $2\leq j\leq 5$.
        Further, make the adjacency among the set $I_i$ and $L_i$ in such a way that, if exactly one of the set $L_i$ or $I_i$ is in the switching set $A$, then the vertices in $L_i\cup I_i$ together induce a $P_{10}$ in $S(G_{\Phi},A)$.     
     
        \item For all $i \neq j$, $C_{i}$ is complete to $C_j$.
     
     \end{itemize}
     
      This completes the construction of the graph $G_{\Phi}$ (see Figure~\ref{fig:cons p10-free} for an example of the construction and Figure~\ref{fig:switch p10-free adjacency} for the adjacency between the sets $L_i$ and $I_i$). 
    
\end{construction}

\begin{figure}[ht]
  \centering
    \begin{tikzpicture}  [myv/.style={circle, draw, inner sep=0pt},myv1/.style={rectangle, draw,inner sep=1pt},myv2/.style={rectangle, draw,inner sep=2.5pt},my/.style={rectangle, draw,dashed,inner sep=0pt},myv3/.style={circle, draw, inner sep=0.25pt},scale=0.85, myv5/.style={draw, inner sep=1pt}] 

\tikzset{
dotted_block/.style={draw=block,  dash pattern=on 3pt off 2pt, rectangle}}

     \node (wa1) at (-2,3.5)  {\resizebox{0.04\textwidth}{!}{\begin{tikzpicture}[myv/.style={circle, draw, inner sep=1.5pt},myv1/.style={circle, draw, inner sep=1.5pt, white},myv2/.style={rectangle,dotted,line width=0.5mm, draw,inner ysep=2.5pt,inner xsep=2pt},myv3/.style={rectangle,dotted,line width=0.5mm, draw,inner ysep=6.8pt,inner xsep=2pt}]
  \node (z) at (0,0) {};

  \node[myv] (a) at (-0.3,0.6) {};
  \node[myv] (j) at (0.3,0.6) {};
  \node[myv] (b) at (-0.3,0) {};
  \node[myv] (c) at (-0.3,0.3) {};
  \node[myv] (d) at (0.3,0.3) {};
  \node[myv] (e) at (0.3,0) {};
  \node[myv] (f) at (0.3,-0.3) {};
  \node[myv] (g) at (-0.3,-0.3) {};
  \node[myv] (h) at (0.6,-0.3) {};
  \node[myv1] (i) at (-0.6,-0.3) {}; 
\node[myv2][fit=(a) (f) (g)] {}; 
\node[myv3][fit=(a)(i) (f) (h)] {}; 
  \draw (a) -- (c);
  \draw (a) -- (j);
  \draw (b) -- (c);
  \draw (j) -- (d);
  \draw (d) -- (e);
  \draw (f) -- (e);
  \draw (b) -- (g);
  \draw (h) -- (f);
  
\end{tikzpicture}} };
     \node (x) at ((-2,4.15)  {\small{$B_{11}$}};

     \node  (wa2) at (-1,3.5) {\resizebox{0.04\textwidth}{!}{\begin{tikzpicture}[myv/.style={circle, draw, inner sep=1.5pt},myv1/.style={circle, draw, inner sep=1.5pt, white},myv2/.style={rectangle,dotted,line width=0.5mm, draw,inner ysep=2.5pt,inner xsep=2pt},myv3/.style={rectangle,dotted,line width=0.5mm, draw,inner ysep=6.8pt,inner xsep=2pt}]
  \node (z) at (0,0) {};

  \node[myv] (a) at (-0.3,0.6) {};
  \node[myv] (j) at (0.3,0.6) {};
  \node[myv] (b) at (-0.3,0) {};
  \node[myv] (c) at (-0.3,0.3) {};
  \node[myv] (d) at (0.3,0.3) {};
  \node[myv] (e) at (0.3,0) {};
  \node[myv] (f) at (0.3,-0.3) {};
  \node[myv] (g) at (-0.3,-0.3) {};
  \node[myv] (h) at (0.6,-0.3) {};
  \node[myv1] (i) at (-0.6,-0.3) {}; 
\node[myv2][fit=(a) (f) (g)] {}; 
\node[myv3][fit=(a)(i) (f) (h)] {}; 
  \draw (a) -- (c);
  \draw (a) -- (j);
  \draw (b) -- (c);
  \draw (j) -- (d);
  \draw (d) -- (e);
  \draw (f) -- (e);
  \draw (b) -- (g);
  \draw (h) -- (f);
  
\end{tikzpicture}} };
     \node (x) at ((-1,4.15)  {\small{$B_{12}$}};
    
     \node  (wa3) at (0,3.5)  {\resizebox{0.04\textwidth}{!}{\begin{tikzpicture}[myv/.style={circle, draw, inner sep=1.5pt},myv1/.style={circle, draw, inner sep=1.5pt, white},myv2/.style={rectangle,dotted,line width=0.5mm, draw,inner ysep=2.5pt,inner xsep=2pt},myv3/.style={rectangle,dotted,line width=0.5mm, draw,inner ysep=6.8pt,inner xsep=2pt}]
  \node (z) at (0,0) {};

  \node[myv] (a) at (-0.3,0.6) {};
  \node[myv] (j) at (0.3,0.6) {};
  \node[myv] (b) at (-0.3,0) {};
  \node[myv] (c) at (-0.3,0.3) {};
  \node[myv] (d) at (0.3,0.3) {};
  \node[myv] (e) at (0.3,0) {};
  \node[myv] (f) at (0.3,-0.3) {};
  \node[myv] (g) at (-0.3,-0.3) {};
  \node[myv] (h) at (0.6,-0.3) {};
  \node[myv1] (i) at (-0.6,-0.3) {}; 
\node[myv2][fit=(a) (f) (g)] {}; 
\node[myv3][fit=(a)(i) (f) (h)] {}; 
  \draw (a) -- (c);
  \draw (a) -- (j);
  \draw (b) -- (c);
  \draw (j) -- (d);
  \draw (d) -- (e);
  \draw (f) -- (e);
  \draw (b) -- (g);
  \draw (h) -- (f);
  
\end{tikzpicture}} };
     \node (x) at ((0,4.15)  {\small{$B_{13}$}};
   
     \node  (wa4) at (1,3.5)  {\resizebox{0.04\textwidth}{!}{\begin{tikzpicture}[myv/.style={circle, draw, inner sep=1.5pt},myv1/.style={circle, draw, inner sep=1.5pt, white},myv2/.style={rectangle,dotted,line width=0.5mm, draw,inner ysep=2.5pt,inner xsep=2pt},myv3/.style={rectangle,dotted,line width=0.5mm, draw,inner ysep=6.8pt,inner xsep=2pt}]
  \node (z) at (0,0) {};

  \node[myv] (a) at (-0.3,0.6) {};
  \node[myv] (j) at (0.3,0.6) {};
  \node[myv] (b) at (-0.3,0) {};
  \node[myv] (c) at (-0.3,0.3) {};
  \node[myv] (d) at (0.3,0.3) {};
  \node[myv] (e) at (0.3,0) {};
  \node[myv] (f) at (0.3,-0.3) {};
  \node[myv] (g) at (-0.3,-0.3) {};
  \node[myv] (h) at (0.6,-0.3) {};
  \node[myv1] (i) at (-0.6,-0.3) {}; 
\node[myv2][fit=(a) (f) (g)] {}; 
\node[myv3][fit=(a)(i) (f) (h)] {}; 
  \draw (a) -- (c);
  \draw (a) -- (j);
  \draw (b) -- (c);
  \draw (j) -- (d);
  \draw (d) -- (e);
  \draw (f) -- (e);
  \draw (b) -- (g);
  \draw (h) -- (f);
  
\end{tikzpicture}} };
     \node (x) at ((1,4.15)  {\small{$B_{14}$}};

    \node  (wa5) at (2,3.5)  {\resizebox{0.04\textwidth}{!}{\begin{tikzpicture}[myv/.style={circle, draw, inner sep=1.5pt},myv1/.style={circle, draw, inner sep=1.5pt, white},myv2/.style={rectangle,dotted,line width=0.5mm, draw,inner ysep=2.5pt,inner xsep=2pt},myv3/.style={rectangle,dotted,line width=0.5mm, draw,inner ysep=6.8pt,inner xsep=2pt}]
  \node (z) at (0,0) {};

  \node[myv] (a) at (-0.3,0.6) {};
  \node[myv] (j) at (0.3,0.6) {};
  \node[myv] (b) at (-0.3,0) {};
  \node[myv] (c) at (-0.3,0.3) {};
  \node[myv] (d) at (0.3,0.3) {};
  \node[myv] (e) at (0.3,0) {};
  \node[myv] (f) at (0.3,-0.3) {};
  \node[myv] (g) at (-0.3,-0.3) {};
  \node[myv] (h) at (0.6,-0.3) {};
  \node[myv1] (i) at (-0.6,-0.3) {}; 
\node[myv2][fit=(a) (f) (g)] {}; 
\node[myv3][fit=(a)(i) (f) (h)] {}; 
  \draw (a) -- (c);
  \draw (a) -- (j);
  \draw (b) -- (c);
  \draw (j) -- (d);
  \draw (d) -- (e);
  \draw (f) -- (e);
  \draw (b) -- (g);
  \draw (h) -- (f);
  
\end{tikzpicture}} };
    \node (x) at ((2,4.15)  {\small{$B_{15}$}};
    ;

     \node (wa6) at (4.5,3.5)  {\resizebox{0.04\textwidth}{!}{\begin{tikzpicture}[myv/.style={circle, draw, inner sep=1.5pt},myv1/.style={circle, draw, inner sep=1.5pt, white},myv2/.style={rectangle,dotted,line width=0.5mm, draw,inner ysep=2.5pt,inner xsep=2pt},myv3/.style={rectangle,dotted,line width=0.5mm, draw,inner ysep=6.8pt,inner xsep=2pt}]
  \node (z) at (0,0) {};

  \node[myv] (a) at (-0.3,0.6) {};
  \node[myv] (j) at (0.3,0.6) {};
  \node[myv] (b) at (-0.3,0) {};
  \node[myv] (c) at (-0.3,0.3) {};
  \node[myv] (d) at (0.3,0.3) {};
  \node[myv] (e) at (0.3,0) {};
  \node[myv] (f) at (0.3,-0.3) {};
  \node[myv] (g) at (-0.3,-0.3) {};
  \node[myv] (h) at (0.6,-0.3) {};
  \node[myv1] (i) at (-0.6,-0.3) {}; 
\node[myv2][fit=(a) (f) (g)] {}; 
\node[myv3][fit=(a)(i) (f) (h)] {}; 
  \draw (a) -- (c);
  \draw (a) -- (j);
  \draw (b) -- (c);
  \draw (j) -- (d);
  \draw (d) -- (e);
  \draw (f) -- (e);
  \draw (b) -- (g);
  \draw (h) -- (f);
  
\end{tikzpicture}} };
     \node (x) at ((4.5,4.15)  {\small{$B_{21}$}};

     \node  (wa7) at (5.5,3.5) {\resizebox{0.04\textwidth}{!}{\begin{tikzpicture}[myv/.style={circle, draw, inner sep=1.5pt},myv1/.style={circle, draw, inner sep=1.5pt, white},myv2/.style={rectangle,dotted,line width=0.5mm, draw,inner ysep=2.5pt,inner xsep=2pt},myv3/.style={rectangle,dotted,line width=0.5mm, draw,inner ysep=6.8pt,inner xsep=2pt}]
  \node (z) at (0,0) {};

  \node[myv] (a) at (-0.3,0.6) {};
  \node[myv] (j) at (0.3,0.6) {};
  \node[myv] (b) at (-0.3,0) {};
  \node[myv] (c) at (-0.3,0.3) {};
  \node[myv] (d) at (0.3,0.3) {};
  \node[myv] (e) at (0.3,0) {};
  \node[myv] (f) at (0.3,-0.3) {};
  \node[myv] (g) at (-0.3,-0.3) {};
  \node[myv] (h) at (0.6,-0.3) {};
  \node[myv1] (i) at (-0.6,-0.3) {}; 
\node[myv2][fit=(a) (f) (g)] {}; 
\node[myv3][fit=(a)(i) (f) (h)] {}; 
  \draw (a) -- (c);
  \draw (a) -- (j);
  \draw (b) -- (c);
  \draw (j) -- (d);
  \draw (d) -- (e);
  \draw (f) -- (e);
  \draw (b) -- (g);
  \draw (h) -- (f);
  
\end{tikzpicture}} };
     \node (x) at ((5.5,4.15)  {\small{$B_{22}$}};
    
     \node  (wa8) at (6.5,3.5)  {\resizebox{0.04\textwidth}{!}{\begin{tikzpicture}[myv/.style={circle, draw, inner sep=1.5pt},myv1/.style={circle, draw, inner sep=1.5pt, white},myv2/.style={rectangle,dotted,line width=0.5mm, draw,inner ysep=2.5pt,inner xsep=2pt},myv3/.style={rectangle,dotted,line width=0.5mm, draw,inner ysep=6.8pt,inner xsep=2pt}]
  \node (z) at (0,0) {};

  \node[myv] (a) at (-0.3,0.6) {};
  \node[myv] (j) at (0.3,0.6) {};
  \node[myv] (b) at (-0.3,0) {};
  \node[myv] (c) at (-0.3,0.3) {};
  \node[myv] (d) at (0.3,0.3) {};
  \node[myv] (e) at (0.3,0) {};
  \node[myv] (f) at (0.3,-0.3) {};
  \node[myv] (g) at (-0.3,-0.3) {};
  \node[myv] (h) at (0.6,-0.3) {};
  \node[myv1] (i) at (-0.6,-0.3) {}; 
\node[myv2][fit=(a) (f) (g)] {}; 
\node[myv3][fit=(a)(i) (f) (h)] {}; 
  \draw (a) -- (c);
  \draw (a) -- (j);
  \draw (b) -- (c);
  \draw (j) -- (d);
  \draw (d) -- (e);
  \draw (f) -- (e);
  \draw (b) -- (g);
  \draw (h) -- (f);
  
\end{tikzpicture}} };
     \node (x) at ((6.5,4.15)  {\small{$B_{23}$}};
   
     \node  (wa9) at (7.5,3.5)  {\resizebox{0.04\textwidth}{!}{\begin{tikzpicture}[myv/.style={circle, draw, inner sep=1.5pt},myv1/.style={circle, draw, inner sep=1.5pt, white},myv2/.style={rectangle,dotted,line width=0.5mm, draw,inner ysep=2.5pt,inner xsep=2pt},myv3/.style={rectangle,dotted,line width=0.5mm, draw,inner ysep=6.8pt,inner xsep=2pt}]
  \node (z) at (0,0) {};

  \node[myv] (a) at (-0.3,0.6) {};
  \node[myv] (j) at (0.3,0.6) {};
  \node[myv] (b) at (-0.3,0) {};
  \node[myv] (c) at (-0.3,0.3) {};
  \node[myv] (d) at (0.3,0.3) {};
  \node[myv] (e) at (0.3,0) {};
  \node[myv] (f) at (0.3,-0.3) {};
  \node[myv] (g) at (-0.3,-0.3) {};
  \node[myv] (h) at (0.6,-0.3) {};
  \node[myv1] (i) at (-0.6,-0.3) {}; 
\node[myv2][fit=(a) (f) (g)] {}; 
\node[myv3][fit=(a)(i) (f) (h)] {}; 
  \draw (a) -- (c);
  \draw (a) -- (j);
  \draw (b) -- (c);
  \draw (j) -- (d);
  \draw (d) -- (e);
  \draw (f) -- (e);
  \draw (b) -- (g);
  \draw (h) -- (f);
  
\end{tikzpicture}} };
     \node (x) at ((7.5,4.15)  {\small{$B_{24}$}};

    \node  (wa10) at (8.5,3.5)  {\resizebox{0.04\textwidth}{!}{\begin{tikzpicture}[myv/.style={circle, draw, inner sep=1.5pt},myv1/.style={circle, draw, inner sep=1.5pt, white},myv2/.style={rectangle,dotted,line width=0.5mm, draw,inner ysep=2.5pt,inner xsep=2pt},myv3/.style={rectangle,dotted,line width=0.5mm, draw,inner ysep=6.8pt,inner xsep=2pt}]
  \node (z) at (0,0) {};

  \node[myv] (a) at (-0.3,0.6) {};
  \node[myv] (j) at (0.3,0.6) {};
  \node[myv] (b) at (-0.3,0) {};
  \node[myv] (c) at (-0.3,0.3) {};
  \node[myv] (d) at (0.3,0.3) {};
  \node[myv] (e) at (0.3,0) {};
  \node[myv] (f) at (0.3,-0.3) {};
  \node[myv] (g) at (-0.3,-0.3) {};
  \node[myv] (h) at (0.6,-0.3) {};
  \node[myv1] (i) at (-0.6,-0.3) {}; 
\node[myv2][fit=(a) (f) (g)] {}; 
\node[myv3][fit=(a)(i) (f) (h)] {}; 
  \draw (a) -- (c);
  \draw (a) -- (j);
  \draw (b) -- (c);
  \draw (j) -- (d);
  \draw (d) -- (e);
  \draw (f) -- (e);
  \draw (b) -- (g);
  \draw (h) -- (f);
  
\end{tikzpicture}} };
    \node (x) at ((8.5,4.15)  {\small{$B_{25}$}};
    ;

     \node[myv3] (i1) at (-2,5.5) {$I_{11}$};
     \node[myv3] (i2) at (-1,5.5) {$I_{12}$};
     \node[myv3] (i3) at (0,5.5) {$I_{13}$};
     \node[myv3] (i4) at (1,5.5) {$I_{14}$};
     \node[myv3] (i5) at (2,5.5) {$I_{15}$};

     \node[myv3] (j1) at (4.5,5.5) {$I_{21}$};
     \node[myv3] (j2) at (5.5,5.5) {$I_{22}$};
     \node[myv3] (j3) at (6.5,5.5) {$I_{23}$};
     \node[myv3] (j4) at (7.5,5.5) {$I_{24}$};
     \node[myv3] (j5) at (8.5,5.5) {$I_{25}$};

    \node[myv3] (x1) at (-2,1.25) {$x_{1}$};
 
    \node[myv3] (x2) at (-0.5,1.25) {${x_{2}}$};

    \node[myv3] (x3) at (1,1.25) {$x_{3}$};

    \node[myv3] (x4) at (2.5,1.25) {${x_{4}}$};
    
    \node[myv3] (x5) at (4,1.25) {${x_{5}}$};

    \node[myv3] (x6) at (5.5,1.25) {${x_{6}}$};

    \node[myv3] (x7) at (7,1.25) {${x_{7}}$};

    \node[myv3] (x8) at (8.5,1.25) {${x_{8}}$};

    \node[my](b)[fit=(wa1) (wa5),  inner xsep=0.1ex, inner ysep=1.2ex, label=left:$B_{1}$] {}; 
 
    \node[my] (i) [fit=(i1) (i5),  inner xsep=0.5ex, inner ysep=0.5ex, label=left:$I_{1}$] {}; 

     \node[my] (j) [fit=(j1) (j5),  inner xsep=0.5ex, inner ysep=0.5ex, label=right:$I_{2}$] {}; 
    
    \node[my] (x) [fit=(x1) (x5),  inner xsep=1.5ex, inner ysep=0.5ex, label=left:$L_{1}$] {}; 
    \node[my] (y) [fit=(wa1) (i) (wa5),  inner xsep=2.5ex, inner ysep=1.75ex, label=above:$C_{1}$] {};

    \draw [double,line width=0.3mm] (-2.4,5.25) to[bend right=80]  (-2.5,1.3);

    \node[my](b)[fit=(wa6) (wa10),  inner xsep=0.25ex, inner ysep=1.2ex, label=right:$B_{2}$] {};

    \node[my][dotted,line width=0.4mm] (x) [fit=(x4) (x8),  inner xsep=1.5ex, inner ysep=0.5ex, label=right:$L_{2}$] {}; 
    \node[my] (z) [fit=(wa6) (j) (wa10),  inner xsep=2.25ex, inner ysep=1.75ex, label=above:$C_{2}$] {};

    \draw [line width=0.5mm] (y)--(z);
   
    \draw [double,line width=0.3mm] (8.9,5.25) to[bend left=80]  (9,1.3);

  \draw [line width=0.3mm](-2,3.2) -- (x1);
  \draw [line width=0.3mm](-2,3.2) -- (x2);
  \draw [line width=0.3mm](-1,3.2) -- (x2);
  \draw [line width=0.3mm](-1,3.2) -- (x3);
  \draw [line width=0.3mm](0,3.2) -- (x3);
  \draw [line width=0.3mm](0,3.2) -- (x4);
  \draw [line width=0.3mm](1,3.2) -- (x4);
  \draw [line width=0.3mm](1,3.2) -- (x5);
  \draw [line width=0.3mm](2,3.2) -- (x5);

  \draw [line width=0.3mm](4.5,3.2) -- (x4);
  \draw [line width=0.3mm](4.5,3.2) -- (x5);
  \draw [line width=0.3mm](5.5,3.2) -- (x5);
  \draw [line width=0.3mm](5.5,3.2) -- (x6);
  \draw [line width=0.3mm](6.5,3.2) -- (x6);
  \draw [line width=0.3mm](6.5,3.2) -- (x7);
  \draw [line width=0.3mm](7.5,3.2) -- (x7);
  \draw [line width=0.3mm](7.5,3.2) -- (x8);
  \draw [line width=0.3mm](8.5,3.2) -- (x8);
  
  \draw [line width=0.3mm](i) -- (-2,3.7);
  \draw [line width=0.3mm](i) -- (-1,3.7);
  \draw [line width=0.3mm](i) -- (0,3.7);
  \draw [line width=0.3mm](i) -- (1,3.7);
  \draw [line width=0.3mm](i) -- (2,3.7);

   \draw [line width=0.3mm](j) -- (4.5,3.7);
  \draw [line width=0.3mm](j) -- (5.5,3.7);
  \draw [line width=0.3mm](j) -- (6.5,3.7);
  \draw [line width=0.3mm](j) -- (7.5,3.7);
  \draw [line width=0.3mm](j) -- (8.5,3.7);

\end{tikzpicture}  
     \caption{An example of Construction \ref{cons p10-free} with the formula $\Phi = C_1\wedge C_2$,  where $C_1 = \{x_1, x_2, x_3, x_4, x_5\}$ and $C_2 = \{x_4, x_5, x_6, x_7, x_8\}$.
    Single lines connecting two rectangles indicate that 
    each vertex in one rectangle is adjacent to all vertices in the other rectangle. 
    The double line connecting two rectangles indicates that 
    each vertex in one rectangle is adjacent to the vertices in the other rectangle in such a way that 
    if a switching set $A$ contains all the vertices of one rectangle and no vertex of the other rectangle,
    then a $P_{10}$ is induced by these two sets of vertices after switching (Figure~\ref{fig:switch p10-free adjacency} shows this connectivity).}
  \label{fig:cons p10-free}
\end{figure}

\begin{figure}[!htb]
  \centering
    \resizebox{0.5\textwidth}{!}{\begin{tikzpicture}  [myv/.style={circle, draw, inner sep=0pt},myv1/.style={rectangle, draw,inner sep=1pt},myv2/.style={rectangle, draw,inner sep=2.5pt},my/.style={rectangle, draw,dashed,inner sep=0pt},myv3/.style={circle, draw, inner sep=0.25pt},scale=0.85, myv5/.style={circle, draw, inner sep=0.75pt}] 

\tikzset{
dotted_block/.style={draw=block,  dash pattern=on 3pt off 2pt, rectangle}}

     \node[myv3] (i1) at (-0.5,4.5) {$I_{i1}$};
     \node[myv3] (i2) at (2,4.5) {$I_{i2}$};
     \node[myv3] (i3) at (4.5,4.5) {$I_{i3}$};
     \node[myv3] (i4) at (7,4.5) {$I_{i4}$};
     \node[myv3] (i5) at (9.5,4.5) {$I_{i5}$};

    \node[myv5] (x1) at (-0.5,1.5) {$x_{1}$};
 
    \node[myv5] (x2) at (2,1.5) {${x_{2}}$};

    \node[myv5] (x3) at (4.5,1.5) {$x_{3}$};

    \node[myv5] (x4) at (7,1.5) {${x_{4}}$};
    
    \node[myv5] (x5) at (9.5,1.5) {${x_{5}}$};

    \node[my] (i) [fit=(i1) (i5),  inner xsep=0.5ex, inner ysep=0.5ex, label=left:$I_{i}$] {}; 
    
    \node[my] (x) [fit=(x1) (x5),  inner xsep=0.5ex, inner ysep=0.5ex, label=left:$L_{i}$] {};

  \draw [line width=0.2mm](i1) -- (x3);
  \draw [line width=0.2mm](i1) -- (x4);
  \draw [line width=0.2mm](i1) -- (x5);
\draw [line width=0.2mm](i2) -- (x1);
  \draw [line width=0.2mm](i2) -- (x4);
  \draw [line width=0.2mm](i2) -- (x5);
\draw [line width=0.2mm](i3) -- (x1);
  \draw [line width=0.2mm](i3) -- (x2);
  
  \draw [line width=0.2mm](i3) -- (x5);
\draw [line width=0.2mm](i4) -- (x1);
  \draw [line width=0.2mm](i4) -- (x2);
  \draw [line width=0.2mm](i4) -- (x3);
\draw [line width=0.2mm](i5) -- (x1);
  \draw [line width=0.2mm](i5) -- (x2);
  \draw [line width=0.2mm](i5) -- (x3);
  \draw [line width=0.2mm](i5) -- (x4);

\end{tikzpicture} } 
     \caption{The adjacency between $L_i$ and $I_i$ in Construction~\ref{cons p10-free}}
  \label{fig:switch p10-free adjacency}
\end{figure}

We recall that the vertices in $L_i$ and one vertex each from $B_{ij}$s ($1\leq j\leq 5$) induce a $P_{10}$. If we have a truth assignment which satisfies $\Phi$, then the vertices in $L$ corresponding to the TRUE literals can be switched to obtain a $P_{10}$-free graph (Lemma~\ref{lem:p10:forward}). The backward direction is easy and is proved in Lemma~\ref{lem:p10:backward}.

\begin{lemma}\label{lem:p10:backward}
    Let $\Phi$ be an instance of \FiSATM. If $S(G_{\Phi},A)$ is $P_{10}$-free, for some $A\subseteq V(G_{\Phi})$, then there exists a truth assignment satisfying $\Phi$.
\end{lemma}
\begin{proof}
We claim that assigning TRUE to the variables corresponding to the variable vertices in $A\cap L$ satisfies $\Phi$. It is sufficient to prove that 
$A\cap L_i\neq \emptyset$ and $L_i\setminus A\neq \emptyset$, for every $1\leq i\leq m$. 

For a contradiction, assume that $A\cap L_i=\emptyset$, for some $1\leq i\leq m$.
Since $L_i$ and one vertex each from $B_{ij}$ induces a $P_{10}$, we obtain that $B_{ij}\subseteq A$, for some $1\leq j\leq 5$. Then $I_i\subseteq A$ (otherwise, there is a $P_{10}$ induced in $S(G_{\Phi},A)$ by $B_{ij}$ and a vertex in $I_i$ not in $A$ - recall that one end vertex $v_{ij}$ of the $P_9$ formed by $B_{ij}$ is not adjacent to $I_i$). Then at least one vertex from $L_i$ is in $A$, otherwise there is a $P_{10}$ induced in $S(G_{\Phi},A)$ by $I_i\cup L_i$. 
This gives us a contradiction.

Next we show that $L_i$ is not a subset of $A$.
For a contradiction, assume that $L_i\setminus A=\emptyset$. 
Then at least one vertex $I_{i\ell}\in I_i$ (for some $1\leq \ell\leq 5$) is in $A$ - otherwise there is an $P_{10}$ induced in $S(G_{\Phi}, A)$ by $L_i\cup I_i$.
Then at least one vertex from each $B_{ij}$ (for $1\leq j\leq 5$) must be in $A$ - otherwise there is a $P_{10}$ induced in $S(G_{\Phi}, A)$ by $I_{i\ell}$ and $B_{ij}$, where $B_{ij}\cap A=\emptyset$. Then there is a $P_{10}$ induced by $L_i$ and one vertex, which is in $A$, from each $B_{ij}$ (for $1\leq j\leq 5$). This is a contradiction.  
\end{proof}

We next handle the forward direction. 
Now onward we assume that $A$ is a subset of $L$ such that $L_i\cap A\neq \emptyset$ and $L_i\setminus A\neq \emptyset$. Let $G' = S(G_{\Phi}, A)$. With the help of \cref{lem:p10:forward:c,lem:p10:forward:bij,lem:p10:forward:ci,lem:p10:forward:lind,lem:p10:forward:ci2k2,lem:p10:forward:cili}, and Observation~\ref{obs:p10:forward:l}, we prove that $G'$ is $P_{10}$-free. For a contradiction, assume that $R\subseteq V(G')$ induces a $P_{10}$.

\begin{lemma}
    \label{lem:p10:forward:c}
    $R$ is not a subset of $C$.
\end{lemma}
\begin{proof}
    For a contradiction, assume that $R\subseteq C$. 
    We observe that $I_i$, for $1\leq i\leq m$, is a module in $G_{\Phi}[C]$. Therefore, by Observation~\ref{obs:module:prime}, $|R\cap I_i|\leq 1$. The set $B_i$ induces a collection of five $P_9$s. Therefore, $R$ is not a subset of $B_i$. A vertex from $I_i$ and vertices from $B_i$ together cannot induce a path of more than 5 vertices. Therefore $R$ cannot be a subset of $C_i$. We observe that $C_i$ is a module in $G_{\Phi}[C]$. Therefore, by Observation~\ref{obs:module:prime}, $|R\cap C_i|\leq 1$, for $1\leq i\leq m$. Since $C_i$ is complete 
    to $C_j$ for $1\leq i<j\leq m$, we obtain that $R$ is not a subset of $C$. 
\end{proof}

Observation~\ref{obs:p10:forward:l} follows from the fact that switching a subset of vertices of an edgeless graph 
produces a complete bipartite graph.
\begin{observation}
    \label{obs:p10:forward:l}
$R$ is not a subset of $L$.
\end{observation}

\begin{lemma}
\label{lem:p10:forward:bij}
    $|R\cap B_{ij}|\leq 2$, for $1 \leq i \leq m$, $1 \leq j \leq 5 $. If $|R\cap B_{ij}|=2$, then $v_{ij}\in R$.
\end{lemma}

\begin{proof}
    For a contradiction assume that $|R\cap B_{ij}|\geq 3$. 
    Since $B_{ij}$ does not induce a $P_{10}$, we obtain that there is a vertex $u\in R\setminus B_{ij}$ which is adjacent to at least one vertex in $R\cap B_{ij}$.
    If $u\in L\cup C_\ell$, where $\ell\neq i$, then $u$ is adjacent to all the vertices in $R\cap B_{ij}$ creating a claw. 
    Then $u\in I_i$ and the adjacency between $u$ and $R\cap B_{ij}$ does not create a claw only when $|R\cap B_{ij}|=3$ and $v_{ij}\in R$. Further $(R\cap B_{ij})\cup \{u\}$
    induces a either a $P_4$ or a $P_3+K_1$. In either case, $R$ contains one more vertex $w$ from $I_i$ other than $u$. Then $\{u,w\}\cup ((R\cap B_{ij})\setminus \{v_{ij}\})$ induces a $C_4$, which is a contradiction. These arguments also imply that if $|R\cap B_{ij}|=2$, then one of the vertex in $R\cap B_{ij}$ must be $v_{ij}$.
\end{proof}

\begin{lemma}
    \label{lem:p10:forward:ci}
    $R\cap C\subseteq C_i$, for some $1\leq i\leq m$.
\end{lemma}
\begin{proof}
    Recall that $C_i$ is complete to $C_j$ if $i\neq j$. Therefore, if $R$ has nonempty intersection with $C_i, C_j$, and $C_k$, for $i\neq j\neq k$, then there is a triangle in the graph induced by $R$. Now assume that $R$ has nonempty intersection with $C_i$ and $C_j$, where $i\neq j$. If $|R\cap C_i|\geq 2$ and $|R\cap C_j|\geq 2$, then there is an $C_4$ in the graph induced by $R$. Therefore, assume that $|R\cap C_i|=1$ and $1\leq |R\cap C_j|\leq 2$. Then $R\cap C$ induces either a $P_3$ or a $K_2$. Let $v$ be the vertex in $R\cap C_i$.
    Clearly $|R\cap L|\geq 7$. Recall that the vertices of a $P_{10}$ can be partitioned into two independent sets each of size 5. Also, recall that $L$ induces a complete bipartite graph in $G'$. 
    Let $L'$ and $L''$ be the two independent sets which forms a partition of the complete bipartite graph induced by $L$. Then there are exactly 4 vertices in $R$ from one of the partition, say $L'$, which forms an independent set of size 5 along with $v$, and there are at least 3 vertices from $L''$ which forms an independent set of size 5 with $R\cap C_j$. Then there is a $C_4$ formed by two vertices from $R\cap L'$ and two vertices from $R\cap L''$, which gives us a contradiction. 
\end{proof}

\begin{lemma}
    \label{lem:p10:forward:lind}
    If $R\cap C\subseteq C_i$ (for some $1\leq i\leq m$) and $R\cap L$ is an independent set, then $R\cap L\subseteq L_i$.
\end{lemma}
\begin{proof}
Since $R\cap L$ is an independent set, either $R\cap L\subseteq A$ or $R\cap L\cap A = \emptyset$. Since the size of a maximum independent set is 5 in $P_{10}$, we obtain that $|R\cap L|\leq 5$. Therefore, $|R\cap C_i|\geq 5$. For a contradiction assume that $R\cap L_j\neq \emptyset$, for some $j\neq i$. Let $v\in R\cap L_j$. We note that $v$ is either adjacent to all vertices in $R\cap C_i$ or nonadjacent to all vertices in $R\cap C_i$. In the former case, $v$ is a vertex with degree at least 5, and in the later case $v$ is not adjacent to any other vertex in the graph induced by $R$. Both are contradictions. 
\end{proof}

\begin{lemma}
    \label{lem:p10:forward:ci2k2}
    The graph induced by $R\cap C_i$ is $2K_2$-free. Further, if $R\cap C_i$ is not an independent set, then $R\cap I_i\neq \emptyset$ and none of the vertex in $R\cap I_i$ is an isolated vertex in the graph induced by $R\cap C_i$.
\end{lemma}
\begin{proof}
    For a contradiction, assume that $P\subseteq R\cap C_i$ induces a $2K_2$. 
    It is straight-forward to verify that every vertex in $I_i$ is adjacent to at least one end vertex of every $K_2$ induced by the vertices in $C_i$.
    This implies that $P\cap I_i=\emptyset$ and hence $P\subseteq B_i$. By Lemma~\ref{lem:p10:forward:bij} and the fact that $B_{ij}$ and $B_{i\ell}$ are nonadjacent (for $j\neq \ell$), we obtain that $P$ has two vertices from $B_{ij}$ and two vertices form $B_{i\ell}$, for some $j\neq \ell$. Moreover, $P=\{u_j,v_{ij},u_\ell,v_{i\ell}\}$, where $u_j$ is the vertex in $B_{ij}$ adjacent to $v_{ij}$
    and $u_\ell$ is the vertex in $B_{i\ell}$ adjacent
    to $v_{i\ell}$. If $R\cap I_i$ has at least two vertices, then there is a $C_4$ formed by those two vertices and $\{u_j,u_\ell\}$. Therefore, $|R\cap I_i|\leq 1$. Clearly, $R\setminus P$ has at least one vertex adjacent to some vertices in $P$. Recall that every vertex outside $C_i$ is either complete to $B_{ij}$ or nonadjacent to $B_{ij}$ and if such a vertex is complete to $B_{ij}$, then it forms a triangle with $\{u_j,v_{ij}\}$. Therefore, the vertices in $R\setminus P$ which are adjacent to some vertices in $P$ must be from $I_i$. This implies that there is exactly one vertex in $R\cap I_i$ and it forms a $P_5$ with the vertices in $P$. Then there must be at least one more vertex in $R\setminus P$ adjacent to either $v_{ij}$ or $v_{i\ell}$. Since $|(R\setminus P)\cap I_i| = 1$, such a vertex does not exist. This is a contradiction. These arguments also imply that if $R\cap C_i$ is not an independent set, then $R\cap I_i\neq \emptyset$ and none of the vertices in $R\cap L_i$ is an isolated vertex in the graph induced by $R\cap C_i$.
\end{proof}

\begin{lemma}
\label{lem:p10:forward:cili}
    If $R\subseteq (C_i\cup L_i)$, then $R\cap L_i$ is not an independent set in $G'$.
 \end{lemma}
 
 \begin{proof}
    For a contradiction assume that $R\cap L_i$ induces an $xK_1$, for $x\geq 1$.
    It can be easily observed that the following claims are true. 
    
    Claim 1: Each vertex in $I_i$ is adjacent to at least 3 vertices in $L_i$ in $G_{\Phi}$.
    
    Claim 2: Each set $B_{ij}$ is adjacent to at least $|A\cap L_i|-2$ vertices in $A\cap L_i$ in $G'$.

    Claim 3: $R\cap L_i\subseteq A$ or $R\cap L_i\cap A=\emptyset$.

    Claim 1 and 2 follow from construction of $G_{\Phi}$ and Claim 3 follows from the assumption that $R\cap L_i$ forms an independent set and the fact that $L$ induces a complete bipartite graph in $G'$, where the bipartition is $(A, L\setminus A)$.

    The fact that $L_i$ contains 5 vertices and at least one vertex in $L_i$ belongs to $A$, and at least one vertex in $L_i$ does not belong to $A$ implies that the graph induced by $R\cap L_i$ is $5K_1$-free in $G'$.

    Now assume that $R\cap L_i$ induces an $xK_1$, for $x\leq 4$. By Lemma~\ref{lem:p10:forward:ci2k2}, the graph induced by $C_i\cap R$ is $2K_2$-free. Hence, the graph induced by $C_i\cap R$ can have at most 3 edges (observe that a path on 4 vertices does not have an induced $2K_2$). 
    Therefore, there are at least 6 edges between $R\cap L_i$ and $R\cap C_i$. Since every vertex in $R\cap L_i$ can have at most 2 edges in the graph induced by $R$, we obtain that $|R\cap L_i|\geq 3$.

    Now assume that $R\cap L_i$ induces $3K_1$ in $G'$. 
    Then by the above arguments, the graph induced by $R\cap C_i$
    is $P_4+3K_1$, where a vertex $u$ in $R\cap I_i$ is one of the two middle vertex of the induced $P_4$. Among the other 3 vertices of the $P_4$,
    two are from $B_{ij}$ and the other is from $B_{i\ell}$, for some $j\neq \ell$. 
    By Lemma~\ref{lem:p10:forward:bij}, the two vertices in $R\cap B_{ij}$ are $v_{ij}$ and $u_j$, where $u_j$ is the vertex adjacent to $v_{ij}$. 
    Let the other vertex in the $P_4$ be $u_\ell\in B_{i\ell}$. We also have that every vertex in $R\cap L_i$
    has exactly two neighbors in $R\cap C_i$. By Claim 3, either $R\cap L_i\subseteq A$ or $R\cap L_i\cap A=\emptyset$.
    If $R\cap L_i$ has no elements in $A$, then by Claim 1, $u$ is adjacent to at least one 
    of the three vertices in $R\cap L_i$ and hence has a degree 3 
    in the graph induced by $R$, which is a contradiction. Assume that $R\cap L_i\subseteq A$.
    By Claim 2, each vertex in $B_{ij}$ is adjacent to 
    at least one vertex in $R\cap L_i$. Then 
    there is a triangle in the graph induced by $R$, which is a contradiction.
    
    Suppose  $R\cap L_i$ induces an $xK_1$, for $x=4$.
    Since the degree of each vertex in $R\cap L_i$
    is at most 2 in the induced $P_{10}$, there is at least one edge in the graph induced by $R\cap C_i$.
    Therefore by Lemma~\ref{lem:p10:forward:ci2k2}, $R\cap I_i$ is nonempty and every vertex in $R\cap I_i$ is an end point of some edge in the graph induced by $R\cap C_i$.
    
    Now assume that $R\cap L_i\cap A=\emptyset$. By Claim 2, each vertex in $R\cap I_i$ is adjacent to at least three vertices in $L_i$ in $G_{\Phi}$.
    Then each of them is adjacent to at least two vertices in $R\cap L_i$. Then those vertices
    have degree at least 3 in the graph induced by $R$
    in $G'$. This is a contradiction. 

    Now assume that $R\cap L_i\subseteq A$. As shown above, there is a vertex $u_j\in B_{ij}$ (for some $1\leq j\leq 5$) adjacent to some vertex $u\in I_i$
    such that both $u_j$ and $u$ are in $R$. 
    By Claim 2, $u_j$ is adjacent to at least $2$
    vertices in $R\cap L_i$. This implies that
    $u_j$ has degree 3 in the graph induced by $R$
    in $G'$, which is a contradiction.
\end{proof}

 Now, we are ready to prove the forward direction of the reduction.
 
\begin{lemma}
\label{lem:p10:forward}
   Let $\Phi$ be a yes-instance of \FiSATM,  and $\psi$ be a truth assignment satisfying $\Phi$. 
   Let $A$ be the set of variable vertices whose corresponding variables were assigned \TRUE\ by $\psi$. Let $G'$ be $S(G_{\Phi}, A)$.
   Then $G'$ is $P_{10}$-free.
\end{lemma}

\begin{proof}
     Assume for a contradiction that there exists a set $R\subseteq V(G')$ such that $R$ induces 
     a $P_{10}$ in $G'$. By Lemma~\ref{lem:p10:forward:c},
     the graph induced by $C$ is $P_{10}$-free. 
     Therefore, $R\cap L$ is nonempty. By Observation~\ref{obs:p10:forward:l}, the graph induced by $L$ in $G'$ is $P_{10}$-free.
     Therefore, $R\cap C$ is nonempty. By Lemma~\ref{lem:p10:forward:ci}, $R\cap C\subseteq C_i$. Lemmas~\ref{lem:p10:forward:lind} and \ref{lem:p10:forward:cili} imply that the graph 
     induced by $R\cap L$ is not an independent set.
     Then, since $L$ induces a complete bipartite graph in $G'$, $R\cap L$ induces either a $K_2$ or a $P_3$. By Lemma~\ref{lem:p10:forward:ci2k2}, the graph induced by $R\cap C_i$ cannot have more than 3 edges. Therefore, there are at least 4 edges between $R\cap L$ and $R\cap C_i$. This implies that at least one vertex in $R\cap L$
     gets a degree 3 in the graph induced by $R$ in $G'$. This is a contradiction.
\end{proof}

Lemmas~\ref{lem:p10:forward}, \ref{lem:p10:backward}, and Proposition~\ref{pro:ksatm}, and the fact that the number of vertices in $G_{\Phi}$ is linear in the number of variables and clauses in $\Phi$ 
imply Theorem~\ref{thm:p10}.

\begin{theorem}
\label{thm:p10}
\SWTF{P_{10}} is \NPC\ and cannot be solved in 
$2^{o(n)}$-time, assuming ETH, where $n$ is the number of vertices in the input graph.
\end{theorem}
 \subsection{Cycle}
\label{subsec:cycle}
In this section we prove that \SWTF{C_7} is \NPC\
and cannot be solved in subexponential-time, assuming ETH. The reduction is from \TSATM.
\begin{construction}
\label{con:c7}
    Let $\Phi$ be a \TSATM\ formula with $n$ variables $X_1, X_2, \cdots, X_n$, and $m$ clauses $Y_1, Y_2, \ldots,Y_m$. 
    We construct a graph $G_{\Phi}$ as follows.
\begin{itemize}
    \item  For each variable $X_i$ in $\Phi$, introduce a variable vertex $x_i$.  Let $L$ be the set of all variable vertices, which forms an independent set of size $n$.

    \item For each clause $Y_i$ in $\Phi$ of the form  $\{\ell_{i1}, \ell_{i2}, \ell_{i3}\}$, 
    introduce a set of clause vertices, also named $Y_i$, which
    is the union of a set $I_i$, which induces a $K_2+2K_1$, and a set $B_i$. The set $B_i$ is the union of 8 levels of vertices, where the set forming $j$\textsuperscript{th} level is denoted by $B_{ij}$, for $1\leq j\leq 8$. 
    Each set $B_{ij}$ is the union of 4 sets, $B_{ij\ell}$, for $1\leq \ell\leq 4$, where each set $B_{ij\ell}$ induces a $P_{6}$. The two end points of the $P_6$ induced by each $B_{ij\ell}$
    is denoted by $p_{ij\ell}$ and $q_{ij\ell}$ and let $B'_{ij\ell}$ denote $B_{ij\ell}\setminus \{p_{ij\ell}, q_{ij\ell}\}$. The set $B_{i11}$ is complete to $B_{i14}$. 
    The set of all vertices in a level $j$ is denoted by $\beta_j$, i.e., $\beta_j = \bigcup_{i=1}^{m} B_{ij}$, for $1\leq j\leq 8$.
    The set $B'_{ij\ell}$ is complete to $B_{i(j+1)\ell}$, for $1\leq j\leq 7$, $1\leq \ell\leq 4$.
    Similarly, $B'_{i8\ell}$ is complete to $I_i$, for $1\leq \ell\leq 4$.
    The set of union of $I_i$s is denoted by $I$, and the set of union of $Y_i$s is denoted by $Y$. Let the $3$ vertices introduced (in the previous step) 
    for the variables $\ell_{i1}, \ell_{i2}, \ell_{i3}$ be denoted by  
    $L_i=\{x_{i1},x_{i2},x_{i3}\}$. 
    Make the adjacency among the sets $B_{i1\ell}$s, for $1\leq \ell \leq 4$, and $L_i$ in such a way that, taking one vertex from each $B_{i1\ell}$ along with the vertices in
    $L_i$ induces a $C_7$, where the vertices in $L_i$ correspond to an independent set of size $3$ in $C_7$. More precisely, $x_{i\ell}$ is complete
    to $B_{i1\ell}\cup B_{i1(\ell+1)}$, for $1\leq \ell\leq 3$.
    Similarly, make the adjacency between the vertices in $I_i$ and the vertices 
    in $L_i$ in such a way that, if exactly one of the set $L_i$ or $I_i$ is in a switching set $A$, then these vertices together induce a $C_7$ in $S(G_{\Phi},A)$. This adjacency is shown in Figure~\ref{fig:c7:liii}.
   
    \item For all $i \neq j$, $B_{i1}$ is complete to $B_{j1}$. 
     
    \item The set $B_{i1}$ is complete to $I$, for $1\leq i\leq m$.
     
    \item For all $i \neq j$, $I_i$ is complete to $I_j$. 
\end{itemize}

 This completes the construction of the graph $G_{\Phi}$. Figure~\ref{fig:c7} illustrates an example. 
\end{construction}

\begin{figure}[ht]
  \centering
    \begin{tikzpicture}  [myv/.style={circle, draw, inner sep=0pt,line width=0.3mm},myv1/.style={rectangle, draw,inner sep=1pt},myv2/.style={rectangle, draw,inner sep=2.5pt},my/.style={rectangle, draw,dashed,inner sep=0pt},my1/.style={rectangle, draw,dotted,inner sep=0pt,line width=0.3mm},myv3/.style={circle, draw, inner sep=0.25pt,line width=0.3mm},scale=0.85, myv5/.style={rectangle, draw, inner sep=0.5pt}] 

\tikzset{
dotted_block/.style={draw=black,  dash pattern=on 3pt off 2pt,
            inner ysep=1mm,inner xsep=0mm, rectangle}}
 
\node[inner sep=0pt] (a1) at (-1.5,5)   {\resizebox{0.04\textwidth}{!}{\begin{tikzpicture}[myv/.style={circle, draw, inner sep=1pt},myv1/.style={circle, draw, inner sep=1.5pt, white},myv2/.style={rectangle, draw,inner sep=1.5pt}]

\tikzset{
dotted_block/.style={draw=black,  dash pattern=on 3pt off 2pt,
            inner ysep=1mm,inner xsep=0.5mm, rectangle}}

  \node (z) at (0,0) {};

\node[myv] (c) at (-0.2,0.2) {};
  \node[myv] (d) at (-0.2,0) {};
  \node[myv] (e) at (0.2,0.2) {};
  \node[myv] (f) at (0.2,0) {};
\node[myv] (g) at (0.4,0) {};
  \node[myv] (h) at (-0.4,0) {}; 

   \node[myv2][dotted_block,fit=(c) (d) (e) (f)] {};

\draw (c) -- (d);
  \draw (c) -- (e);
  \draw (f) -- (e);
\draw (f) -- (g);
  \draw (h) -- (d);
  
\end{tikzpicture}} };
     \node[myv5] (wa1) [dotted_block,fit=(a1)] {};
    
     \node[inner sep=0pt] (a2) at (0,5) {\resizebox{0.04\textwidth}{!}{\begin{tikzpicture}[myv/.style={circle, draw, inner sep=1pt},myv1/.style={circle, draw, inner sep=1.5pt, white},myv2/.style={rectangle, draw,inner sep=1.5pt}]

\tikzset{
dotted_block/.style={draw=black,  dash pattern=on 3pt off 2pt,
            inner ysep=1mm,inner xsep=0.5mm, rectangle}}

  \node (z) at (0,0) {};

\node[myv] (c) at (-0.2,0.2) {};
  \node[myv] (d) at (-0.2,0) {};
  \node[myv] (e) at (0.2,0.2) {};
  \node[myv] (f) at (0.2,0) {};
\node[myv] (g) at (0.4,0) {};
  \node[myv] (h) at (-0.4,0) {}; 

   \node[myv2][dotted_block,fit=(c) (d) (e) (f)] {};

\draw (c) -- (d);
  \draw (c) -- (e);
  \draw (f) -- (e);
\draw (f) -- (g);
  \draw (h) -- (d);
  
\end{tikzpicture}} };
     \node[myv5] (wa2) [dotted_block,fit=(a2)] {};
    
     \node[inner sep=0pt] (a3) at (1.5,5) {\resizebox{0.04\textwidth}{!}{\begin{tikzpicture}[myv/.style={circle, draw, inner sep=1pt},myv1/.style={circle, draw, inner sep=1.5pt, white},myv2/.style={rectangle, draw,inner sep=1.5pt}]

\tikzset{
dotted_block/.style={draw=black,  dash pattern=on 3pt off 2pt,
            inner ysep=1mm,inner xsep=0.5mm, rectangle}}

  \node (z) at (0,0) {};

\node[myv] (c) at (-0.2,0.2) {};
  \node[myv] (d) at (-0.2,0) {};
  \node[myv] (e) at (0.2,0.2) {};
  \node[myv] (f) at (0.2,0) {};
\node[myv] (g) at (0.4,0) {};
  \node[myv] (h) at (-0.4,0) {}; 

   \node[myv2][dotted_block,fit=(c) (d) (e) (f)] {};

\draw (c) -- (d);
  \draw (c) -- (e);
  \draw (f) -- (e);
\draw (f) -- (g);
  \draw (h) -- (d);
  
\end{tikzpicture}} };
     \node[myv5] (wa3) [dotted_block,fit=(a3)] {};
 
    \node[inner sep=0pt] (a4) at (3,5) {\resizebox{0.04\textwidth}{!}{\begin{tikzpicture}[myv/.style={circle, draw, inner sep=1pt},myv1/.style={circle, draw, inner sep=1.5pt, white},myv2/.style={rectangle, draw,inner sep=1.5pt}]

\tikzset{
dotted_block/.style={draw=black,  dash pattern=on 3pt off 2pt,
            inner ysep=1mm,inner xsep=0.5mm, rectangle}}

  \node (z) at (0,0) {};

\node[myv] (c) at (-0.2,0.2) {};
  \node[myv] (d) at (-0.2,0) {};
  \node[myv] (e) at (0.2,0.2) {};
  \node[myv] (f) at (0.2,0) {};
\node[myv] (g) at (0.4,0) {};
  \node[myv] (h) at (-0.4,0) {}; 

   \node[myv2][dotted_block,fit=(c) (d) (e) (f)] {};

\draw (c) -- (d);
  \draw (c) -- (e);
  \draw (f) -- (e);
\draw (f) -- (g);
  \draw (h) -- (d);
  
\end{tikzpicture}} };
     \node[myv5] (wa4) [dotted_block,fit=(a4)] {};
 
\node[inner sep=0pt] (b1) at (-1.5,6.5){\resizebox{0.04\textwidth}{!}{\begin{tikzpicture}[myv/.style={circle, draw, inner sep=1pt},myv1/.style={circle, draw, inner sep=1.5pt, white},myv2/.style={rectangle, draw,inner sep=1.5pt}]

\tikzset{
dotted_block/.style={draw=black,  dash pattern=on 3pt off 2pt,
            inner ysep=1mm,inner xsep=0.5mm, rectangle}}

  \node (z) at (0,0) {};

\node[myv] (c) at (-0.2,0.2) {};
  \node[myv] (d) at (-0.2,0) {};
  \node[myv] (e) at (0.2,0.2) {};
  \node[myv] (f) at (0.2,0) {};
\node[myv] (g) at (0.4,0) {};
  \node[myv] (h) at (-0.4,0) {}; 

   \node[myv2][dotted_block,fit=(c) (d) (e) (f)] {};

\draw (c) -- (d);
  \draw (c) -- (e);
  \draw (f) -- (e);
\draw (f) -- (g);
  \draw (h) -- (d);
  
\end{tikzpicture}} };
     \node[myv5] (wb1) [dotted_block,fit=(b1)] {};
   
     \node[inner sep=0pt] (b2) at (0,6.5) {\resizebox{0.04\textwidth}{!}{\begin{tikzpicture}[myv/.style={circle, draw, inner sep=1pt},myv1/.style={circle, draw, inner sep=1.5pt, white},myv2/.style={rectangle, draw,inner sep=1.5pt}]

\tikzset{
dotted_block/.style={draw=black,  dash pattern=on 3pt off 2pt,
            inner ysep=1mm,inner xsep=0.5mm, rectangle}}

  \node (z) at (0,0) {};

\node[myv] (c) at (-0.2,0.2) {};
  \node[myv] (d) at (-0.2,0) {};
  \node[myv] (e) at (0.2,0.2) {};
  \node[myv] (f) at (0.2,0) {};
\node[myv] (g) at (0.4,0) {};
  \node[myv] (h) at (-0.4,0) {}; 

   \node[myv2][dotted_block,fit=(c) (d) (e) (f)] {};

\draw (c) -- (d);
  \draw (c) -- (e);
  \draw (f) -- (e);
\draw (f) -- (g);
  \draw (h) -- (d);
  
\end{tikzpicture}} };
     \node[myv5] (wb2) [dotted_block,fit=(b2)] {};
   
     \node[inner sep=0pt] (b3) at (1.5,6.5) {\resizebox{0.04\textwidth}{!}{\begin{tikzpicture}[myv/.style={circle, draw, inner sep=1pt},myv1/.style={circle, draw, inner sep=1.5pt, white},myv2/.style={rectangle, draw,inner sep=1.5pt}]

\tikzset{
dotted_block/.style={draw=black,  dash pattern=on 3pt off 2pt,
            inner ysep=1mm,inner xsep=0.5mm, rectangle}}

  \node (z) at (0,0) {};

\node[myv] (c) at (-0.2,0.2) {};
  \node[myv] (d) at (-0.2,0) {};
  \node[myv] (e) at (0.2,0.2) {};
  \node[myv] (f) at (0.2,0) {};
\node[myv] (g) at (0.4,0) {};
  \node[myv] (h) at (-0.4,0) {}; 

   \node[myv2][dotted_block,fit=(c) (d) (e) (f)] {};

\draw (c) -- (d);
  \draw (c) -- (e);
  \draw (f) -- (e);
\draw (f) -- (g);
  \draw (h) -- (d);
  
\end{tikzpicture}} };
     \node[myv5] (wb3) [dotted_block,fit=(b3)] {};
 
    \node[inner sep=0pt] (b4) at (3,6.5) {\resizebox{0.04\textwidth}{!}{\begin{tikzpicture}[myv/.style={circle, draw, inner sep=1pt},myv1/.style={circle, draw, inner sep=1.5pt, white},myv2/.style={rectangle, draw,inner sep=1.5pt}]

\tikzset{
dotted_block/.style={draw=black,  dash pattern=on 3pt off 2pt,
            inner ysep=1mm,inner xsep=0.5mm, rectangle}}

  \node (z) at (0,0) {};

\node[myv] (c) at (-0.2,0.2) {};
  \node[myv] (d) at (-0.2,0) {};
  \node[myv] (e) at (0.2,0.2) {};
  \node[myv] (f) at (0.2,0) {};
\node[myv] (g) at (0.4,0) {};
  \node[myv] (h) at (-0.4,0) {}; 

   \node[myv2][dotted_block,fit=(c) (d) (e) (f)] {};

\draw (c) -- (d);
  \draw (c) -- (e);
  \draw (f) -- (e);
\draw (f) -- (g);
  \draw (h) -- (d);
  
\end{tikzpicture}} };
     \node[myv5] (wb4) [dotted_block,fit=(b4)] {};
 
\node[inner sep=0pt] (c1) at (-1.5,8.5){\resizebox{0.04\textwidth}{!}{\begin{tikzpicture}[myv/.style={circle, draw, inner sep=1pt},myv1/.style={circle, draw, inner sep=1.5pt, white},myv2/.style={rectangle, draw,inner sep=1.5pt}]

\tikzset{
dotted_block/.style={draw=black,  dash pattern=on 3pt off 2pt,
            inner ysep=1mm,inner xsep=0.5mm, rectangle}}

  \node (z) at (0,0) {};

\node[myv] (c) at (-0.2,0.2) {};
  \node[myv] (d) at (-0.2,0) {};
  \node[myv] (e) at (0.2,0.2) {};
  \node[myv] (f) at (0.2,0) {};
\node[myv] (g) at (0.4,0) {};
  \node[myv] (h) at (-0.4,0) {}; 

   \node[myv2][dotted_block,fit=(c) (d) (e) (f)] {};

\draw (c) -- (d);
  \draw (c) -- (e);
  \draw (f) -- (e);
\draw (f) -- (g);
  \draw (h) -- (d);
  
\end{tikzpicture}} };
     \node[myv5]  (wc1) [dotted_block,fit=(c1)] {};
  
     \node[inner sep=0pt] (c2) at (0,8.5) {\resizebox{0.04\textwidth}{!}{\begin{tikzpicture}[myv/.style={circle, draw, inner sep=1pt},myv1/.style={circle, draw, inner sep=1.5pt, white},myv2/.style={rectangle, draw,inner sep=1.5pt}]

\tikzset{
dotted_block/.style={draw=black,  dash pattern=on 3pt off 2pt,
            inner ysep=1mm,inner xsep=0.5mm, rectangle}}

  \node (z) at (0,0) {};

\node[myv] (c) at (-0.2,0.2) {};
  \node[myv] (d) at (-0.2,0) {};
  \node[myv] (e) at (0.2,0.2) {};
  \node[myv] (f) at (0.2,0) {};
\node[myv] (g) at (0.4,0) {};
  \node[myv] (h) at (-0.4,0) {}; 

   \node[myv2][dotted_block,fit=(c) (d) (e) (f)] {};

\draw (c) -- (d);
  \draw (c) -- (e);
  \draw (f) -- (e);
\draw (f) -- (g);
  \draw (h) -- (d);
  
\end{tikzpicture}} };
     \node[myv5] (wc2) [dotted_block,fit=(c2)] {};
 
     \node[inner sep=0pt] (c3) at (1.5,8.5) {\resizebox{0.04\textwidth}{!}{\begin{tikzpicture}[myv/.style={circle, draw, inner sep=1pt},myv1/.style={circle, draw, inner sep=1.5pt, white},myv2/.style={rectangle, draw,inner sep=1.5pt}]

\tikzset{
dotted_block/.style={draw=black,  dash pattern=on 3pt off 2pt,
            inner ysep=1mm,inner xsep=0.5mm, rectangle}}

  \node (z) at (0,0) {};

\node[myv] (c) at (-0.2,0.2) {};
  \node[myv] (d) at (-0.2,0) {};
  \node[myv] (e) at (0.2,0.2) {};
  \node[myv] (f) at (0.2,0) {};
\node[myv] (g) at (0.4,0) {};
  \node[myv] (h) at (-0.4,0) {}; 

   \node[myv2][dotted_block,fit=(c) (d) (e) (f)] {};

\draw (c) -- (d);
  \draw (c) -- (e);
  \draw (f) -- (e);
\draw (f) -- (g);
  \draw (h) -- (d);
  
\end{tikzpicture}} };
     \node[myv5] (wc3) [dotted_block,fit=(c3)] {};
   
    \node[inner sep=0pt] (c4) at (3,8.5) {\resizebox{0.04\textwidth}{!}{\begin{tikzpicture}[myv/.style={circle, draw, inner sep=1pt},myv1/.style={circle, draw, inner sep=1.5pt, white},myv2/.style={rectangle, draw,inner sep=1.5pt}]

\tikzset{
dotted_block/.style={draw=black,  dash pattern=on 3pt off 2pt,
            inner ysep=1mm,inner xsep=0.5mm, rectangle}}

  \node (z) at (0,0) {};

\node[myv] (c) at (-0.2,0.2) {};
  \node[myv] (d) at (-0.2,0) {};
  \node[myv] (e) at (0.2,0.2) {};
  \node[myv] (f) at (0.2,0) {};
\node[myv] (g) at (0.4,0) {};
  \node[myv] (h) at (-0.4,0) {}; 

   \node[myv2][dotted_block,fit=(c) (d) (e) (f)] {};

\draw (c) -- (d);
  \draw (c) -- (e);
  \draw (f) -- (e);
\draw (f) -- (g);
  \draw (h) -- (d);
  
\end{tikzpicture}} };
     \node[myv5] (wc4) [dotted_block,fit=(c4)] {};

     \node[myv] at (-1.5,7.25) {};
     \node[myv] at (-1.5,7.5) {};
     \node[myv] at (-1.5,7.75) {};
     
     \node[myv]  at (0,7.25) {};
     \node[myv]  at (0,7.5) {};
     \node[myv]  at (0,7.75) {};
     
     \node[myv]  at (1.5,7.25) {};
     \node[myv]  at (1.5,7.5) {};
     \node[myv]  at (1.5,7.75) {};
     
     \node[myv]  at (3,7.25) {};
     \node[myv]  at (3,7.5) {};
     \node[myv]  at (3,7.75) {};

     \node[myv3] (i1) at (-1.5,10) {$I_{11}$};
     \node[myv3] (i2) at (0,10) {$I_{12}$};
     \node[myv3] (i3) at (1.5,10) {$I_{13}$};
     \node[myv3] (i4) at (3,10) {$I_{14}$};

     \node[myv3] (j1) at (6,10) {$I_{21}$};
     \node[myv3] (j2) at (7.5,10) {$I_{22}$};
     \node[myv3] (j3) at (9,10) {$I_{23}$};
     \node[myv3] (j4) at (10.5,10) {$I_{24}$};

    \node[myv3] (x1) at (-1.25,3) {$x_{1}$};
 
    \node[myv3] (x2) at (1.75,3) {${x_{2}}$};

    \node[myv3] (x3) at (4.75,3) {$x_{3}$};

    \node[myv3] (x4) at (7.75,3) {$x_{4}$};
 
    \node[myv3] (x5) at (10.75,3) {${x_{5}}$};

\node[inner sep=0pt] (aj1) at (6,5)   {\resizebox{0.04\textwidth}{!}{\begin{tikzpicture}[myv/.style={circle, draw, inner sep=1pt},myv1/.style={circle, draw, inner sep=1.5pt, white},myv2/.style={rectangle, draw,inner sep=1.5pt}]

\tikzset{
dotted_block/.style={draw=black,  dash pattern=on 3pt off 2pt,
            inner ysep=1mm,inner xsep=0.5mm, rectangle}}

  \node (z) at (0,0) {};

\node[myv] (c) at (-0.2,0.2) {};
  \node[myv] (d) at (-0.2,0) {};
  \node[myv] (e) at (0.2,0.2) {};
  \node[myv] (f) at (0.2,0) {};
\node[myv] (g) at (0.4,0) {};
  \node[myv] (h) at (-0.4,0) {}; 

   \node[myv2][dotted_block,fit=(c) (d) (e) (f)] {};

\draw (c) -- (d);
  \draw (c) -- (e);
  \draw (f) -- (e);
\draw (f) -- (g);
  \draw (h) -- (d);
  
\end{tikzpicture}} };
     \node[myv5] (waj1) [dotted_block,fit=(aj1)] {};
    
     \node[inner sep=0pt] (aj2) at (7.5,5) {\resizebox{0.04\textwidth}{!}{\begin{tikzpicture}[myv/.style={circle, draw, inner sep=1pt},myv1/.style={circle, draw, inner sep=1.5pt, white},myv2/.style={rectangle, draw,inner sep=1.5pt}]

\tikzset{
dotted_block/.style={draw=black,  dash pattern=on 3pt off 2pt,
            inner ysep=1mm,inner xsep=0.5mm, rectangle}}

  \node (z) at (0,0) {};

\node[myv] (c) at (-0.2,0.2) {};
  \node[myv] (d) at (-0.2,0) {};
  \node[myv] (e) at (0.2,0.2) {};
  \node[myv] (f) at (0.2,0) {};
\node[myv] (g) at (0.4,0) {};
  \node[myv] (h) at (-0.4,0) {}; 

   \node[myv2][dotted_block,fit=(c) (d) (e) (f)] {};

\draw (c) -- (d);
  \draw (c) -- (e);
  \draw (f) -- (e);
\draw (f) -- (g);
  \draw (h) -- (d);
  
\end{tikzpicture}} };
     \node[myv5] (waj2) [dotted_block,fit=(aj2)] {};
    
     \node[inner sep=0pt] (aj3) at (9,5) {\resizebox{0.04\textwidth}{!}{\begin{tikzpicture}[myv/.style={circle, draw, inner sep=1pt},myv1/.style={circle, draw, inner sep=1.5pt, white},myv2/.style={rectangle, draw,inner sep=1.5pt}]

\tikzset{
dotted_block/.style={draw=black,  dash pattern=on 3pt off 2pt,
            inner ysep=1mm,inner xsep=0.5mm, rectangle}}

  \node (z) at (0,0) {};

\node[myv] (c) at (-0.2,0.2) {};
  \node[myv] (d) at (-0.2,0) {};
  \node[myv] (e) at (0.2,0.2) {};
  \node[myv] (f) at (0.2,0) {};
\node[myv] (g) at (0.4,0) {};
  \node[myv] (h) at (-0.4,0) {}; 

   \node[myv2][dotted_block,fit=(c) (d) (e) (f)] {};

\draw (c) -- (d);
  \draw (c) -- (e);
  \draw (f) -- (e);
\draw (f) -- (g);
  \draw (h) -- (d);
  
\end{tikzpicture}} };
     \node[myv5] (waj3) [dotted_block,fit=(aj3)] {};
 
    \node[inner sep=0pt] (aj4) at (10.5,5) {\resizebox{0.04\textwidth}{!}{\begin{tikzpicture}[myv/.style={circle, draw, inner sep=1pt},myv1/.style={circle, draw, inner sep=1.5pt, white},myv2/.style={rectangle, draw,inner sep=1.5pt}]

\tikzset{
dotted_block/.style={draw=black,  dash pattern=on 3pt off 2pt,
            inner ysep=1mm,inner xsep=0.5mm, rectangle}}

  \node (z) at (0,0) {};

\node[myv] (c) at (-0.2,0.2) {};
  \node[myv] (d) at (-0.2,0) {};
  \node[myv] (e) at (0.2,0.2) {};
  \node[myv] (f) at (0.2,0) {};
\node[myv] (g) at (0.4,0) {};
  \node[myv] (h) at (-0.4,0) {}; 

   \node[myv2][dotted_block,fit=(c) (d) (e) (f)] {};

\draw (c) -- (d);
  \draw (c) -- (e);
  \draw (f) -- (e);
\draw (f) -- (g);
  \draw (h) -- (d);
  
\end{tikzpicture}} };
     \node[myv5] (waj4) [dotted_block,fit=(aj4)] {};

\node[inner sep=0pt] (bj1) at (6,6.5){\resizebox{0.04\textwidth}{!}{\begin{tikzpicture}[myv/.style={circle, draw, inner sep=1pt},myv1/.style={circle, draw, inner sep=1.5pt, white},myv2/.style={rectangle, draw,inner sep=1.5pt}]

\tikzset{
dotted_block/.style={draw=black,  dash pattern=on 3pt off 2pt,
            inner ysep=1mm,inner xsep=0.5mm, rectangle}}

  \node (z) at (0,0) {};

\node[myv] (c) at (-0.2,0.2) {};
  \node[myv] (d) at (-0.2,0) {};
  \node[myv] (e) at (0.2,0.2) {};
  \node[myv] (f) at (0.2,0) {};
\node[myv] (g) at (0.4,0) {};
  \node[myv] (h) at (-0.4,0) {}; 

   \node[myv2][dotted_block,fit=(c) (d) (e) (f)] {};

\draw (c) -- (d);
  \draw (c) -- (e);
  \draw (f) -- (e);
\draw (f) -- (g);
  \draw (h) -- (d);
  
\end{tikzpicture}} };
     \node[myv5] (wbj1) [dotted_block,fit=(bj1)] {};
    
     \node[inner sep=0pt] (bj2) at (7.5,6.5) {\resizebox{0.04\textwidth}{!}{\begin{tikzpicture}[myv/.style={circle, draw, inner sep=1pt},myv1/.style={circle, draw, inner sep=1.5pt, white},myv2/.style={rectangle, draw,inner sep=1.5pt}]

\tikzset{
dotted_block/.style={draw=black,  dash pattern=on 3pt off 2pt,
            inner ysep=1mm,inner xsep=0.5mm, rectangle}}

  \node (z) at (0,0) {};

\node[myv] (c) at (-0.2,0.2) {};
  \node[myv] (d) at (-0.2,0) {};
  \node[myv] (e) at (0.2,0.2) {};
  \node[myv] (f) at (0.2,0) {};
\node[myv] (g) at (0.4,0) {};
  \node[myv] (h) at (-0.4,0) {}; 

   \node[myv2][dotted_block,fit=(c) (d) (e) (f)] {};

\draw (c) -- (d);
  \draw (c) -- (e);
  \draw (f) -- (e);
\draw (f) -- (g);
  \draw (h) -- (d);
  
\end{tikzpicture}} };
     \node[myv5] (wbj2) [dotted_block,fit=(bj2)] {};
   
     \node[inner sep=0pt] (bj3) at (9,6.5) {\resizebox{0.04\textwidth}{!}{\begin{tikzpicture}[myv/.style={circle, draw, inner sep=1pt},myv1/.style={circle, draw, inner sep=1.5pt, white},myv2/.style={rectangle, draw,inner sep=1.5pt}]

\tikzset{
dotted_block/.style={draw=black,  dash pattern=on 3pt off 2pt,
            inner ysep=1mm,inner xsep=0.5mm, rectangle}}

  \node (z) at (0,0) {};

\node[myv] (c) at (-0.2,0.2) {};
  \node[myv] (d) at (-0.2,0) {};
  \node[myv] (e) at (0.2,0.2) {};
  \node[myv] (f) at (0.2,0) {};
\node[myv] (g) at (0.4,0) {};
  \node[myv] (h) at (-0.4,0) {}; 

   \node[myv2][dotted_block,fit=(c) (d) (e) (f)] {};

\draw (c) -- (d);
  \draw (c) -- (e);
  \draw (f) -- (e);
\draw (f) -- (g);
  \draw (h) -- (d);
  
\end{tikzpicture}} };
     \node[myv5] (wbj3) [dotted_block,fit=(bj3)] {};
  
    \node[inner sep=0pt] (bj4) at (10.5,6.5) {\resizebox{0.04\textwidth}{!}{\begin{tikzpicture}[myv/.style={circle, draw, inner sep=1pt},myv1/.style={circle, draw, inner sep=1.5pt, white},myv2/.style={rectangle, draw,inner sep=1.5pt}]

\tikzset{
dotted_block/.style={draw=black,  dash pattern=on 3pt off 2pt,
            inner ysep=1mm,inner xsep=0.5mm, rectangle}}

  \node (z) at (0,0) {};

\node[myv] (c) at (-0.2,0.2) {};
  \node[myv] (d) at (-0.2,0) {};
  \node[myv] (e) at (0.2,0.2) {};
  \node[myv] (f) at (0.2,0) {};
\node[myv] (g) at (0.4,0) {};
  \node[myv] (h) at (-0.4,0) {}; 

   \node[myv2][dotted_block,fit=(c) (d) (e) (f)] {};

\draw (c) -- (d);
  \draw (c) -- (e);
  \draw (f) -- (e);
\draw (f) -- (g);
  \draw (h) -- (d);
  
\end{tikzpicture}} };
     \node[myv5] (wbj4) [dotted_block,fit=(bj4)] {};

\node[inner sep=0pt] (cj1) at (6,8.5){\resizebox{0.04\textwidth}{!}{\begin{tikzpicture}[myv/.style={circle, draw, inner sep=1pt},myv1/.style={circle, draw, inner sep=1.5pt, white},myv2/.style={rectangle, draw,inner sep=1.5pt}]

\tikzset{
dotted_block/.style={draw=black,  dash pattern=on 3pt off 2pt,
            inner ysep=1mm,inner xsep=0.5mm, rectangle}}

  \node (z) at (0,0) {};

\node[myv] (c) at (-0.2,0.2) {};
  \node[myv] (d) at (-0.2,0) {};
  \node[myv] (e) at (0.2,0.2) {};
  \node[myv] (f) at (0.2,0) {};
\node[myv] (g) at (0.4,0) {};
  \node[myv] (h) at (-0.4,0) {}; 

   \node[myv2][dotted_block,fit=(c) (d) (e) (f)] {};

\draw (c) -- (d);
  \draw (c) -- (e);
  \draw (f) -- (e);
\draw (f) -- (g);
  \draw (h) -- (d);
  
\end{tikzpicture}} };
     \node[myv5]  (wcj1) [dotted_block,fit=(cj1)] {};
  
     \node[inner sep=0pt] (cj2) at (7.5,8.5) {\resizebox{0.04\textwidth}{!}{\begin{tikzpicture}[myv/.style={circle, draw, inner sep=1pt},myv1/.style={circle, draw, inner sep=1.5pt, white},myv2/.style={rectangle, draw,inner sep=1.5pt}]

\tikzset{
dotted_block/.style={draw=black,  dash pattern=on 3pt off 2pt,
            inner ysep=1mm,inner xsep=0.5mm, rectangle}}

  \node (z) at (0,0) {};

\node[myv] (c) at (-0.2,0.2) {};
  \node[myv] (d) at (-0.2,0) {};
  \node[myv] (e) at (0.2,0.2) {};
  \node[myv] (f) at (0.2,0) {};
\node[myv] (g) at (0.4,0) {};
  \node[myv] (h) at (-0.4,0) {}; 

   \node[myv2][dotted_block,fit=(c) (d) (e) (f)] {};

\draw (c) -- (d);
  \draw (c) -- (e);
  \draw (f) -- (e);
\draw (f) -- (g);
  \draw (h) -- (d);
  
\end{tikzpicture}} };
     \node[myv5] (wcj2) [dotted_block,fit=(cj2)] {};
 
     \node[inner sep=0pt] (cj3) at (9,8.5) {\resizebox{0.04\textwidth}{!}{\begin{tikzpicture}[myv/.style={circle, draw, inner sep=1pt},myv1/.style={circle, draw, inner sep=1.5pt, white},myv2/.style={rectangle, draw,inner sep=1.5pt}]

\tikzset{
dotted_block/.style={draw=black,  dash pattern=on 3pt off 2pt,
            inner ysep=1mm,inner xsep=0.5mm, rectangle}}

  \node (z) at (0,0) {};

\node[myv] (c) at (-0.2,0.2) {};
  \node[myv] (d) at (-0.2,0) {};
  \node[myv] (e) at (0.2,0.2) {};
  \node[myv] (f) at (0.2,0) {};
\node[myv] (g) at (0.4,0) {};
  \node[myv] (h) at (-0.4,0) {}; 

   \node[myv2][dotted_block,fit=(c) (d) (e) (f)] {};

\draw (c) -- (d);
  \draw (c) -- (e);
  \draw (f) -- (e);
\draw (f) -- (g);
  \draw (h) -- (d);
  
\end{tikzpicture}} };
     \node[myv5] (wcj3) [dotted_block,fit=(cj3)] {};
   
    \node[inner sep=0pt] (cj4) at (10.5,8.5) {\resizebox{0.04\textwidth}{!}{\begin{tikzpicture}[myv/.style={circle, draw, inner sep=1pt},myv1/.style={circle, draw, inner sep=1.5pt, white},myv2/.style={rectangle, draw,inner sep=1.5pt}]

\tikzset{
dotted_block/.style={draw=black,  dash pattern=on 3pt off 2pt,
            inner ysep=1mm,inner xsep=0.5mm, rectangle}}

  \node (z) at (0,0) {};

\node[myv] (c) at (-0.2,0.2) {};
  \node[myv] (d) at (-0.2,0) {};
  \node[myv] (e) at (0.2,0.2) {};
  \node[myv] (f) at (0.2,0) {};
\node[myv] (g) at (0.4,0) {};
  \node[myv] (h) at (-0.4,0) {}; 

   \node[myv2][dotted_block,fit=(c) (d) (e) (f)] {};

\draw (c) -- (d);
  \draw (c) -- (e);
  \draw (f) -- (e);
\draw (f) -- (g);
  \draw (h) -- (d);
  
\end{tikzpicture}} };
     \node[myv5] (wcj4) [dotted_block,fit=(cj4)] {};

     \node[myv] at (6,7.25) {};
     \node[myv] at (6,7.5) {};
     \node[myv] at (6,7.75) {};
     
     \node[myv]  at (7.5,7.25) {};
     \node[myv]  at (7.5,7.5) {};
     \node[myv]  at (7.5,7.75) {};
     
     \node[myv]  at (9,7.25) {};
     \node[myv]  at (9,7.5) {};
     \node[myv]  at (9,7.75) {};
     
     \node[myv]  at (10.5,7.25) {};
     \node[myv]  at (10.5,7.5) {};
     \node[myv]  at (10.5,7.75) {};

    \node[my1](p)[fit=(a1) (a4),  inner xsep=1ex, inner ysep=1.5ex] {}; 
      \node[inner sep=0pt,label=left:$B_{11}$] (bi) at (-1.95,5) {};
 
     \node[my1](b)[fit=(b1) (b4),  inner xsep=1ex, inner ysep=1.5ex] {}; 
    \node[inner sep=0pt,label=left:$B_{12}$] (x) at (-1.95,6.5) {};
   \node[my1](b)[fit=(c1) (c4),  inner xsep=1ex, inner ysep=1.5ex] {}; 
     \node[inner sep=0pt,label=left:$B_{18}$] (x) at (-1.95,8.5) {};

      \node[my1](q)[fit=(aj1) (aj4),  inner xsep=1ex, inner ysep=1.5ex] {}; 
      \node[inner sep=0pt,label=left:$B_{21}$] (x) at (11.9,4.5) {};
 
     \node[my1](b)[fit=(bj1) (bj4),  inner xsep=1ex, inner ysep=1.5ex] {}; 
    \node[inner sep=0pt,label=left:$B_{22}$] (x) at (11.9,6.5) {};
   \node[my1](b)[fit=(cj1) (cj4),  inner xsep=1ex, inner ysep=1.5ex] {}; 
     \node[inner sep=0pt,label=left:$B_{28}$] (x) at (11.9,8.5) {};

      \node[my1][line width=0.4mm](pq)[fit=(p) (q),  inner xsep=1ex, inner ysep=1ex] {}; 
 
    \node[my] (i) [fit=(i1) (i4),  inner xsep=1.2ex, inner ysep=1.7ex, label=right:$I_{1}$] {}; 

    \node[my] (j) [fit=(j1) (j4),  inner xsep=1.2ex, inner ysep=1.7ex, label=left:$I_{2}$] {}; 

     \node[my] (ij) [fit=(i) (j),  inner xsep=1ex, inner ysep=0.5ex, label=above:$I$] {};
    
    \node[my] (x) [fit=(x1) (x3),  inner xsep=1.5ex, inner ysep=2ex, label=left:$L_{1}$] {}; 

     \node[my] [dotted,line width=0.4mm] (y) [fit=(x3) (x5),  inner xsep=1.5ex, inner ysep=2ex, label=right:$L_{2}$] {};

    \node[my] (x) [fit=(wa1) (i) (wc4),  inner xsep=2.4ex, inner ysep=2.5ex, label=above:$Y_{1}$] {}; 

     \node[my] (y) [fit=(waj1) (j) (wcj4),  inner xsep=2.4ex, inner ysep=2.5ex, label=above:$Y_{2}$] {};

    \draw[line width=0.3mm] (-1.4,5.38) to[bend left=17]  (2.9, 5.38);
  
    \draw[line width=0.3mm] (6.1,5.38) to[bend left=17]  (10.6, 5.38);

    \draw[line width=0.3mm] (4.5,9.2) -- (4.5,5.75);

    \draw [double,line width=0.3mm] (11.15,10.1) to[bend left=30]  (11.35,2.8);
    \draw [double,line width=0.3mm] (-2.15,10.1) to[bend right=30]  (-1.85,2.8);

\draw[line width=0.3mm] (wa1) -- (x1);
   \draw[line width=0.3mm] (0,4.68) -- (x1);
    \draw[line width=0.3mm] (0,4.68) -- (x2);
   \draw[line width=0.3mm] (1.5,4.68) -- (x2);
    \draw[line width=0.3mm] (1.5,4.68) -- (x3);
   \draw[line width=0.3mm] (3,4.68) -- (x3);

\draw[line width=0.3mm] (waj1) -- (x3);
   \draw[line width=0.3mm] (7.5,4.68) -- (x3);
    \draw[line width=0.3mm] (7.5,4.68) -- (x4);
   \draw[line width=0.3mm] (9,4.68) -- (x4);
    \draw[line width=0.3mm] (9,4.68) -- (x5);
   \draw[line width=0.3mm] (10.5,4.68) -- (x5);
 
   \draw[line width=0.3mm] (a1) -- (wb1);
   \draw[line width=0.3mm] (a2) -- (wb2);
   \draw[line width=0.3mm] (a3) -- (wb3);
   \draw[line width=0.3mm] (a4) -- (wb4);

    \draw[line width=0.3mm] (aj1) -- (wbj1);
   \draw[line width=0.3mm] (aj2) -- (wbj2);
   \draw[line width=0.3mm] (aj3) -- (wbj3);
   \draw[line width=0.3mm] (aj4) -- (wbj4);
  
   \draw[line width=0.3mm] (i1) to[bend left=25] (i4);
   \draw[line width=0.3mm] (i) -- (-1.28,8.7);
   \draw[line width=0.3mm] (i) -- (c2);
   \draw[line width=0.3mm] (i) -- (c3);
   \draw[line width=0.3mm] (i) -- (3,8.7);

   \draw[line width=0.3mm] (j1) to[bend left=25] (j4);
   \draw[line width=0.3mm] (j) -- (6.1,8.7);
   \draw[line width=0.3mm] (j) -- (cj2);
   \draw[line width=0.3mm] (j) -- (cj3);
   \draw[line width=0.3mm] (j) -- (10.5,8.7);

   \draw[line width=0.3mm] (3.6,4.75) -- (5.5,4.75);

   \draw[line width=0.3mm] (3.6,9.55) -- (5.4,9.55);
  
\end{tikzpicture}  
     \caption{An example of Construction \ref{con:c7} with the formula $\Phi = C_1\wedge C_2$,  where $C_1=\{x_1, x_2, x_3\}$ and $C_2=\{x_3, x_4, x_5\}$.
    A single line connecting two rectangles indicates that 
    each vertex in one rectangle is adjacent to all vertices in the other rectangle.
    The adjacency between $L_i$ and $I_i$ is indicated by a double line and is illustrated in Figure~\ref{fig:c7:liii}.}
  \label{fig:c7}
\end{figure}

\begin{figure}[!htb]
  \centering
    \resizebox{0.4\textwidth}{!}{\begin{tikzpicture}  [myv/.style={circle, draw, inner sep=0pt},myv1/.style={rectangle, draw,inner sep=1pt},myv2/.style={rectangle, draw,inner sep=2.5pt},my/.style={rectangle, draw,dashed,inner sep=0pt},myv3/.style={circle, draw, inner sep=0.25pt},scale=0.85, myv5/.style={circle, draw, inner sep=0.75pt}] 

\tikzset{
dotted_block/.style={draw=block,  dash pattern=on 3pt off 2pt, rectangle}}

     \node[myv3] (i1) at (-0.5,4.5) {$I_{i1}$};
     \node[myv3] (i2) at (2,4.5) {$I_{i2}$};
     \node[myv3] (i3) at (4.5,4.5) {$I_{i3}$};
     \node[myv3] (i4) at (7,4.5) {$I_{i4}$};

    \node[myv5] (x1) at (0,2.25) {$x_{1}$};
 
    \node[myv5] (x2) at (3,2.25) {${x_{2}}$};

    \node[myv5] (x3) at (6,2.25) {$x_{3}$};

    \node[my] (i) [fit=(i1) (i4),  inner xsep=0.5ex, inner ysep=0.5ex, label=left:$I_{i}$] {}; 
    
    \node[my] (x) [fit=(x1) (x3),  inner xsep=0.5ex, inner ysep=0.5ex, label=left:$L_{i}$] {};

  \draw [line width=0.3mm](i1) to[bend left=20] (i4);
  \draw [line width=0.3mm](i1) -- (x2);
  \draw [line width=0.3mm](i1) -- (x3);
\draw [line width=0.3mm](i2) -- (x3);
  \draw [line width=0.3mm](i3) -- (x1);
\draw [line width=0.3mm](i4) -- (x1);
  \draw [line width=0.3mm](i4) -- (x2);

\end{tikzpicture}  
}     \caption{The adjacency between $L_i$ and $I_i$ in $G_\Phi$ as described in Construction~\ref{con:c7}}
  \label{fig:c7:liii}
\end{figure}

First we prove the backward direction of the reduction.

\begin{lemma}
    \label{lem:c7:backward}
    Let $A$ be a subset of vertices of $G_{\Phi}$
    such that $G' = S(G_{\Phi}, A)$ is $C_7$-free. Then
    assigning TRUE to all the variables corresponding to the variable vertices in $L\cap A$ satisfies $\Phi$.
\end{lemma}
\begin{proof}
    It is sufficient to prove that $A\cap L_i\neq \emptyset$ and $L_i\setminus A\neq \emptyset$, for $1\leq i\leq m$.
    For a contradiction, first assume that $A\cap L_i=\emptyset$, for some $1\leq i\leq m$. Recall that the vertices in $L_i$
    along with one vertex each from $B_{i1\ell}$ (for $1\leq \ell\leq 4$) induces a $C_7$ in $G_{\Phi}$.
    Therefore, since $L_i$ has no vertex in $A$, at least one set $B_{i1\ell}\subseteq A$ (for some $1\leq \ell\leq 4$). This implies that every vertex in $B_{i2\ell}$ is in $A$ - otherwise there is a $C_7$ induced in $G'$ by $B_{i1\ell}$ and a vertex in $B_{i2\ell}\setminus A$. Continuing these arguments, we obtain that $B_{i8\ell}\subseteq A$. This in turn implies that $I_i\subseteq A$. Then by the construction, $L_i\cup I_i$ induces a $C_7$ in $G'$, which is a contradiction. 

    For a contradiction, next assume that $L_i\subseteq A$, for some $1\leq i\leq m$.
    Then at least one vertex in $I_i$ is in $A$ - otherwise there is a $C_7$ induced in $G'$ by $L_i\cup I_i$.
    Then at least one vertex each from $B_{i8\ell}$ is in $A$, for $1\leq \ell\leq 4$ - otherwise there is a $C_7$ induced by one vertex in $I_i\cap A$ and a set $B_{i8\ell}$ such that $B_{i8\ell}\cap A=\emptyset$. This implies that at least one vertex in each $B_{i7\ell}$ (for $1\leq \ell\leq 4$)
    is in $A$. Continuing these arguments, we obtain
    that at least one vertex each from $B_{i1\ell}$ (for $1\leq \ell\leq 4$) is in $A$. Then those
    vertices along with $L_i$ induce a $C_7$ in 
    $G'$, which is a contradiction.
\end{proof}

Now onward we assume that $A$ is a subset of $L$
such that $A\cap L_i\neq \emptyset$ and $L_i\setminus A\neq \emptyset$, for every $1\leq i\leq m$. 
Let $G' = S(G_{\Phi}, A)$.
For the forward direction of the reduction, it is 
sufficient to prove that $G'$ is $C_7$-free.
We do this in Lemma~\ref{lem:c7:forward} with the help 
of Observation~\ref{obs:c7:l}, and Lemmas~\ref{lem:c7:bij-atmost1} and \ref{lem:c7:2-8}. Assume for a contradiction  that $R\subseteq V(G')$
induces a $C_7$ in $G'$.
Observation~\ref{obs:c7:l} is implied by the fact
that the graph induced by $L$ in $G'$ is a complete
bipartite graph.

\begin{observation}
    \label{obs:c7:l}
    $R\setminus L\neq \emptyset$.
\end{observation}

\begin{lemma}
    \label{lem:c7:bij-atmost1}
    $|R\cap B_{ij\ell}|\leq 1$, for $1\leq i\leq m$, $1\leq j\leq 8$, $1\leq \ell\leq 4$.
\end{lemma}
\begin{proof}
    Assume that $R\cap B_{ij\ell}$ induces a graph with at least one edge. Let $a,b$ be the end points of a longest path $P$ in the graph induced by $R\cap B_{ij\ell}$. If $P$ has 6 vertices, then the only remaining vertex in $R$ must be complete to exactly $\{a,b\}$ in the path $P$. This is not possible by the construction. Therefore, $P$ has only at most 5 vertices. The vertex $a$ has a neighbor $a'$
    and $b$ has a neighbor $b'$ in the $C_7$ induced 
    by $R$ such that $a',b'\notin B_{ij\ell}$. Then by
    the construction, either $a'$ is adjacent to $b$
    or $b'$ is adjacent to $a$. Therefore, the graph
    induced by $R$ has a cycle of length at most 6, which is a contradiction.

    Assume that $R\cap B_{ij\ell}$ induces an edgeless graph with at least two vertices, say $a, b$.
    There are exactly two neighbors $a',a''$ of $a$ and two neighbors $b',b''$ of $b$ such that $a',a'',b',b''\in R\setminus B_{ij\ell}$ and $|\{a',a''\}\cap \{b',b''\}|\leq 1$ and $\{a,b,a',a'',b',b''\}$ either induce a $P_5$ or induces a $P_6$. 
    But, by the construction, $b$ is adjacent to both $a'$ and $a''$, or $a$ is adjacent to both $b'$ and $b''$. Therefore, we get a contradiction.
\end{proof}

\begin{lemma}
    \label{lem:c7:2-8}
    $R\cap B_{ij\ell} = \emptyset$, for every
    $1\leq i\leq m, 2\leq j\leq 8, 1\leq \ell\leq 4$.
\end{lemma}
\begin{proof}
    For a contradiction, assume that $z\in R\cap B_{ij\ell}$. By Lemma~\ref{lem:c7:bij-atmost1},
    $B_{ij\ell}$ has no other vertices in $R$.
    Let $R=\{z,u,v,u',v',u'',v''\}$, and let the edges
    of the $C_7$ induced by $R$ be $\{zu, uu', u'u'', u''v'', v''v', v'v, vz\}$. 
    
    Assume that $j\geq 3$.
    By Lemma~\ref{lem:c7:bij-atmost1}, it is not possible that $u,v\in B'_{i(j-1)\ell}$.
    Similarly, it is not possible that $u,v\in B'_{i(j+1)\ell}$ (for $2\leq j\leq 7$).
    We also note that, due to the adjacency between
    $L\cap A$ and $B_{i(j-1)\ell}$, it is not possible that $u\in B_{i(j-1)\ell}$ and $v\in L\cap A$ (or vice versa). Similarly, it is not possible that
    $u\in B_{i(j+1)\ell}$ and $v\in L\cap A$ (or vice versa), for $2\leq j\leq 7$.
    Therefore, we need to consider only the following cases based on the membership of $u$ and $v$: (1) $u,v\in L\cap A$, 
    (2) $u,v\in I_i$,
    (3) $u\in I_i, v\in B'_{i(j-1)\ell}$, (4) $u\in B'_{i(j-1)\ell}, v\in B'_{i(j+1)\ell}$. 
    We note that (4) is not applicable when $j=8$
    and (2) and (3) are applicable only for $j=8$.
    
    Assume that (1) holds, i.e., $u,v\in L\cap A$. If there is a vertex, say $w \in R\cap (\beta_2\cup \beta_3\cup\ldots\cup \beta_8)$ such that $w\neq z$, then there is a $C_4$ formed by
    $\{z,u,v,w\}$. Therefore, $R\cap (\beta_2\cup \beta_3\cup\ldots\cup \beta_8) = \{z\}$.
    Similarly, if there is a vertex $w \in R\cap (L\setminus A)$ then also there is a $C_4$
    formed by $\{z,u,v,w\}$. Therefore, $R\cap (L\setminus A) = \emptyset$. 
    Since $I$ is complete to $\beta_1$, it is not possible that $u'\in \beta_1$ and $v'\in I$ (or vice versa). 
    Now, there are only two cases to consider based on the membership of $u'$ and $v'$: (a) $u',v'\in \beta_1$, or (b) $u',v'\in I$. We obtain contradictions below in each case.

    Assume that $u',v'\in \beta_1$. Recall that $u'$ and $v'$ are not adjacent. Since $B_{i'1}$ is complete to
    $B_{i''1}$, for $i'\neq i''$,
    we obtain that $u',v'\in B_{i'1}$, for some $1\leq i'\leq m$. Assume that a vertex from $I$ is in $R$.
    Then that vertex, along with $\{z,u,v,u',v'\}$ form a $C_6$. Therefore 
    $R\cap I=\emptyset$. Since both $u'$ and $v'$ are 
    adjacent to every vertex in $B_{i''1}$ (for $i''\neq i'$), we obtain that $u'',v''\in B_{i'1}$. 
    It implies that $R$ has exactly 4 vertices from 
    $B_{i'1}$ and has one vertex each from each set $B_{i'1\ell'}$ (for $1\leq \ell'\leq 4$), due to Lemma~\ref{lem:c7:bij-atmost1}. This is a contradiction as they induce a $K_2+2K_1$
    instead of a $P_4$.
    Assume the next case, i.e., $u',v'\in I$. Clearly $u'$ and $v'$ are nonadjacent and therefore must be from a set $I_{i'}$ (recall that $I_{i'}$ is complete to $I_{i''}$, for $i'\neq i''$). If $R$ has a vertex
    from $\beta_1$, then it is adjacent to both $u'$
    and $v'$ and then the graph induced by $R$
    contains a $C_6$. Therefore $R\cap \beta_1=\emptyset$. Then $u'',v''\in I_{i'}$. Since $I_{i'}$ induces
    $K_2+2K_1$ (instead of a $P_4$), we get a 
    contradiction. 

    Assume that (2) holds, i.e., $u,v\in I_i$.
    Recall that this case is applicable only when $j=8$. If $R$ has some vertex $w$ from $\beta_1$,
    then $\{u,v,w,z\}$ form a $C_4$. 
    Therefore $R\cap \beta_1=\emptyset$. 
    Therefore, either (a) $u',v'\in L$, or (b) $u'\in L, v'\in I$. If $R\cap L$ induces a graph with at least one edge, then $R\cap A$ is nonempty and a vertex in $R\cap A$ is adjacent to $z$, which is a contradiction. Therefore, $R\cap L$ induces an edgeless graph. Further, $R\cap L\cap A = \emptyset$.

    Assume that $u',v'\in L$.  Then $u'',v''\in I_i$ (if $u'',v''\in I\setminus I_i$, then $u$ is adjacent to $u''$, which is a contradiction).
    Then $z$ is adjacent to both $u''$ and $v''$, as $B_{ij\ell}$ is complete to $I_i$, which is a contradiction. Assume the next case, i.e., $u'\in L, v'\in I$. Then $v'\in I_i$ (otherwise $u$ will be adjacent to $v'$). Therefore, $z$ is adjacent to $v'$, which is a contradiction. 

    Assume that (3) holds, i.e., $u\in I_i, v\in B'_{i(j-1)\ell}$. Recall that this case is applicable only for $j=8$. This implies that a vertex in $B'_{i7\ell}$ is in $R$, a case handled by the other cases. 

    Now, consider the last case, i.e., $u\in B'_{i(j-1)\ell}, v\in B'_{i(j+1)\ell}$. This implies that
    $R$ has nonempty intersection with each level from 2 to 8, producing a path on 7 vertices, which is a contradiction.

    What is left to prove is the case when $j=2$.
    Let $j=2$.
    Further, by the first part of the proof, we can safely ignore the vertices from $\beta_3\cup\beta_4\cup\ldots\cup\beta_8$ from the rest of the proof. 

    Based on the membership of $u$ and $v$, we obtain the following cases: (1) $u,v\in L\cap A$, (2) $u\in B_{i1\ell}$, $v\in L\cap A$. By Lemma~\ref{lem:c7:bij-atmost1}, it is not possible that $u,v\in B'_{i1\ell}$.

    Assume that (1) holds, i.e., $u,v\in L\cap A$.
    Here, by following the proof corresponding to the similar case when $j\geq 3$, we obtain a contradiction. 

    Now, assume that (2) holds, i.e., $u\in B'_{i1\ell}$ and $v\in L\cap A$.
    If both $u'$ and $v'$ are in $L$, then there is an induced $K_2+K_1$ in the graph induced by $L$, which is a contradiction - recall that $L$ 
    induces a complete bipartite graph in $G'$ and a graph is complete bipartite if and only if it does not have an induced $K_2+K_1$. Since $L\cap A$ 
    is complete to $\beta_2$, we obtain that $u'\notin \beta_2$. Since $u$ must be nonadjacent to $v'$,
    we obtain that $v'\notin I$. Since $u'$ must be 
    nonadjacent to both $v$ and $z$,
    we obtain that $u'\notin L$. 
    Since $\beta_1$ is complete to $I$, and $u'$ and
    $v'$ are nonadjacent, we cannot have $u'\in I$
    and $v'\in \beta_1$.
    Therefore, the cases to be considered are: (a) $u', v'\in \beta_1$,
    (b) $u'\in \beta_1$, $v'\in L$, (c) $u'\in I$, $v'\in L$, (d) $u'\in I$, $v'\in \beta_2$, (e) $u'\in \beta_1$, $v'\in \beta_2$. The rest of the 
    proof obtains contradictions in each of these cases.

    Assume that $u',v'\in \beta_1$. Since $I$ is 
    complete to $\beta_1$, we obtain that $R\cap I=\emptyset$. 
    If $R\cap \beta_2$ has a vertex other than $z$,
    then the vertex $v$ gets a degree 3 in the graph
    induced by $R$, which is a contradiction.
    Therefore, neither $u''$ nor $v''$ is from
    $\beta_2$. If either $u''$ or $v''$
    is from $\beta_1$, then $R\cap \beta_1$
    induces a graph having $P_3+K_1$ or $2K_2$
    as an induced subgraph, which is not possible due to the construction. Therefore, neither $u''$ nor $v''$ is from $\beta_1$. Then
    both $u''$ and $v''$ must be
    from $L$. 
    Then $R\cap L$ induces a $K_2+K_1$, 
    which is a contradiction.
    
    Assume that $u'\in \beta_1$, $v'\in L$. Then $R\cap I=\emptyset$ (otherwise there is a triangle formed by $\{u,u'\}$ and a vertex in $R\cap I$). If $v''\in L$, then $v''\in A$ and is adjacent to $z$, and then there is a $C_4$ formed by
    $\{z,v,v',v''\}$, which is a contradiction.
    If $u''\in L$, then 
    the graph induced by $L$ has a $K_2+K_1$, which is a contradiction, as $L$ induces a complete bipartite graph. Therefore, both $u''$ and $v''$ are in $\beta_1$
    and the graph induced by $\{u,u',u'',v''\}$
    forms a $P_4$, which is a contradiction due to Lemma~\ref{lem:c7:bij-atmost1}.
    
    Assume that $u'\in I$, $v'\in L$. 
    If $v''\in L$, then there is a $C_4$ formed by
    $\{z,v,v',v''\}$. If $v''\in \beta_1$,
    then there is a $C_6$ formed by $\{z,u,u',v'',v',v\}$.
    It is not possible that $v''\in \beta_2$, as then $v'\in A$, which is a contradiction. If $v''\in I$,
    then there is a $C_5$ formed by $\{z,v,v',v'',u\}$.
    Therefore, $v''$
    cannot be from $L\cup\beta_1\cup\beta_2\cup I$, which is a contradiction.
    
    Assume that $u'\in I$, $v'\in \beta_2$. 
    If $v''\in L$, then $\{z,v,v',v''\}$
    forms a $C_4$. If $v''\in \beta_1$, then
    $\{z,u,u',v'',v',v\}$ forms a $C_6$.
    If $v''\in \beta_2$, then $\{v,v',v''\}$
    forms a triangle. 
    If it not possible that $v''\in I$ as then $v'\in \beta_2$
    is not adjacent to $v''$.
    Therefore, $v''\notin L\cup \beta_1\cup\beta_2\cup I$, which is a contradiction.
    
    Assume that $u'\in \beta_1$, $v'\in \beta_2$.
    If $v''\in L$, then $v''\in A$ and there is a 
    $C_4$ formed by $\{z,v,v',v''\}$. If $v''\in \beta_2$, then there is a triangle $\{v,v',v''\}$.
    If $v''\in I$, then there is a triangle, $\{u,u',v''\}$. Therefore, $v''\in \beta_1$.
    If $u''\in L$, then there is a $C_4$ formed by
    $\{z,v,v',u''\}$ (if $u''\in A$), 
or there is a $C_5$ formed by $\{z,u,u',u'',v\}$ (if $u''\notin A$). If $u''\in \beta_1$, then there is a $P_4$ induced by $R\cap \beta_1$, which is not possible due to the construction and Lemma~\ref{lem:c7:bij-atmost1}. If $u''\in \beta_2$,
    then the vertex $v$ gets a degree 3 in the 
    graph induced by $R$. If $u''\in I$,
    then there is a triangle formed by $\{u,u',u''\}$.
    Therefore, $u''\notin L\cup\beta_1\cup\beta_2\cup I$, which is a contradiction.
\end{proof}

 Now, we are ready to prove the forward direction of the reduction.
 
\begin{lemma}
\label{lem:c7:forward}
   Let $\Phi$ be a yes-instance of \TSATM, and $\psi$ be a truth assignment satisfying $\Phi$. 
   Let $A$ be the set of variable vertices whose corresponding variables were assigned \TRUE\ by $\psi$. Let $G'$ be $S(G_{\Phi}, A)$.
   Then $G'$ is $C_7$-free.
\end{lemma}
\begin{proof}
    For a contradiction, assume that $R$ induces a $C_7$ in $G'$.
    By Lemma~\ref{lem:c7:2-8}, we can ignore all vertices from levels 2 to 8. What remains in the graph is $L\cup \beta_1\cup I$. 
    Clearly, $\beta_1$ induces a $C_7$-free graph.
    Similarly, $I$ induces a $C_7$-free graph.
    Then by Observation~\ref{obs:module:prime}, we 
    obtain that $R\setminus (\beta_1\cup I)\neq \emptyset$, and therefore, $R\cap L\neq \emptyset$. By Observation~\ref{obs:c7:l}, we have that $R\cap (\beta_1\cup I)\neq \emptyset$. 

    We recall that $R\cap L$ induces a complete bipartite graph. Therefore, if $R\cap L$ has at least 4 vertices, then it forms either an
    independent set of 4 vertices, or a claw, or a $C_4$. Therefore, $|R\cap L|\leq 3$.
    Since $\beta_1$ is complete to $I$, with 
    similar arguments we obtain that, if both $R\cap I$ and $R\cap \beta_1$
    are nonempty, then $|R\cap (I\cup \beta_1)|\leq 3$. In that case, $|R\cap L|\geq 4$, which is a contradiction. 
    Therefore, either $R\cap I = \emptyset$ or $R\cap \beta_1=\emptyset$. Assume that $R\cap I=\emptyset$. 
    By Lemma~\ref{lem:c7:bij-atmost1}, $|R\cap B_{i1}|\leq 4$. This, along with the fact that  $B_{i1}$ is complete to $B_{i'1}$ (for $i'\neq i$), implies that $|R\cap \beta_1|\leq 4$. Further, if $|R\cap \beta_1|=4$,
    then $R\cap \beta_1\subseteq B_{i1}$, for some $1\leq i\leq m$.
    Therefore, $|R\cap L|=3$, $|R\cap \beta_1|=4$,
    $R\cap \beta_1\subseteq B_{i1}$, and $R$ has exactly one vertex each
    from $B_{i1\ell}$, for $1\leq \ell\leq 4$. Then, 
    either $R\cap L$ induces an edgeless graph of
    3 vertices, or a $P_3$. If it induces an edgeless graph of 3 vertices, we obtain that $R\cap L\subseteq L_i$ (a vertex in $L\setminus L_i$ is either complete to $B_{i1}$ or nonadjacent to $B_{i1}$ based on whether it belongs to $A$ or not). This gives us a contradiction as 
    at least one vertex of $L_i$ is in $A$ and at least one vertex of $L_i$ is not in $A$. If $R\cap L$ induces a $P_3$, then by a simple degree counting, $R$ cannot induce a $C_7$. The case when 
    $R\cap \beta_1=\emptyset$ can be handled in a similar way.
    This completes the proof. 
\end{proof}

Lemmas~\ref{lem:c7:backward}, \ref{lem:c7:forward}, 
and Proposition~\ref{pro:ksatm}, and the fact that the number of vertices in $G_{\Phi}$ is linear in the number of variables and clauses in $\Phi$ 
imply Theorem~\ref{thm:c7}.

\begin{theorem}
\label{thm:c7}
\SWTF{C_7} is \NPC\ and cannot be solved in 
$2^{o(n)}$-time, assuming ETH, where $n$ is the number of vertices in the input graph.
\end{theorem}
 
We are unable to use these reductions to prove that \SWTF{P_\ell}, for some $\ell<10$ or \SWTF{C_\ell}, for some $\ell<7$ is \NPC\ - the forward direction of the reduction fails in such cases. 
We defer generalizations of these results to $P_t$-free graph (for every $t\geq 10$) and to $C_t$-free graphs (for every $t\geq 7$), to a future version of this paper. 

\section{Concluding remarks}

There are many interesting questions one can ask about the characterization and computation of lower and upper switching classes of various graph classes. Here we list a few of them. 

Since recognizing \UP{\mathcal{F}(P_{10})} 
and recognizing \UP{\mathcal{F}(C_7)} are \NPC, by Proposition~\ref{pro:switching classes}(\ref{pro:switching classes:item:c}), we obtain that
recognizing \LW{\mathcal{G}} is \NPC, where $\mathcal{G}$ is the class of graphs containing an induced $P_{10}$ or the class of graphs containing an induced $C_7$. Note that these classes are non-hereditary. 
For a hereditary graph class $\mathcal{G}$, is it true that whenever $\mathcal{G}$ is recognizable in polynomial-time, lower $\mathcal{G}$ switching class is also recognizable in polynomial-time? 
We know by Proposition~\ref{pro:lower-h-free} that this is true whenever $\mathcal{G}$
is characterized by a finite set of forbidden induced
subgraphs. 

Is it true that recognizing upper $H$-free switching class is polynomially equivalent to recognizing the upper $H'$-free switching class, where $H$ and $H'$ are switching equivalent? We know that the answer to the corresponding question for lower switching class is trivial, as both lower $H$-free and lower $H'$-free switching classes can be recognized in polynomial-time.
In particular, can we recognize the upper $H$-free switching class in polynomial time when $H$ is $C_{4}$, $K_{4}$, or diamond?  For each of them, we know a switching equivalent $H'$ such that the upper $H'$-free switching class can be recognized in polynomial time.

Let $\mathcal{G}$ be a graph class. Assume that, for any graph $G$, there are only polynomial number of ways to switch $G$ to a graph in $\mathcal{G}$. Then every large enough graph $G$ can be switched to a graph not in $\mathcal{G}$. Therefore, \LW{\mathcal{G}} is finite. Is it true that whenever \LW{\mathcal{G}} is finite, then \UP{\mathcal{G}} can be recognized in polynomial-time?

What is the smallest integer $\ell$ such that the recognition of \UP{\mathcal{F}(P_\ell)} is \NPC? We know that $5\leq\ell\leq 10$.
Similarly, what is the smallest integer $\ell$ such that the recognition of \UP{\mathcal{F}(C_\ell)} is \NPC? We know that $4\leq\ell\leq 7$.

\appendix
\section*{Appendix: Omitted table and figures}

\begin{table}[h!] 
  \caption{The switching equivalents of some simple graphs.
  Vertices in $G$ are numbered consecutively.  Each column is the set $A$, while the omitted sets are either not applicable or symmetric to one of the given ones.}
  \label{tbl:switching-equivalents}
  \begin{center}
    \small
    \begin{tabular}{ l | c c | c c  c | c  c c} 
      \toprule
      $G$ & \{1\} & \{2\} & \{1,2 \} & \{1, 3\} & \{1, 4\}&       \{1, 2, 3\}&       \{1, 3, 5\}&       \{1, 2, 4\}
      \\ \midrule
      $P_{4}$ & paw & $P_{3} + K_{1}$ & diamond & $P_{2} + 2 K_{1}$ & $P_{4}$
      \\
      \midrule
      $C_{4}$ & claw & & $C_{4}$ & $I_{4}$ &
      \\
      \midrule
      $C_{5}$ & bull & & gem & $P_{4} + K_{1}$ & $P_{4}$
      \\
      \midrule
      $C_{6}$ & (1,1,2,1,1) & & (1,2,2,1) & (2,1,2,0,1) & $C_{6}$ & (2,2,2) & $3 K_{2}$ & (1,1,2,1,1)
      \\
      \bottomrule
    \end{tabular}
  \end{center}
\end{table}

\begin{figure}[h]
\tikzstyle{filled vertex}  = [{circle,draw=blue,fill=black!50,inner sep=1pt}]  \tikzstyle{empty vertex}  = [{circle, draw, fill = white, inner sep=1.pt}]
  \centering \small
  
\begin{tikzpicture}[scale=.5] \def\n{7}
      \def\radius{2}      
      \coordinate (v0) at ({90 + 180 / \n}:\radius) {};
      \foreach \i in {1,..., \n} {
        \pgfmathsetmacro{\angle}{90 - (\i - .5) * (360 / \n)}
        \node[empty vertex] (v\i) at (\angle:\radius) {};

   }

      \foreach \i in {1,..., \n}
        \draw let \n1 = {int(\i - 1)} in (v\n1) -- (v\i);

    \end{tikzpicture}
    \;
\begin{tikzpicture}[scale=.5] \def\n{6}
      \def\radius{2}      
      \coordinate (v0) at ({90 + 180 / \n}:\radius) {};
      \foreach \i in {1,..., \n} {
        \pgfmathsetmacro{\angle}{90 - (\i - .5) * (360 / \n)}
        \node[empty vertex] (v\i) at (\angle:\radius) {};

   }
\node[filled vertex] (c) at (0, 2.5) {};

        \foreach \i in {3,4} 
        \draw (c) -- (v\i);
      \foreach \i in {2,..., \n}
        \draw let \n1 = {int(\i - 1)} in (v\n1) -- (v\i);
     \foreach \i in {2}
        \draw let \n1 = {int(\i - 1)} in (v\n1) -- (v\i);
\foreach \i in {2, 5} 
        \node[filled vertex] at (v\i) {};

    \end{tikzpicture}
    \;
\begin{tikzpicture}[scale=.5] \def\n{6}
      \def\radius{2}      
      \coordinate (v0) at ({90 + 180 / \n}:\radius) {};
      \foreach \i in {1,..., \n} {
        \pgfmathsetmacro{\angle}{90 - (\i - .5) * (360 / \n)}
        \node[empty vertex] (v\i) at (\angle:\radius) {};

   }
 \node[empty vertex] (c) at (0, 2.5) {};

      \foreach \i in {1,..., \n}
        \draw let \n1 = {int(\i - 1)} in (v\n1) -- (v\i);

        \foreach \i in {5}
        \draw let \n1 = {int(\i - 3)} in (v\n1) -- (v\i);

         \foreach \i in {1}
        \draw let \n1 = {int(\i + 4)} in (v\n1) -- (v\i);
        \foreach \i in {2}
        \draw let \n1 = {int(\i + 2)} in (v\n1) -- (v\i);
\foreach \i in {1, 4} 
        \node[filled vertex] at (v\i) {};
        
    \end{tikzpicture}
    \;
\begin{tikzpicture}[scale=.5] \def\n{5}
      \def\radius{2}      
      \coordinate (v0) at ({90 + 180 / \n}:\radius) {};
      \foreach \i in {1,..., \n} {
        \pgfmathsetmacro{\angle}{90 - (\i - .5) * (360 / \n)}
        \node[empty vertex] (v\i) at (\angle:-\radius) {};

   }

   \node[empty vertex] (c) at (1.5,2) {};
   \node[empty vertex] (d) at (0,1.5) {};

     \foreach \i in {5,4,3,2}
        \draw let \n1 = {int(\i -1)} in (v\n1) -- (v\i);
    \foreach \i in {5}
        \draw let \n1 = {int(\i - 4)} in (v\n1) -- (v\i);

        \draw  (d) -- (v3);
        \draw  (d) -- (v4);
        \draw  (d) -- (v2);
        \draw  (d) -- (v2);
        \draw  (c) -- (v3);
        \draw  (v2) -- (v4);
        \draw  (v4) -- (v1);
        
      \foreach \i in {3, 5,2} 
        \node[filled vertex] at (v\i) {};
        
    \end{tikzpicture}\;

\begin{tikzpicture}[scale=.5] \def\n{6}
      \def\radius{2}      
      \coordinate (v0) at ({90 + 180 / \n}:\radius) {};
      \foreach \i in {1,..., \n} {
        \pgfmathsetmacro{\angle}{90 - (\i - .5) * (360 / \n)}
        \node[empty vertex] (v\i) at (\angle:\radius) {};

   }
      
      \node[empty vertex] (c) at (0, 0) {};
      \foreach \i in {1, 6} 
        \draw (c) -- (v\i);
      \foreach \i in {1,..., \n}
        \draw let \n1 = {int(\i - 1)} in (v\n1) -- (v\i);
       \foreach \i in {1}
        \draw let \n1 = {int(\i + 2)} in (v\n1) -- (v\i);
      \foreach \i in {6, 3} 
        \node[filled vertex] at (v\i) {};
        
    \end{tikzpicture}
    \;
  \vspace{0.25cm}  
\begin{tikzpicture}[scale=.5] \def\n{6}
      \def\radius{2}      
      \coordinate (v0) at ({90 + 180 / \n}:\radius) {};
      \foreach \i in {1,..., \n} {
        \pgfmathsetmacro{\angle}{90 - (\i - .5) * (360 / \n)}
        \node[empty vertex] (v\i) at (\angle:\radius) {};

   }
 \node[filled vertex] (c) at (0, 2.5) {};
      \foreach \i in {2,3,4,5} 
        \draw (c) -- (v\i);
 
      \foreach \i in {2,..., \n}
        \draw let \n1 = {int(\i - 1)} in (v\n1) -- (v\i);

        \foreach \i in {2}
        \draw let \n1 = {int(\i - 1)} in (v\n1) -- (v\i);

         \foreach \i in {3}
        \draw let \n1 = {int(\i + 1)} in (v\n1) -- (v\i);

    \end{tikzpicture}
    \;
\begin{tikzpicture}[scale=.5] \def\n{3}
      \def\radius{2.5}      
      \coordinate (v0) at ({90 + 180 / \n}:\radius) {};
      \foreach \i in {1,..., \n} {
        \pgfmathsetmacro{\angle}{90 - (\i - .5) * (360 / \n)}
        \node[empty vertex] (v\i) at (\angle:-\radius) {};

   }
      \node[filled vertex] (c) at  (0,-0.75){};
      
      \node[empty vertex] (d) at (0.5, 0.5) {};
       \node[filled vertex] (e) at (-0.5, 0.5) {};
        \node[empty vertex] (f) at (-0.7, -0.25) {};
   
         \draw (c) -- (d);
      \foreach \i in {1,2,3} 
        \draw (c) -- (v\i);
      \foreach \i in {3,2}
        \draw let \n1 = {int(\i - 1)} in (v\n1) -- (v\i);
       \foreach \i in {3}
        \draw let \n1 = {int(\i - 2)} in (v\n1) -- (v\i);

       \draw (c) -- (d);
       \draw (c) -- (e);
       \draw (c) -- (f);
       \draw (d) -- (v2);
       \draw (d) -- (v3);
       \draw (e) -- (v1);
       \draw (e) -- (v2);
       \draw (e) -- (f);
       \draw (f) -- (v1);
        
      \foreach \i in {3} 
        \node[filled vertex] at (v\i) {};
        
    \end{tikzpicture}
    \;
\begin{tikzpicture}[scale=.5] \def\n{5}
      \def\radius{2}      
      \coordinate (v0) at ({90 + 180 / \n}:\radius) {};
      \foreach \i in {1,..., \n} {
        \pgfmathsetmacro{\angle}{90 - (\i - .5) * (360 / \n)}
        \node[empty vertex] (v\i) at (\angle:-\radius) {};

   }
      
      \node[empty vertex] (c) at (1,0) {};
      \node[filled vertex] (d) at (-1,0) {};

      \foreach \i in {2,3,4,5}
        \draw let \n1 = {int(\i - 1)} in (v\n1) -- (v\i);

         \foreach \i in {5}
        \draw let \n1 = {int(\i - 4)} in (v\n1) -- (v\i);
       \draw (c) -- (d);
       \draw (c) -- (v4);
       \draw (c) -- (v5);
       \draw (d) -- (v2);
       \draw (d) -- (v1);
       \draw (d) -- (v5);

\foreach \i in {5, 3} 
        \node[filled vertex] at (v\i) {};
        
    \end{tikzpicture}
    \;
\begin{tikzpicture}[scale=.5] \def\n{3}
      \def\radius{2.5}      
      \coordinate (v0) at ({90 + 180 / \n}:\radius) {};
      \foreach \i in {1,..., \n} {
        \pgfmathsetmacro{\angle}{90 - (\i - .5) * (360 / \n)}
        \node[empty vertex] (v\i) at (\angle:-\radius) {};

   }
      \node[filled vertex] (c) at  (0,-0.75){};
      
      \node[empty vertex] (d) at (0.5, 0.5) {};
       \node[empty vertex] (e) at (-0.5, 0.5) {};
        \node[empty vertex] (f) at (-0.7, -0.25) {};
   
         \draw (c) -- (d);
      \foreach \i in {1,2,3} 
        \draw (c) -- (v\i);
      \foreach \i in {3,2}
        \draw let \n1 = {int(\i - 1)} in (v\n1) -- (v\i);
       \foreach \i in {3}
        \draw let \n1 = {int(\i - 2)} in (v\n1) -- (v\i);

       \draw (c) -- (d);
       \draw (c) -- (e);
\draw (d) -- (v3);
       \draw (e) -- (v1);
       \draw (e) -- (v2);
       \draw (e) -- (f);
       \draw (f) -- (v1);
        
      \foreach \i in {1} 
        \node[filled vertex] at (v\i) {};
        
    \end{tikzpicture}
    \;

  \caption{Switching equivalent graphs of $C_7$.  The set $A$ consists of the solid nodes. These are also graphs part of lower planar switching class having maximum number of vertices.}
  \label{fig:switching equivalent C7}
\end{figure}

 
\bibliographystyle{plain}
\bibliography{main}
\appendix
\end{document}